%% file: spenner20260700icalp.tex
\documentclass[a4paper,UKenglish,cleveref, autoref, thm-restate, numberwithinsect]{lipics-v2021} % DAS: numberwithinsect added
\title{Deciding DFA-Primality is \fontComplexityClass{NP}-Hard} % DAS: Title added

\author{Daniel Alexander Spenner}{Technische Universität Dortmund, Germany}{daniel.spenner@tu-dortmund.de}{https://orcid.org/0009-0001-2784-5914}{}
\authorrunning{D. A. Spenner}
\Copyright{Daniel Alexander Spenner}
\ccsdesc[100]{Theory of computation~Regular languages}
\ccsdesc[100]{Theory of computation~Problems, reductions and completeness} % DAS: Added ACM 2012 classifications from https://dl.acm.org/ccs/ccs_flat.cfm

\keywords{Deterministic finite automaton (DFA), Regular languages, Finite languages, Decomposition, Primality, NP-Hardness} % DAS: Added comma-separated list of keywords

\category{} %optional, e.g. invited paper

\nolinenumbers %uncomment to disable line numbering

\EventEditors{Sayan Bhattacharya, Danupon Nanongkai, Michael Benedikt, and Gabriele Puppis}
\EventNoEds{4}
\EventLongTitle{53rd International Colloquium on Automata, Languages, and Programming (ICALP 2026)}
\EventShortTitle{ICALP 2026}
\EventAcronym{ICALP}
\EventYear{2026}
\EventDate{July 7--10, 2026}
\EventLocation{Royal Holloway, University of London, Egham, United Kingdom}
\EventLogo{}
\SeriesVolume{374}
\ArticleNo{173}
\input{spenner20260700icalpCommands}
\input{spenner20260700icalpPackages}
\input{spenner20260700icalpTheoremlike}

\begin{document}

\maketitle

% DAS: Abstract added
\begin{abstract}
	A DFA $\mathcal{A}$ is \highlightDef{composite} if there exist DFAs $\mathcal{A}_1,\dots,\mathcal{A}_t$ with $\opLang{\mathcal{A}} = \bigcap_{i=1}^{t} \opLang{\mathcal{A}_i}$ such that each $\mathcal{A}_i$ has strictly less states than the minimal DFA deciding $\opLang{\mathcal{A}}$. Otherwise, it is \highlightDef{prime}.
	
	\probPrimeDFA{} is the problem of deciding primality for a given DFA.
	It was defined by Kupferman and Mosheiff in 2015 and it was shown to be \fontComplexityClass{NL}-hard and in \fontComplexityClass{ExpSpace}.
	
	This paper proves the \fontComplexityClass{NP}-hardness of \probPrimeDFA{}, thereby making the first progress in closing this doubly-exponential gap.
	It proves the \fontComplexityClass{NP}-hardness by a reduction from the propositional logic satisfiability problem.
	The correctness of the reduction relies on an involved characterization of primality for a class of DFAs which contains those that can occur in the reduction.
\end{abstract}

\newpage % DAS: Added so that introduction starts at new page
\input{introduction}
\input{preliminaries}
\input{npHardnessShortV2}
\input{discussion}

%%
%% Bibliography
%%

%% Please use bibtex, 

\bibliography{spenner20260700icalp}

% !!!!!!!!!! VERSIONSTART !!!!!!!!!!
% !!!!!!!!!! FULL VERSION !!!!!!!!!!
\newpage % DAS: Added so that appendix starts at new page
\appendix
\input{npHardness}
% !!!!!!!!!! CONF VERSION !!!!!!!!!!
% NOTHING
% !!!!!!!!!! VERSIONEND  !!!!!!!!!!

\end{document}

%% file: spenner20260700icalpCommands.tex
\newcommand{\symADFA}{\text{ADFA}} %alternative: {\text{ADFA}} % ADFA symbol
\newcommand{\symADFAp}{\text{ADFA}^+} %alternative: {\text{ADFA}^{+}} % pADFA symbol
\newcommand{\symCNFDFA}{CNF-DFA}
\newcommand{\symCVar}{\mu} % Symbol for the variable used to store the number of c's at the end of a max-visiting word.
\newcommand{\symElem}{e}
\newcommand{\symHash}{\#} % hash symbol
\newcommand{\symInd}{\text{ind}} % Ind symbol
\newcommand{\symLinLenVar}{\lambda} % Symbol for the variable used to store the "length" of an ADFA^+.
\newcommand{\symNatNum}{\mathbb{N}} % Natural numbers symbol
\newcommand{\symNatNumGeq}[1]{\mathbb{N}_{\geq #1}} % Natural-numbers-greater-than operator
\newcommand{\symPPLong}{max-pumping-property} % Pumping-property symbol
\newcommand{\symPPShort}{mp-prop} % Symbol for the abbreviation of the pumping property
\newcommand{\symRun}{\chi} % Symbol used to represent runs

\newcommand{\opBigO}[1]{\mathcal{O}(#1)} % Big-O oprator
\newcommand{\opDecompSet}[1]{\alpha(#1)} % Set of DFAs usable in a decomposition operator
\newcommand{\opDiffWords}[1]{\mathcal{K}(#1)} % Difficult words operator
\newcommand{\opInd}[1]{\symInd{}(#1)} % Ind operator
\newcommand{\opLang}[1]{\mathcal{L}(#1)} % Language operator
\newcommand{\opLinLen}[1]{\text{len}(#1)} % Linear length operator (length of the longest word on which an ADFA does not enter a sink)
\newcommand{\opLitToBin}[1]{h(#1)}
\newcommand{\opPump}[2]{{#1}[#2]} % Pumping operator, #1 gets pumped according to #2, where #2 should be of the form i,j;l.
\newcommand{\opSize}[1]{|#1|} % Size operator
\newcommand{\highlightDef}[1]{\textit{#1}} % Definition highlighting
\newcommand{\fontComplexityClass}[1]{\textsc{#1}} % Complexity class font
\newcommand{\fontProblem}[1]{\textsc{#1}} % Problem font

\newcommand{\probCNFSAT}{\fontProblem{CNFSAT}}

\newcommand{\probPrimeADFApMLS}[1]{\fontProblem{#1Prime-}\text{S}\symADFAp}
\newcommand{\probPrimeDFA}[1]{\fontProblem{#1Prime-DFA}}
\newcommand{\probPrimeDFAfin}[1]{\fontProblem{#1Prime-DFA}_\text{fin}}

\newcommand{\probSPrimeDFA}[1]{\fontProblem{#1S-Prime-DFA}}

%% file: spenner20260700icalpPackages.tex
\usepackage{tikz}
\usetikzlibrary{automata, positioning, arrows}
\tikzset{
	node distance=3cm,	% Minimale Distanz zwischen Knoten
	->,	% Gerichtete Kanten
	initial text=$ $, % Setzt initialen Text
}

\crefname{paragraph}{Section}{Sections}

\makeatletter
\newcommand\IfRestateTF{%
  \ifx\label\thmt@gobble@label % or just compared to \@gobble
    \expandafter\@firstoftwo
  \else
    \expandafter\@secondoftwo
  \fi
}
\makeatother
\newcommand{\RestateRemark}{\IfRestateTF{{\normalfont\bfseries (Restated) }}{}}

%% file: spenner20260700icalpTheoremlike.tex
\theoremstyle{definition}
\newtheorem{definitionRoman}[theorem]{Definition}
\theoremstyle{claimstyle}
\newtheorem{claimWithinProof}{Claim}
\AtBeginEnvironment{proof}{\setcounter{claimWithinProof}{0}}

%% file: introduction.tex
\section{Introduction}
\label{sec:introduction}
We consider the primality of deterministic finite automata (DFA). Intuitively, the basic question is: Given a DFA $\mathcal{A}$, are there a number of DFAs strictly smaller than $\mathcal{A}$ itself that combined decide the same language as $\mathcal{A}$?
More formally, a DFA $\mathcal{A}$ is \highlightDef{composite} if there exist DFAs $\mathcal{A}_1,\dots,\mathcal{A}_t$ with $\opLang{\mathcal{A}} = \bigcap_{i=1}^{t} \opLang{\mathcal{A}_i}$ such that the size of every $\mathcal{A}_i$ is strictly smaller than the index of $\mathcal{A}$. Otherwise, $\mathcal{A}$ is \highlightDef{prime}.
Here, the size of $\mathcal{A}$ is the number of its states, while the index of $\mathcal{A}$ is the size of the minimal DFA deciding $\opLang{\mathcal{A}}$.
\probPrimeDFA{} denotes the problem of deciding primality for a given DFA.

The notion of compositionality of finite automata was introduced in \cite{kupferman2015prime},
although a similar but more restricted notion was already studied in \cite{gazi2008assisted}.
With \cite[Theorems 2.4 and 2.5]{kupferman2015prime}, \probPrimeDFA{} is \fontComplexityClass{NL}-hard and in \fontComplexityClass{ExpSpace}.
Up to now, progress in closing this surprisingly large doubly-exponential gap has proven elusive.
In this paper, we improve the lower complexity bound of \probPrimeDFA{}.

Compositionality in general is an important notion in both practical and theoretical computer science \cite{roever1998compositionality,tripakis2016compositionality}.
Regular languages and DFAs are key concepts of theoretical computer science and the question whether a DFA can be decomposed in this fashion seems natural and worth studying.
Furthermore, the general idea of
automaton decomposition
can be motivated by LTL model checking, while a question related to
automaton decomposition
arises in the field of automaton identification. Both will be briefly discussed below.

\subparagraph*{Contributions}
We prove the \fontComplexityClass{NP}-hardness of \probPrimeDFA{} by reducing the propositional logic satisfiability problem to \probPrimeDFA{}.

Given a formula $\Phi$, the reduction yields a DFA $\mathcal{A}_\Phi$ so that every assignment for $\Phi$ can be associated with a word of a specific form which is accepted by $\mathcal{A}_\Phi$ if and only if the underlying assignment satisfies $\Phi$.
The proof that this is actually a reduction to \probPrimeDFA{} relies on an involved characterization of the compositionality of a class of DFAs which contains all DFAs of the form $\mathcal{A}_\Phi$.
Namely, we characterize the compositionality of "almost-acyclic" linear safety DFAs
(essentially, DFAs whose states are linearly ordered, whose only rejecting state is a rejecting sink, who possess an accepting sink, and whose only cycles are its sinks).
This characterization then shows that $\mathcal{A}_\Phi$ accepts a word of this specific form if and only if it is prime. This proves the correctness of the reduction.

Thus, the \fontComplexityClass{NP}-hardness proof entails the construction of $\mathcal{A}_\Phi$, the characterization of the compositionality of the mentioned DFA class, and the use of this characterization to link the satisfiability of $\Phi$ and the compositionality of $\mathcal{A}_\Phi$.
Additionally, we further exploit the characterization of the compositionality of "almost-acyclic" linear safety DFAs to prove that the restriction of \probPrimeDFA{} to such DFAs is \fontComplexityClass{NP}-complete, meaning that the established lower bound is tight for this specific class of DFAs.

We attempt to give comprehensive proof sketches for these results in \cref{sec:npHardnessShortV2}.
% !!!!!!!!!! VERSIONSTART !!!!!!!!!!
% !!!!!!!!!! FULL VERSION !!!!!!!!!!
The complete formal proofs can be found in \cref{sec:npHardness}.
% !!!!!!!!!! CONF VERSION !!!!!!!!!!
%The full version of this paper contains an additional section in its appendix \cite[Appendix A]{spenner2026deciding}. There, we present the full argument with complete formal proofs.
% !!!!!!!!!! VERSIONEND  !!!!!!!!!!

\subparagraph*{Related Work}
As already mentioned, the notion of compositionality was introduced in \cite{kupferman2015prime}, where the mentioned complexity bounds for \probPrimeDFA{} were established. This initial work was followed up by \cite{jecker2020unary,jecker2021decomposing,spenner2023decomposing}. These works focused on restricted classes of DFAs, namely, unary DFAs, permutation DFAs, and acyclic DFAs and thereby finite languages, respectively. However, the initial complexity bounds for \probPrimeDFA{} remained unaltered.

The general idea of decomposing finite automata can be motivated by LTL model checking, where the validity of a specification, given as an LTL-formula, is checked for a system.
The automata-based approach entails translating the specification into a finite automaton \cite{vardi1986automata}.
Since the LTL model checking problem is \fontComplexityClass{PSpace}-complete in the size of the LTL-formula \cite[Theorem 5.48]{baier2008principles}, it is desirable to decompose the formula into a conjunction of subformulas. This can also be understood as decomposing the finite automaton corresponding to the formula.

A question related to the decomposition of automata arises in the field of automaton identification. The basic task here is, given a set of labeled words, to construct an automaton conforming to this set \cite{gold1978complexity}. An interesting approach is to construct multiple automata instead of one, which can lead to smaller and more intuitive solutions \cite{lauffer2022learning}.

However, the results obtained in this paper use decompositions with DFAs that are only slightly smaller than the given DFA (most often by just one state). For practical applications, decompositions into much smaller components might be more relevant.

%% file: preliminaries.tex
\section{Preliminaries}
\label{sec:preliminaries}

A \highlightDef{deterministic finite automaton} (DFA) is a $5$-tuple $\mathcal{A} = (Q,\Sigma,q_I,\delta,F)$ where $Q$ is a finite set of states, $\Sigma$ is a finite non-empty alphabet, $q_I \in Q$ is an initial state, $\delta: Q \times \Sigma \rightarrow Q$ is a transition function, and $F \subseteq Q$ is a set of accepting states. As usual, we extend $\delta$ to words: $\delta: Q \times \Sigma^* \rightarrow Q$ with $\delta(q,\varepsilon) = q$ and $\delta(q,\sigma_1\dots\sigma_n) = \delta(\delta(q,\sigma_1\dots\sigma_{n-1}),\sigma_n)$.
% Commented out because unnecessary. In MA, it was used in Section 5.
%For $q \in Q$, DFA $\mathcal{A}^q$ is constructed out of $\mathcal{A}$ by setting $q$ as the initial state, thus $\mathcal{A}^q = (Q,\Sigma,q,\delta,F)$.

The \highlightDef{run} of $\mathcal{A}$ on a word $w = \sigma_1\dots\sigma_n$ starting in state $q$ is the sequence $q_0,\sigma_1,q_1,\dots,\sigma_n,q_n$ with $q_0 = q$ and $q_i = \delta(q_{i-1},\sigma_i)$ for each $i \in \{1,\dots,n\}$. The \highlightDef{initial run} of $\mathcal{A}$ on $w$ is the run of $\mathcal{A}$ on $w$ starting in $q_I$. 
The run of $\mathcal{A}$ on $w$ starting in $q$ is \highlightDef{accepting} if $q_n \in F$. Otherwise, it is \highlightDef{rejecting}.
The DFA $\mathcal{A}$ \highlightDef{accepts} $w$ if the initial run of $\mathcal{A}$ on $w$ is accepting. Otherwise, it \highlightDef{rejects} $w$.
The \highlightDef{language} $\opLang{\mathcal{A}}$ of $\mathcal{A}$ is the set of words accepted by $\mathcal{A}$. We say that $\mathcal{A}$ \highlightDef{decides} $\opLang{\mathcal{A}}$. A language is \highlightDef{regular} if there exists a DFA deciding it. Since we only consider regular languages, we use the terms language and regular language interchangeably.

The \highlightDef{size} $\opSize{\mathcal{A}}$ of $\mathcal{A}$ is the number of states in $Q$. A DFA $\mathcal{A}$ is \highlightDef{minimal} if $\opLang{\mathcal{A}} \neq \opLang{\mathcal{B}}$ holds for every DFA $\mathcal{B}$ with $\opSize{\mathcal{B}} < \opSize{\mathcal{A}}$. It is well known that, for every regular language $L$, there exists a canonical minimal DFA deciding $L$. The \highlightDef{index} $\opInd{L}$ of $L$ is the size of this canonical minimal DFA. The index of $\mathcal{A}$ is the index of the language decided by $\mathcal{A}$, thus $\opInd{\mathcal{A}} = \opInd{\opLang{\mathcal{A}}}$. Note that $\mathcal{A}$ is minimal if and only if $\opSize{\mathcal{A}} = \opInd{\mathcal{A}}$.

We borrow a few terms from graph theory. Let $q_0,\sigma_1,q_1,\dots,\sigma_n,q_n$ be the run of $\mathcal{A}$ on $w = \sigma_1\dots\sigma_n$ starting in $q_0$. Then $q_0,\dots,q_n$ is a \highlightDef{path} in $\mathcal{A}$ from $q_0$ to $q_n$.
% START Version with length
%The \highlightDef{length} of this path is $n$. Thus, for two states $q,q'$, there exists a path from $q$ to $q'$ in $\mathcal{A}$ of length $n$ if and only if there exists a $w \in \Sigma^n$ with $\delta(q,w) = q'$. State $q'$ is \highlightDef{reachable from} $q$ if there exists a path from $q$ to $q'$. Otherwise, $q'$ is \highlightDef{unreachable from} $q$. Obviously, if $q'$ is reachable from $q$, then there exists a path from $q$ to $q'$ of a length strictly smaller than $\opSize{\mathcal{A}}$.
% END Version withlength
% START Version without length
Thus, for two states $q,q'$, there exists a path from $q$ to $q'$ in $\mathcal{A}$ if and only if there exists a $w \in \Sigma^*$ with $\delta(q,w) = q'$. State $q'$ is \highlightDef{reachable from} $q$ if there exists a path from $q$ to $q'$. Otherwise, $q'$ is \highlightDef{unreachable from} $q$.
% END Version without length
We say that $q'$ is \highlightDef{reachable} if it is reachable from $q_I$. Otherwise, it is \highlightDef{unreachable}. A \highlightDef{cycle} in $\mathcal{A}$ is a path $q_0,\dots,q_n$ in $\mathcal{A}$ where $q_0 = q_n$ and $n \in \symNatNumGeq{1}$.
We call a state $q$ a \highlightDef{sink} if $\delta(q,\sigma) = q$ for all $\sigma$, that is,
%if the only out-$q$-transition is a self-loop.
if all out-$q$-transitions are self-loops.
%
% NOTE: The definition of acyclicity, linearity and safety were removed from here and placed instead in npHardnessShortV2.

We introduce the notions of compositionality and primality of DFAs and languages, following the definitions in \cite{kupferman2015prime}:
\begin{restatable}{definitionRoman}{defcompositionality}
	\label{def:compositionality}
	For $k \in \symNatNumGeq{1}$, a DFA $\mathcal{A}$ is \highlightDef{$k$-decomposable} if there exist DFAs $\mathcal{A}_1,\dots,\mathcal{A}_t$ with $\opLang{\mathcal{A}} = \bigcap_{i=1}^{t} \opLang{\mathcal{A}_i}$ and $\opSize{\mathcal{A}_i} \leq k$ for each $i \in \{1,\dots,t\}$, where $t \in \symNatNumGeq{1}$. We call such DFAs $\mathcal{A}_1,\dots,\mathcal{A}_t$ a \highlightDef{$k$-decomposition} of $\mathcal{A}$.
	We call $\mathcal{A}$ \highlightDef{composite} if $\mathcal{A}$ is $k$-decomposable for a $k < \opInd{\mathcal{A}}$, that is, if it is $(\opInd{\mathcal{A}}-1)$-decomposable. Otherwise, we call $\mathcal{A}$ \highlightDef{prime}.\lipicsEnd
\end{restatable}

When analyzing the compositionality of a given DFA $\mathcal{A}$, it is sufficient to consider minimal DFAs $\mathcal{B}$ strictly smaller than the minimal DFA of $\mathcal{A}$ with $\opLang{\mathcal{A}} \subseteq \opLang{\mathcal{B}}$. Thus, we define $\opDecompSet{\mathcal{A}} = \{\mathcal{B} \mid \mathcal{B} \text{ is a minimal DFA with } \opInd{\mathcal{B}}<\opInd{\mathcal{A}} \text{ and } \opLang{\mathcal{A}} \subseteq \opLang{\mathcal{B}}\}$. Obviously, the DFA $\mathcal{A}$ is composite if and only if $\opLang{\mathcal{A}} = \bigcap_{\mathcal{B} \in \opDecompSet{\mathcal{A}}} \opLang{\mathcal{B}}$. We call a word $w \in (\bigcap_{\mathcal{B} \in \opDecompSet{\mathcal{A}}} \opLang{\mathcal{B}}) \setminus \opLang{\mathcal{A}}$ a \highlightDef{primality witness} of $\mathcal{A}$. Clearly, a DFA $\mathcal{A}$ is composite if and only if it has no primality witness.

%We extend the notions of $k$-decompositions, compositionality, primality and primality witnesses to regular languages by identifying a regular language with its minimal DFA.
We say that a regular language is prime (composite) if its minimal automaton is prime (composite).

We denote the problem of deciding primality for a given DFA by \probPrimeDFA{}.
\probPrimeDFA{} is \fontComplexityClass{NL}-hard and in \fontComplexityClass{ExpSpace} \cite[Theorems 2.4 and 2.5]{kupferman2015prime}.

%% file: npHardnessShortV2.tex
\section{\fontComplexityClass{NP}-Hardness of \probPrimeDFA{}}
\label{sec:npHardnessShortV2}
We improve the lower complexity bound for \probPrimeDFA{} by proving:

\begin{restatable}{theorem}{theprimeDFANPHard}
	\label{the:primeDFANPHard}\RestateRemark
	The problem \probPrimeDFA{} is \fontComplexityClass{NP}-hard.\lipicsEnd
\end{restatable}

As mentioned, this is the first improvement of a complexity bound for \probPrimeDFA{} since \cite{kupferman2015prime} introduced \probPrimeDFA{} and proved that it is \fontComplexityClass{NL}-hard and in \fontComplexityClass{ExpSpace}.

\paragraph*{Proof Idea}
We prove the \fontComplexityClass{NP}-hardness of \probPrimeDFA{} (\cref{the:primeDFANPHard}) by a reduction from \probCNFSAT.
Given a CNF-formula $\Phi$, the reduction yields a DFA $\mathcal{A}_\Phi$
% !!!!!!!!!! VERSIONSTART !!!!!!!!!!
% !!!!!!!!!! FULL VERSION !!!!!!!!!!
(\cref{def:aPhi,fig:aPhiaPhiXTransitions})
% !!!!!!!!!! CONF VERSION !!!!!!!!!!
%(Figure 1, definition in full version \cite[Definition A.1]{spenner2026deciding})
% !!!!!!!!!! VERSIONEND  !!!!!!!!!!
with the goal that $\Phi$ is satisfiable if and only if $\mathcal{A}_\Phi$ is prime (\cref{lem:aPhiCharacterization}).

In the following, we call DFAs of the form $\mathcal{A}_\Phi$ \symCNFDFA{}s.
The correctness of the reduction relies on a (somewhat technical) characterization of the compositionality of \symCNFDFA{}s. Simplified, for a given \symCNFDFA{} $\mathcal{A}_\Phi$, only the words on which $\mathcal{A}_\Phi$ visits each of its non-sinks before entering its rejecting sink are relevant for its compositionality. We call these the max-visiting words of $\mathcal{A}_\Phi$. We say that $\mathcal{A}_\Phi$ has the \symPPLong{} (\symPPShort) if every max-visiting word of $\mathcal{A}_\Phi$ has a factorization $xyz$ so that $\mathcal{A}_\Phi$ rejects $xy^lz$ for every $l \in \symNatNum$ (\cref{def:pP}). It turns out that a \symCNFDFA{} is composite if and only if it has the \symPPShort{} (\cref{lem:aPhiIsMLSADFAp,the:pmADFAMinLinSafeCompositionality}).

The CNF-DFAs are designed so that the structure of the max-visiting words and their possible factorizations are strongly restricted (\cref{lem:aPhi}).
This has two major benefits. Firstly, we will have a lot of control over the behavior of $\mathcal{A}_\Phi$ on each individual max-visiting word. And secondly, we will have a whole set of max-visiting words to reason about, which could potentially break the \symPPShort{}.

More precisely, the max-visiting words have the form $udc^\symCVar$ with $u \in \{0,1\}^r$, where $r$ is the number of $\Phi$'s variables, $c$ and $d$ are dummy letters, and $\symCVar$ is a value dependent on $|\Phi|$ (\cref{lem:aPhi}~(i)). The prefix $u$ is interpreted as a variable assignment for $\Phi$. Due to the structure of $\mathcal{A}_\Phi$, the only factorization that might witness the \symPPShort{} for such a max-visiting word $udc^\symCVar$ is $x=\varepsilon,y=u,z=dc^\symCVar$ (\cref{lem:aPhi}~(ii)). Thus, exactly the pumpings of the form $u^l d c^\symCVar$ are decisive for the \symPPShort{} and thereby for the compositionality of $\mathcal{A}_\Phi$.

Intuitively, $\mathcal{A}_\Phi$ uses the repetitions of the assignment $u$ in $u^l$ to check whether $u$ satisfies each clause (\cref{lem:aPhiOnerpOnePumpings}). Indeed, let $s$ be the number of clauses of $\Phi$, then the crucial pumping is $u^{s+2} d c^\symCVar$: To check the $s$ clauses, $s$ repetitions of $u$ are needed; two further repetitions of $u$ are needed for technical reasons.
If $u$ does not satisfy $\Phi$, then $\mathcal{A}_\Phi$ detects an unsatisfied clause while reading $u^{s+2}$ and immediately enters its rejecting sink.
Otherwise, that is, if $u$ satisfies $\Phi$, then $\mathcal{A}_\Phi$ detects no unsatisfied clause while reading $u^{s+2}$ and enters the accepting sink when beginning the final $u$-repetition.

To summarize, if $\Phi$ is unsatisfiable, then $\mathcal{A}_\Phi$ has the \symPPShort{} and is therefore composite.
Otherwise, that is, if $\Phi$ is satisfiable, then $\mathcal{A}_\Phi$ does not have the \symPPShort{} and is therefore prime, which is witnessed by the pumping $u^{s+2} d c^\symCVar$ of the max-visiting word $u d c^\symCVar$, where $u \in \{0,1\}^r$ encodes a satisfying assignment for $\Phi$.

We outline the reduction in \cref{subsec:constructionAPhi} and prove its correctness in \cref{subsec:correctnessAPhi}. 
To do so, we characterize the compositionality of the relevant DFAs in \cref{subsubsec:compositionalityADFApMLS} and exploit this characterization to prove the correctness in \cref{subsubsec:characterizationToNPHardness}.
To conclude, we prove, in \cref{subsec:aDFApMLSNPComplete}, that the restriction of \probPrimeDFA{} to the DFAs relevant for the reduction is \fontComplexityClass{NP}-complete, meaning that the established lower bound is tight for this specific class of DFAs.
% !!!!!!!!!! VERSIONSTART !!!!!!!!!!
% !!!!!!!!!! FULL VERSION !!!!!!!!!!
The formal proofs of the results of \cref{sec:npHardnessShortV2} can be found in \cref{sec:npHardness}.
% !!!!!!!!!! CONF VERSION !!!!!!!!!!
%The formal proofs of the results of \cref{sec:npHardnessShortV2} can be found in the full version \cite[Appendix A]{spenner2026deciding}.
% !!!!!!!!!! VERSIONEND  !!!!!!!!!!

\subsection{Construction of $\mathcal{A}_\Phi$}
\label{subsec:constructionAPhi}
We denote the propositional logic satisfiability problem for formulas in conjunctive normal form (CNF) by \probCNFSAT. 
It is well known that \probCNFSAT{} is \fontComplexityClass{NP}-complete \cite{cook1971complexity}.
From here on, we call propositional logic formulas in CNF simply CNF-formulas.

As mentioned, we prove \cref{the:primeDFANPHard} by a reduction from \probCNFSAT. For a given CNF-formula $\Phi$, the reduction yields a \symCNFDFA{} $\mathcal{A}_\Phi$ so that $\Phi$ is satisfiable if and only if $\mathcal{A}_\Phi$ is prime.

Before we can describe the reduction, we have to fix some notation.
For the remainder of \cref{sec:npHardnessShortV2}, let $X = \{x_1, \dots, x_r\}$ be a set of $r \in \symNatNumGeq{1}$ numbered variables, and let $\Phi$ be a CNF-formula over $X$. 
W.l.o.g.\ we assume that $\Phi$ has at least one clause and that every variable appears at most once in every clause. 

To simplify the construction of $\mathcal{A}_\Phi$, we fix a special notation for $\Phi$.
Namely, we write $\Phi = (\symElem_1^1 \vee \dots \vee \symElem_r^1) \wedge \dots \wedge (\symElem_1^s \vee \dots \vee \symElem_r^s)$,
where $s \in \symNatNumGeq{1}$ is the number of clauses of $\Phi$, and $\symElem_i^k \in \{\neg x_i, x_i, \bot\}$ is the $i$-th \highlightDef{element} in the $k$-th clause where $k \in \{1, \dots, s\}, i \in \{1, \dots, r\}$.
In other words, in our notation, every clause has exactly $r$ elements and the $i$-th element of every clause is either a literal of the $i$-th variable $x_i$ or $\bot$.
Here, $\bot$ is a special symbol that always evaluates to \emph{False}.
It is easy to see that $\Phi$ can be written in this way.

Towards the definition of $\mathcal{A}_\Phi$, we introduce some additional notation.
Let $\Sigma = \{0,1,c,d\}$.
We call a string $u \in \{0,1\}^r$ an \highlightDef{assignment string}.
As usual, an assignment over $X$ is a function $\gamma: X \rightarrow \{0,1\}$.
An assignment $\gamma$ over $X$ induces an assignment string $u_\gamma \in \{0,1\}^r$ in the following way: $u_\gamma = \gamma(x_1) \dots \gamma(x_r)$. Conversely, an assignment string $u = \sigma_1 \dots \sigma_r \in \{0,1\}^r$ induces an assignment $\gamma_u$ over $X$ in the following way: $\gamma_u(x_i) = \sigma_i$ for all $i \in \{1, \dots, r\}$.

Now we can consider the reduction, that is, the construction of the \symCNFDFA{} $\mathcal{A}_\Phi = (Q_\Phi, \Sigma, p_0, \delta_\Phi, F_\Phi)$ out of the given CNF-formula $\Phi$. The \symCNFDFA{} $\mathcal{A}_\Phi$ is outlined in \cref{fig:aPhiaPhiXTransitions}.
% !!!!!!!!!! VERSIONSTART !!!!!!!!!!
% !!!!!!!!!! FULL VERSION !!!!!!!!!!
It is formally defined in \cref{def:aPhi}, found in \cref{subsec:constructionAPhiProofs}.
% !!!!!!!!!! CONF VERSION !!!!!!!!!!
%It is formally defined in the full version \cite[Definition A.1]{spenner2026deciding}.
% !!!!!!!!!! VERSIONEND  !!!!!!!!!!
The key property of $\mathcal{A}_\Phi$ that we want to prove is formalized as:

\begin{restatable}{lemma}{lemaPhiCharacterization}
	\label{lem:aPhiCharacterization}\RestateRemark
	The CNF-formula $\Phi$ is satisfiable if and only if the \symCNFDFA{} $\mathcal{A}_\Phi$ is prime. \lipicsEnd
\end{restatable}

\begin{figure} % fig:aPhiaPhiXTransitions
	\newcommand{\gS}{4.75} % Gagdet step
	\newcommand{\lHS}{2.125} % Large horizontal step
	\newcommand{\hS}{2.125} % Horizontal step
	\newcommand{\lVS}{4.5} % Large vertical step
	\newcommand{\vS}{1.625} % Vertical step
	\newcommand{\aS}{1.0} % Auxiliary step
	\newcommand{\minSizeFS}{0.875cm}
	\newcommand{\minSizeSS}{0.875cm}
	\begin{subfigure}{1\textwidth}
		\centering
		\begin{tikzpicture}
			\scriptsize
			
			% First row
			
			\node[state, minimum size=\minSizeSS, initial, accepting]		at (0 *	\hS	+ 0 * \hS,0)			(p0)				{$p_0$};
			\node[state, minimum size=\minSizeSS, accepting]				at (1 * \hS	+ 0 * \hS,0)			(p1)				{$p_1$};
			\node[state, minimum size=\minSizeSS, accepting]				at (2 * \hS	+ 1 * \hS,0)			(pr-1)				{$p_{r-1}$};
			\node[state, minimum size=\minSizeSS, accepting, align=center]	at (3 * \hS	+ 1 * \hS,0)			(pr)				{\tiny{$p_r$}\\\tiny{/$\hat{p}_r^0$}};
			\node[state, minimum size=\minSizeSS, accepting]				at (4 * \hS	+ 1 * \hS,0)			(pc0)				{$p_c^0$};
			
			\node															at (0 * \hS + 0 * \hS, \aS)			(e00)				{$p_+$};
			\node															at (0 * \hS + 0 * \hS, -\aS)		(e01)				{$p_-$};
			\node															at (1 * \hS + 0 * \hS, \aS)			(e1)				{$p_+$};
			\node															at (2 * \hS + 1 * \hS, \aS)			(er-1)				{$p_+$};
			\node															at (3 * \hS + 1 * \hS, \aS)			(er)				{$p_+$};
			
			\node															at (5 * \hS + 1 * \hS, 0)			(ed1)				{$p_-$};
			\node															at (5 * \hS	+ 1 * \hS,\aS)			(erp10)				{$p_1^1$};
			\node															at (5 * \hS	+ 1 * \hS,-\aS)			(erp11)				{$\hat{p}_1^1$};
			
			\draw	(p0)						edge[below]			node{$0,1$}		(p1);
			\draw	(p0)						edge[left]			node{$c$}		(e00);
			\draw	(p0)						edge[left]			node{$d$}		(e01);
			
			\draw	(p1)						edge[dashed,below]	node{$0,1$}		(pr-1);
			\draw	(p1)						edge[left]			node{$c,d$}		(e1);
			
			\draw	(pr-1)						edge[below]			node{$0,1$}		(pr);
			\draw	(pr-1)						edge[left]			node{$c,d$}		(er-1);
			
			\draw	(pr)						edge[below]			node{$d$}		(pc0);
			\draw	(pr)						edge[left]			node{$c$}		(er);
			\draw	(pr)						edge[above,bend left=15]
			node[align=left,pos=0.65]{
				if $\symElem_1^1 = \bot$: $0,1$ \\
				if $\symElem_1^1 = x_1$: $0$ \\
				if $\symElem_1^1 = \neg x_1$: $1$
			}
			(erp10);
			\draw	(pr)						edge[below,bend right=15]
			node[align=left,pos=0.65]{
				if $\symElem_1^1 = x_1$: $1$ \\
				if $\symElem_1^1 = \neg x_1$: $0$
			}
			(erp11);
			
			\draw	(pc0)						edge[below]			node{$c$}		(erp10);
			\draw	(pc0)						edge[below]			node{$0,1,d$}	(ed1);

			% Dots after first row
			
			\node[]at(1.5 *	\hS	+ 0.5 * \hS , -0.5 * \lVS)(d01){$\dots$ clause rows $1$ to $k-1$ $\dots$};

			% Second row
			
			\node[state, minimum size=\minSizeSS, accepting]				at (0 *	\hS	+ 0 * \hS , -1 * \lVS)			(p1k)		{$p_1^k$};
			\node[state, minimum size=\minSizeSS, accepting]				at (0 *	\hS	+ 0 * \hS , -1 * \lVS - 1 * \vS)(p1kHat)	{$\hat{p}_1^k$};
			\node[state, minimum size=\minSizeSS, accepting]				at (1 *	\hS	+ 0 * \hS , -1 * \lVS)			(p2k)		{$p_2^k$};
			\node[state, minimum size=\minSizeSS, accepting]				at (1 *	\hS	+ 0 * \hS , -1 * \lVS - 1 * \vS)(p2kHat)	{$\hat{p}_2^k$};
			\node[state, minimum size=\minSizeSS, accepting]				at (2 *	\hS	+ 1 * \hS , -1 * \lVS)			(prm1k)		{$p_{r-1}^k$};
			\node[state, minimum size=\minSizeSS, accepting]				at (2 *	\hS	+ 1 * \hS , -1 * \lVS - 1 * \vS)(prm1kHat){$\hat{p}_{r-1}^k$};
			\node[state, minimum size=\minSizeSS, accepting]				at (3 *	\hS	+ 1 * \hS , -1 * \lVS)			(prk)		{$p_r^k$};
			\node[state, minimum size=\minSizeSS, accepting]				at (3 *	\hS	+ 1 * \hS , -1 * \lVS - 1 * \vS)(prkHat)	{$\hat{p}_r^k$};
			\node[state, minimum size=\minSizeSS, accepting]				at (4 *	\hS	+ 1 * \hS , -1 * \lVS - 1 * \vS)(pck)		{$p_c^k$};
			
			\node[]															at (3 *	\hS	+ 1 * \hS , -1 * \lVS + \aS)	(erk0)		{$p_-$};
			\node[]															at (3 *	\hS	+ 1 * \hS , -1 * \lVS - 1 * \vS - \aS)(erk1){$p_-$};
			
			\node[]															at (0 *	\hS	+ 0 * \hS , -1 * \lVS + \aS)	(e1k0)		{$p_-$};
			\node[]															at (1 *	\hS	+ 0 * \hS , -1 * \lVS + \aS)	(e2k0)		{$p_-$};
			\node[]															at (2 *	\hS	+ 1 * \hS , -1 * \lVS + \aS)	(erm1k0)		{$p_-$};
			\node[]															at (0 *	\hS	+ 0 * \hS , -1 * \lVS - 1 * \vS - \aS)(e1k1)	{$p_-$};
			\node[]															at (1 *	\hS	+ 0 * \hS , -1 * \lVS - 1 * \vS - \aS)(e2k1)	{$p_-$};
			\node[]															at (2 *	\hS	+ 1 * \hS , -1 * \lVS - 1 * \vS - \aS)(erm1k1)	{$p_-$};
			
			\node															at (5 * \hS + 1 * \hS, -1 * \lVS - 1 * \vS + 0 * \aS)(edk)	{$p_-$};
			\node[]															at (5 *	\hS	+ 1 * \hS, -1 * \lVS - 1 * \vS + 1 * \aS)(erp1k0) {$p_1^{k+1}$};
			\node[]															at (5 *	\hS	+ 1 * \hS, -1 * \lVS - 1 * \vS - 1 * \aS)(erp1k1)	{$\hat{p}_1^{k+1}$};
			
			\draw	(p1k)						edge[dashed]		node{}			(p2k);
			\draw	(p1k)						edge[dashed]		node{}			(p2kHat);
			\draw	(p1k)						edge[left]			node{$c$}		(p1kHat);
			\draw	(p1k)						edge[left]			node{$d$}		(e1k0);
			
			\draw	(p1kHat)					edge[below]			node{$0,1$}		(p2kHat);
			\draw	(p1kHat)					edge[below]			node[pos=0.25]{$c$}(p2k);
			\draw	(p1kHat)					edge[left]			node{$d$}		(e1k1);
			
			\draw	(p2k)						edge[dashed]		node{}			(prm1k);
			\draw	(p2k)						edge[left]			node{$c$}		(p2kHat);
			\draw	(p2k)						edge[left]			node{$d$}		(e2k0);
			
			\draw	(p2kHat)					edge[dashed]		node{}			(prm1kHat);
			\draw	(p2kHat)					edge[left]			node{$d$}		(e2k1);
			
			\draw	(prm1k)						edge[dashed]		node{}			(prk);
			\draw	(prm1k)						edge[dashed]		node{}			(prkHat);
			\draw	(prm1k)						edge[left]			node{$c$}		(prm1kHat);
			\draw	(prm1k)						edge[left]			node{$d$}		(erm1k0);
			
			\draw	(prm1kHat)					edge[below]			node{$0,1$}		(prkHat);
			\draw	(prm1kHat)					edge[below]			node[pos=0.25]{$c$}(prk);
			\draw	(prm1kHat)					edge[left]			node{$d$}		(erm1k1);
			
			\draw	(prk)						edge[left]			node{$0,1,d$}	(erk0);
			\draw	(prk)						edge[left]			node{$c$}		(prkHat);
			
			\draw	(prkHat)					edge[below]			node{$c$}		(pck);
			\draw	(prkHat)					edge[left]			node{$d$}		(erk1);
			\draw	(prkHat)					edge[above,bend left=15]
			node[align=left,pos=0.65]{
				if $\symElem_1^{k+1} = \bot$: $0,1$ \\
				if $\symElem_1^{k+1} = x_1$: $0$ \\
				if $\symElem_1^{k+1} = \neg x_1$: $1$
			}
			(erp1k0);
			\draw	(prkHat)					edge[below,bend right=15]
			node[align=left,pos=0.65]{
				if $\symElem_1^{k+1} = x_1$: $1$ \\
				if $\symElem_1^{k+1} = \neg x_1$: $0$
			}
			(erp1k1);
			
			\draw	(pck)						edge[below]			node{$0,1,d$}	(edk);
			\draw	(pck)						edge[below]			node{$c$}		(erp1k0);

			% Dots after second row
			\node[]at(1.5 * \hS	+ 0.5 * \hS , -1.5 * \lVS - 1 * \vS)(d11){$\dots$ clause rows $k+1$ to $s-1$ $\dots$};

			% Third row
			
			\node[state, minimum size=\minSizeSS, accepting]				at (0 *	\hS	+ 0 * \hS , -2 * \lVS - 1 * \vS)(p1s)		{$p_1^s$};
			\node[state, minimum size=\minSizeSS, accepting]				at (0 *	\hS	+ 0 * \hS , -2 * \lVS - 2 * \vS)(p1sHat)	{$\hat{p}_1^s$};
			\node[state, minimum size=\minSizeSS, accepting]				at (1 *	\hS	+ 0 * \hS , -2 * \lVS - 1 * \vS)(p2s)		{$p_2^s$};
			\node[state, minimum size=\minSizeSS, accepting]				at (1 *	\hS	+ 0 * \hS , -2 * \lVS - 2 * \vS)(p2sHat)	{$\hat{p}_2^s$};
			\node[state, minimum size=\minSizeSS, accepting]				at (2 *	\hS	+ 1 * \hS , -2 * \lVS - 1 * \vS)(prm1s)		{$p_{r-1}^s$};
			\node[state, minimum size=\minSizeSS, accepting]				at (2 *	\hS	+ 1 * \hS , -2 * \lVS - 2 * \vS)(prm1sHat){$\hat{p}_{r-1}^s$};
			\node[state, minimum size=\minSizeSS, accepting]				at (3 *	\hS	+ 1 * \hS , -2 * \lVS - 1 * \vS)(prs)		{$p_r^s$};
			\node[state, minimum size=\minSizeSS, accepting]				at (3 *	\hS	+ 1 * \hS , -2 * \lVS - 2 * \vS)(prsHat)	{$\hat{p}_r^s$};
			\node[state, minimum size=\minSizeSS]							at (4 *	\hS	+ 1 * \hS , -2 * \lVS - 1 * \vS)(p-)		{$p_-$};
			\node[state, minimum size=\minSizeSS, accepting]				at (4 *	\hS	+ 1 * \hS , -2 * \lVS - 2 * \vS)(pcs)		{$p_c^s$};
			\node[state, minimum size=\minSizeSS, accepting]				at (3 *	\hS	+ 1 * \hS , -2 * \lVS - 3 * \vS)(p+)		{$p_+$};
			
			\node[]															at (0 *	\hS	+ 0 * \hS , -2 * \lVS - 1 * \vS + \aS)(e1s0)	{$p_-$};
			\node[]															at (1 *	\hS	+ 0 * \hS , -2 * \lVS - 1 * \vS + \aS)(e2s0)	{$p_-$};
			\node[]															at (2 *	\hS	+ 1 * \hS , -2 * \lVS - 1 * \vS + \aS)(erm1s0)	{$p_-$};
			\node[]															at (0 *	\hS	+ 0 * \hS , -2 * \lVS - 2 * \vS - \aS)(e1s1)	{$p_-$};
			\node[]															at (1 *	\hS	+ 0 * \hS , -2 * \lVS - 2 * \vS - \aS)(e2s1)	{$p_-$};
			\node[]															at (2 *	\hS	+ 1 * \hS , -2 * \lVS - 2 * \vS - \aS)(erm1s1)	{$p_-$};
			
			\draw	(p1s)						edge[dashed]		node{}			(p2s);
			\draw	(p1s)						edge[dashed]		node{}			(p2sHat);
			\draw	(p1s)						edge[left]			node{$c$}		(p1sHat);
			\draw	(p1s)						edge[left]			node{$d$}		(e1s0);
			
			\draw	(p1sHat)					edge[below]			node{$0,1$}		(p2sHat);
			\draw	(p1sHat)					edge[below]			node[pos=0.25]{$c$}(p2s);
			\draw	(p1sHat)					edge[left]			node{$d$}		(e1s1);

			\draw	(p2s)						edge[dashed]		node{}			(prm1s);
			\draw	(p2s)						edge[left]			node{$c$}		(p2sHat);
			\draw	(p2s)						edge[left]			node{$d$}		(e2s0);
			
			\draw	(p2sHat)					edge[dashed]		node{}			(prm1sHat);
			\draw	(p2sHat)					edge[left]			node{$d$}		(e2s1);

			\draw	(prm1s)						edge[dashed]		node{}			(prs);
			\draw	(prm1s)						edge[dashed]		node{}			(prsHat);
			\draw	(prm1s)						edge[left]			node{$c$}		(prm1sHat);
			\draw	(prm1s)						edge[left]			node{$d$}		(erm1s0);
			
			\draw	(prm1sHat)					edge[below]			node{$0,1$}		(prsHat);
			\draw	(prm1sHat)					edge[below]			node[pos=0.25]{$c$}(prs);
			\draw	(prm1sHat)					edge[left]			node{$d$}		(erm1s1);

			\draw	(prs)						edge[above]			node{$0,1,d$}	(p-);
			\draw	(prs)						edge[left]			node{$c$}		(prsHat);
			
			\draw	(prsHat)					edge[left]			node{$0,1$}		(p+);
			\draw	(prsHat)					edge[above]			node{$c$}		(pcs);
			\draw	(prsHat)					edge[left]			node{$d$}		(p-);
			
			\draw	(pcs)						edge[right]			node{$0,1,d$}	(p+);
			\draw	(pcs)						edge[right]			node{$c$}		(p-);
			
			\draw	(p-)						edge[loop right]	node{$0,1,c,d$}	(p-);
			
			\draw	(p+)						edge[loop right]	node{$0,1,c,d$}	(p+);
		\end{tikzpicture}
		\caption{
			\symCNFDFA{} $\mathcal{A}_\Phi$ for CNF-formula $\Phi = (\symElem_1^1 \vee \dots \vee \symElem_r^1) \wedge \dots \wedge (\symElem_1^s \vee \dots \vee \symElem_r^s)$ over set $X = \{x_1, \dots, x_r\}$ of $r \in \symNatNumGeq{1}$ numbered variables.
			Depicted are the assignment device, the $k$-th clause row with $k \in \{1, \dots, s - 1\}$, and the $s$-th and last clause row.
			In the $k$-th clause row, the out-transitions of the states $p_1^k,\dots,p_{r-1}^k$ for the letters $0,1$ are omitted. For such a state $p_i^k$, these $0/1$-out-transitions depend on the element $e_{i+1}^k$ and can lead to $p_{i+1}^k$ or $\hat{p}_{i+1}^k$, as suggested by the dashed lines. The three possible configurations (for $e_{i+1}^k \in \{\bot,\neg x_i,x_i\}$) are depicted in \cref{subfig:aPhiXTransitions}.
			The same omissions are made in the $s$-th clause row.
		}
		\label{subfig:aPhi}
	\end{subfigure}
	\begin{subfigure}{1\textwidth}
		\centering
		\begin{tikzpicture}
			\scriptsize
			
			% bot gadget
			
			\node[state, minimum size=\minSizeFS, initial, accepting]		at (0,0)			(pik)				{$p_i^k$};
			\node[state, minimum size=\minSizeFS, initial, accepting]		at (0,-\vS)			(pikHat)			{$\hat{p}_i^k$};
			\node[state, minimum size=\minSizeFS, accepting]				at (\lHS,0)			(pi+1k)				{$p_{i+1}^k$};
			\node[state, minimum size=\minSizeFS, accepting]				at (\lHS,-\vS)		(pi+1kHat)			{$\hat{p}_{i+1}^k$};
			\node															at (0,\aS)			(e0)				{$p_-$};
			\node															at (0,-\vS-\aS)		(e1)				{$p_-$};
			
			\draw	(pik)						edge[above]			node{$0,1$}			(pi+1k);
			\draw	(pik)						edge[left]			node{$c$}			(pikHat);
			\draw	(pik)						edge[left]			node{$d$}			(e0);
			
			\draw	(pikHat)					edge[below]			node{$0,1$}			(pi+1kHat);
			\draw	(pikHat)					edge[below]			node[pos=0.25]{$c$}	(pi+1k);
			\draw	(pikHat)					edge[left]			node{$d$}			(e1);
			
			\node[]		at (0.5*\lHS,- \vS - 1.2*\aS)			(i)			{(i)};
			
			% x_i gadget
			
			\node[state, minimum size=\minSizeFS, initial, accepting]		at (\gS,0)			(pik)				{$p_i^k$};
			\node[state, minimum size=\minSizeFS, initial, accepting]		at (\gS,-\vS)		(pikHat)			{$\hat{p}_i^k$};
			\node[state, minimum size=\minSizeFS, accepting]				at (\gS + \lHS,0)	(pi+1k)				{$p_{i+1}^k$};
			\node[state, minimum size=\minSizeFS, accepting]				at (\gS + \lHS,-\vS)(pi+1kHat)			{$\hat{p}_{i+1}^k$};
			\node															at (\gS,\aS)		(e0)				{$p_-$};
			\node															at (\gS,-\vS-\aS)	(e1)				{$p_-$};
			
			\draw	(pik)						edge[above]			node{$0$}			(pi+1k);
			\draw	(pik)						edge[right]			node[pos=0.13]{$1$}(pi+1kHat);
			\draw	(pik)						edge[left]			node{$c$}			(pikHat);
			\draw	(pik)						edge[left]			node{$d$}			(e0);
			
			\draw	(pikHat)					edge[below]			node{$0,1$}			(pi+1kHat);
			\draw	(pikHat)					edge[below]			node[pos=0.25]{$c$}	(pi+1k);
			\draw	(pikHat)					edge[left]			node{$d$}			(e1);
			
			\node[]		at (\gS + 0.5*\lHS,- \vS - 1.2*\aS)			(ii)			{(ii)};
			
			% \neg x_i gadget
			
			\node[state, minimum size=\minSizeFS, initial, accepting]		at (2*\gS,0)			(pik)				{$p_i^k$};
			\node[state, minimum size=\minSizeFS, initial, accepting]		at (2*\gS,-\vS)		(pikHat)			{$\hat{p}_i^k$};
			\node[state, minimum size=\minSizeFS, accepting]				at (2*\gS + \lHS,0)	(pi+1k)				{$p_{i+1}^k$};
			\node[state, minimum size=\minSizeFS, accepting]				at (2*\gS + \lHS,-\vS)(pi+1kHat)			{$\hat{p}_{i+1}^k$};
			\node															at (2*\gS,\aS)		(e0)				{$p_-$};
			\node															at (2*\gS,-\vS-\aS)	(e1)				{$p_-$};
			
			\draw	(pik)						edge[above]			node{$1$}			(pi+1k);
			\draw	(pik)						edge[right]			node[pos=0.13]{$0$}(pi+1kHat);
			\draw	(pik)						edge[left]			node{$c$}			(pikHat);
			\draw	(pik)						edge[left]			node{$d$}			(e0);
			
			\draw	(pikHat)					edge[below]			node{$0,1$}			(pi+1kHat);
			\draw	(pikHat)					edge[below]			node[pos=0.25]{$c$}	(pi+1k);
			\draw	(pikHat)					edge[left]			node{$d$}			(e1);
			
			\node[]		at (2*\gS + 0.5*\lHS,- \vS - 1.2*\aS)			(iii)			{(iii)};
		\end{tikzpicture}
		\caption{
			Out-transitions of the states $p_i^k, \hat{p}_i^k$ of the \symCNFDFA{} $\mathcal{A}_\Phi$ depicted in \cref{subfig:aPhi}, where $k \in \{1, \dots, s\},i \in \{1, \dots, r - 1\}$.
			The out-transitions of $p_i^k$ depend on $\symElem_{i+1}^k$.
			(i): Out-transitions if $\symElem_{i+1}^k = \bot$. (ii): Out-transitions if $\symElem_{i+1}^k = x_i$. (iii): Out-transitions if $\symElem_{i+1}^k = \neg x_i$.
		}
		\label{subfig:aPhiXTransitions}
	\end{subfigure}
	\caption{\symCNFDFA{} $\mathcal{A}_\Phi$.}
	\label{fig:aPhiaPhiXTransitions}
\end{figure}

As mentioned above, the compositionality of $\mathcal{A}_\Phi$ depends solely on the max-visiting words of $\mathcal{A}_\Phi$. We will discuss this in \cref{subsubsec:compositionalityADFApMLS}. For now, we will focus on the design of $\mathcal{A}_\Phi$ and give an intuition for it.

The \symCNFDFA{} $\mathcal{A}_\Phi$ consists of two devices, the \highlightDef{assignment device} and the \highlightDef{formula device}.
The assignment device consists of the states $p_0,\dots,p_r$ and $p_c^0$.
To fully traverse the assignment device, that is, to start in $p_0$ and reach $p_c^0$, a string $ud$ with $u \in \{0,1\}^r$ has to be read.
Therefore, the assignment device enforces that every max-visiting word is prefixed by an assignment string $u \in \{0,1\}^r$ (followed by a single $d$).
Subsequently, the assignment string $u$ is interpreted as the assignment $\gamma_u$ for $\Phi$ which it induces.

The formula device consists of $s$ \highlightDef{clause rows}.
For $k \in \{1,\dots,s\}$, the $k$-th clause row encodes the $k$-th clause of $\Phi$, that is, $(\symElem_1^k \vee \dots \vee \symElem_r^k)$.
Thus, the formula device in its entirety encodes all clauses and thereby $\Phi$.
The $k$-th clause row  consists of the states $p_1^k,\hat{p}_1^k,\dots,p_r^k,\hat{p}_r^k$ and $p_c^k$.
Among these, $\hat{p}_r^k$ is the positive target state of the $k$-th clause row, which is also the designated entry point for the subsequent clause row.
And further, $p_r^k$ is the negative target state of the $k$-th clause row, which prevents $\mathcal{A}_\Phi$ from entering the subsequent clause row.

We point out that $p_r$ is part of the assignment device enforcing the prefix $ud$, but also functions similarly to a positive target state for letters $0$ and $1$. In particular, it is the designated entry point for the first clause row. Therefore, we denote the state $p_r$ also by $\hat{p}_r^0$.

The formula device is designed to enforce two assertions.

Firstly, every max-visiting word of $\mathcal{A}_\Phi$ is of the form $udc^\symCVar$, where $u \in \{0,1\}^r$ is an assignment string and $\symCVar$ depends on $|\Phi|$. To be precise, this first assertion is enforced by the two devices together: The assignment device enforces the prefix $ud$ and the formula device the suffix $c^\symCVar$. The first assertion is formalized further below in \cref{lem:aPhi}~(i).

Secondly, if $\gamma_u$ satisfies the $k$-th clause, then $u$ induces a run from the positive target state $\hat{p}_r^{k-1}$ of the $(k-1)$-th clause to the positive target state $\hat{p}_r^k$ of the $k$-th clause. Else, $u$ induces a run from the positive target state $\hat{p}_r^{k-1}$ of the $(k-1)$-th clause to the negative target state $p_r^k$ of the $k$-th clause.
The second assertion is formalized in:

\begin{restatable}{lemma}{lemaPhiClauseRows}
	\label{lem:aPhiClauseRows}\RestateRemark
	Consider an assignment string $u \in \{0,1\}^r$. Let $k \in \{1, \dots, s\}$. If $\gamma_u$ satisfies the $k$-th clause of $\Phi$, then $\delta_\Phi(\hat{p}_r^{k-1},u) = \hat{p}_r^k$ holds. Else, $\delta_\Phi(\hat{p}_r^{k-1},u) = p_r^k$ holds. \lipicsEnd
\end{restatable}

Note that reading $0$ or $1$ in any negative target state $p_r^k$ leads to $p_-$, while reading $0$ or $1$ in the final positive target state $\hat{p}_r^s$ leads to $p_+$.
Together with the second assertion enforced by the formula device, this implies that, for every $l \geq s+2$, a word of the form $u^ldc^\symCVar$ is accepted by $\mathcal{A}_\Phi$ if and only if $\gamma_u$ satisfies $\Phi$.
We expand on this argument in our discussion of \cref{lem:aPhiOnerpOnePumpings} in \cref{subsubsec:characterizationToNPHardness}.
Also in \cref{subsubsec:characterizationToNPHardness}, we will see that $\mathcal{A}_\Phi$ is prime if and only if $\mathcal{A}_\Phi$ accepts some word of the form $u^ldc^\symCVar$. Thus, $\mathcal{A}_\Phi$ is prime if and only if $\Phi$ is satisfiable.

We have outlined the \symCNFDFA{} $\mathcal{A}_\Phi$ yielded by the reduction and given an intuition for its design. Next, we prove the correctness of the reduction.

\subsection{Correctness of the Reduction}
\label{subsec:correctnessAPhi}
We complete the proof of \cref{the:primeDFANPHard} by proving the correctness of the reduction.
That is, we prove \cref{lem:aPhiCharacterization}, that $\Phi$ is satisfiable if and only if $\mathcal{A}_\Phi$ is prime.
We accomplish this in two steps. In \cref{subsubsec:compositionalityADFApMLS}, we characterize the compositionality of \symCNFDFA{}s.
Then in \cref{subsubsec:characterizationToNPHardness}, we exploit this characterization to prove \cref{lem:aPhiCharacterization} and thereby \cref{the:primeDFANPHard}.

\subsubsection{Compositionality of Minimal Linear Safety $\symADFAp$s}
\label{subsubsec:compositionalityADFApMLS}
We first define minimal linear safety $\symADFAp$s and then characterize their compositionality.

\paragraph*{Definitions: $\symADFAp$s, Linear DFAs, Safety DFAs}
\label{par:definitionsADFApSafetyLinearity}
Following convention, a DFA $\mathcal{A}$ is \highlightDef{acyclic} ($\symADFA$) if
%every cycle in $\mathcal{A}$ begins in a rejecting sink.
there are no cycles in $\mathcal{A}$ besides those of the rejecting sinks.
Building on that, a DFA $\mathcal{A}$ is an \highlightDef{$\symADFA$-plus} ($\symADFAp$) if
$\mathcal{A}$ possesses both a rejecting and an accepting sink and
%every cycle in $\mathcal{A}$ begins in a sink.
there are no cycles in $\mathcal{A}$ besides those of the sinks.
Thus, an $\symADFAp$ is "almost" an $\symADFA$ but additionally has an accepting sink.
%Clearly, a DFA decides a finite language if and only if its minimal DFA is an $\symADFA$.

We call a DFA $\mathcal{A} = (Q,\Sigma,q_I,\delta,F)$ \highlightDef{linear} if, for every $q,q' \in Q$ with $q \neq q'$, exactly one of the following holds: (i) $q'$ is reachable from $q$; (ii) $q$ is reachable from $q'$; (iii) $q$ and $q'$ are unreachable from each other, and $q$ or $q'$ or both are sinks.
Obviously, every minimal linear DFA has at least one and at most two sinks.

For a minimal $\symADFAp$ $\mathcal{A} = (Q,\Sigma,q_I,\delta,F)$, we denote by $\opLinLen{\mathcal{A}}$ the length of the longest word $w \in \Sigma^*$ such that words $v,v' \in \Sigma^*$ with $\delta(q_I, w v) \in F$ and $\delta(q_I, w v') \notin F$ exist.
Thus, $\opLinLen{\mathcal{A}}$ is the length of the longest word on which $\mathcal{A}$ does not enter a sink.
Because $\mathcal{A}$ is minimal and has two sinks, such a word always exists.
The following holds:

\begin{restatable}{lemma}{lemaDFASizeAndLinearity}
	\label{lem:aDFASizeAndLinearity}\RestateRemark
	Consider a minimal $\symADFAp$ $\mathcal{A} = (Q, \Sigma, q_I, \delta, F)$. Then the following assertions hold:
	\begin{romanenumerate}
		\item $\opSize{\mathcal{A}} \geq \opLinLen{\mathcal{A}} + 3$.
		\item $\mathcal{A}$ is linear if and only if $\opSize{\mathcal{A}} = \opLinLen{\mathcal{A}} + 3$. \lipicsEnd
	\end{romanenumerate}
\end{restatable}

Finally, we introduce a type of DFA already inspected in \cite{kupferman2015prime}.
A regular language $L \subseteq \Sigma^*$ is a \highlightDef{safety language} if, for every $w \in \Sigma^*$, it holds that $w \notin L$ implies $wy \notin L$ for every $y \in \Sigma^*$.
A DFA $\mathcal{A}$ is a \highlightDef{safety DFA} if $\opLang{\mathcal{A}}$ is a safety language.
Clearly, every non-trivial minimal safety DFA has exactly one rejecting state, and this state is a sink.

% !!!!!!!!!! VERSIONSTART !!!!!!!!!!
% !!!!!!!!!! FULL VERSION !!!!!!!!!!
The general form of minimal linear safety $\symADFAp$s is outlined in \cref{fig:pmADFAMinLinSafe}. The following observation can be easily verified by inspecting \cref{def:aPhi,fig:aPhiaPhiXTransitions}.
% !!!!!!!!!! CONF VERSION !!!!!!!!!!
%The general form of minimal linear safety $\symADFAp$s is outlined in \cite[Figure 3]{spenner2026deciding}. The following observation can be easily verified by inspecting \cite[Definition A.1]{spenner2026deciding} and Figure 1.
% !!!!!!!!!! VERSIONEND  !!!!!!!!!!
\begin{restatable}{lemma}{lemaPhiIsMLSADFAp}
	\label{lem:aPhiIsMLSADFAp}\RestateRemark
	The \symCNFDFA{} $\mathcal{A}_\Phi$ is a minimal linear safety $\symADFAp$. \lipicsEnd
\end{restatable}

Clearly, $\mathcal{A}_\Phi$ is a safety DFA since the rejecting sink $p_-$ is the only rejecting state of $\mathcal{A}_\Phi$.
Further, to see that $\mathcal{A}_\Phi$ is additionally a linear $\symADFAp$, recall the following observation from \cref{subsec:constructionAPhi}: 
To fully traverse $\mathcal{A}_\Phi$, a word of the form $udc^\symCVar$ for $u \in \{0,1\}^r$ and a suitable $\symCVar$ has to be read. On such a word, the states of the assignment device are traversed in the order $p_0,\dots,p_r,p_c^0$, and the states of the formula device are traversed in the order $p_1^1,\hat{p}_1^1,\dots,p_r^1,\hat{p}_r^1,p_c^1,\dots,p_1^s,\hat{p}_1^s,\dots,p_r^s,\hat{p}_r^s,p_c^s$. Intuitively, this order describes the "forward direction" in $\mathcal{A}_\Phi$, and there are no "backward transitions" in $\mathcal{A}_\Phi$. Thus, $\mathcal{A}_\Phi$ is acyclic (disregarding $p_+$) and its states are linearly ordered, meaning that $\mathcal{A}_\Phi$ is a linear $\symADFAp$.
Lastly, regarding the minimality, note that two states of $\mathcal{A}_\Phi$ are differentiated by the suffixes of words $udc^\symCVar$ needed to fully traverse the remainder of $\mathcal{A}_\Phi$ and enter the rejecting sink when starting in the respective states.
For example, two states from the formula device are differentiated by the number of $c$'s needed to reach the rejecting sink.
Thus, $\mathcal{A}_\Phi$ is minimal.
In total, $\mathcal{A}_\Phi$ is a minimal linear safety $\symADFAp$.

Thanks to \cref{lem:aPhiIsMLSADFAp}, it is sufficient to characterize the compositionality of minimal linear safety $\symADFAp$s in order to characterize the compositionality of \symCNFDFA{}s.

\paragraph*{Characterization of the Compositionality of Minimal Linear Safety $\symADFAp$s}
\label{par:characterizationOfADFApMLSCompositionality}
Now we characterize the compositionality of minimal linear safety $\symADFAp$s. With \cref{exm:mpProp} at the end of this section, we illustrate the main ideas of the characterization by considering, by way of example, two minimal linear safety $\symADFAp$s, one composite and one prime.

We point out that when we consider the compositionality of a restricted class of DFAs, the (smaller) DFAs used in the decompositions are not restricted. In particular, the DFAs we use in the characterization of the compositionality of minimal linear safety $\symADFAp$s need not be minimal linear safety $\symADFAp$s themselves.

The compositionality of a DFA $\mathcal{A}$ boils down to the question whether, for each word $w \notin \opLang{\mathcal{A}}$, there exists some DFA $\mathcal{B}$ with $w \notin \opLang{\mathcal{B}}$ and $\mathcal{B} \in \opDecompSet{\mathcal{A}}$, the latter meaning $\opSize{\mathcal{B}} < \opInd{\mathcal{A}}$ and $\opLang{\mathcal{A}} \subseteq \opLang{\mathcal{B}}$.
Indeed, if the answer is yes, the collection of these DFAs yields a decomposition of $\mathcal{A}$ (thanks to $\opSize{\mathcal{B}} < \opInd{\mathcal{A}}$, the collection is finite).
If, on the other hand, for some $w \notin \opLang{\mathcal{A}}$, there is no such $\mathcal{B}$, then $w$ is a primality witness of $\mathcal{A}$.

For a minimal linear safety $\symADFAp$ $\mathcal{A}$, the situation is simplified by the fact that there is only one rejecting state, namely, the rejecting sink.
For words $w$ on which $\mathcal{A}$ enters the rejecting sink without visiting every non-sink, building an appropriate DFA rejecting $w$ is relatively simple. 
Essentially, we can just remove one of the skipped-over states $q$ of $\mathcal{A}$ so that the constructed DFA accepts all words on which $\mathcal{A}$ visits $q$ and acts as $\mathcal{A}$ on all other words, including $w$.
% !!!!!!!!!! VERSIONSTART !!!!!!!!!!
% !!!!!!!!!! FULL VERSION !!!!!!!!!!
Details are given in \cref{def:aip}, \cref{subfig:aip} and \cref{lem:aip} in \cref{subsubsec:compositionalityADFApMLSProofs}.
% !!!!!!!!!! CONF VERSION !!!!!!!!!!
%Details are given in the full version \cite[Definition A.3, Figure 4a and Lemma A.4]{spenner2026deciding}.
% !!!!!!!!!! VERSIONEND  !!!!!!!!!!
Thus, in the following, we only need to consider the words on which $\mathcal{A}$ enters its rejecting sink after visiting every non-sink. Formally, we define the set of these words as:

\begin{restatable}{definitionRoman}{defdifficultWords}
	\label{def:difficultWords}\RestateRemark
	Consider a minimal linear safety $\symADFAp$ $\mathcal{A}$ over an alphabet $\Sigma$. Let $\symLinLenVar = \opLinLen{\mathcal{A}}$.
	Then we define the set $\opDiffWords{\mathcal{A}}$ of the \highlightDef{max-visiting words of} $\mathcal{A}$ as follows:
	\begin{align*}
		\opDiffWords{\mathcal{A}} = \{v \sigma \mid v \in \Sigma^\symLinLenVar \wedge \sigma \in \Sigma \wedge v \in \opLang{\mathcal{A}} \wedge v \sigma \notin \opLang{\mathcal{A}}\}. \tag*{\lipicsEnd}
	\end{align*}
\end{restatable}

On a max-visiting word $w = v \sigma \in \opDiffWords{\mathcal{A}}$, the $\symADFAp$ $\mathcal{A}$ does not skip over any non-sink.
Thus, the idea of removing a skipped-over state to construct an appropriate DFA rejecting $w$ is not applicable.
Indeed, the situation for the max-visiting words is more involved.
It turns out that every relevant $\mathcal{B} \in \opDecompSet{\mathcal{A}}$ is a safety DFA with an accepting sink.
Thus, every such $\mathcal{B}$ possesses one accepting sink, one rejecting sink, and at most $\symLinLenVar$ accepting non-sinks.
% !!!!!!!!!! VERSIONSTART !!!!!!!!!!
% !!!!!!!!!! FULL VERSION !!!!!!!!!!
Here, we omit the details and refer to \cref{lem:safetyDecomp,lem:acceptingSink} and their proofs in \cref{subsubsec:compositionalityADFApMLSProofs}.
% !!!!!!!!!! CONF VERSION !!!!!!!!!!
%Here, we omit the details and refer to the full version \cite[Lemmas A.10 and A.11]{spenner2026deciding}.
% !!!!!!!!!! VERSIONEND  !!!!!!!!!!
Since $\mathcal{B}$ has at most $\symLinLenVar$ non-sinks and the prefix $v$ of $w$ has length $\symLinLenVar$, the initial run of $\mathcal{B}$ on $v$ has to visit a non-sink twice.
%Thus, every DFA in $\opDecompSet{\mathcal{A}}$ is "at least one state short" to keep track of $w$.
Therefore, for some $x,y,z$ with $w = xyz$ and $|y|,|z| > 0$, the DFA $\mathcal{B}$ reaches the same state after reading $x$, $xy$ and every further $xy^l$. This observation gives rise to the definition of the \symPPLong:

\begin{restatable}{definitionRoman}{defpP}
	\label{def:pP}\RestateRemark
	Consider a minimal linear safety $\symADFAp$ $\mathcal{A}$ over an alphabet $\Sigma$. We say that $\mathcal{A}$ has the \highlightDef{\symPPLong} (\highlightDef{\symPPShort}) if, for every max-visiting word $w \in \opDiffWords{\mathcal{A}}$, there are $x,y,z \in \Sigma^*$ with $w = xyz$ and $|y|,|z|>0$ so that $xy^lz \notin \opLang{\mathcal{A}}$ holds for every $l \in \symNatNum$. \lipicsEnd
\end{restatable}

Before we state the desired characterization of the compositionality of minimal linear safety $\symADFAp$s in terms of the \symPPShort{}, we introduce an alternative formulation of the \symPPShort{}, for which we need some additional notation.
Consider a word $w = \sigma_1 \dots \sigma_n \in \Sigma^n$.
For $i, j \in \{1,\dots,n+1\}, i < j, l \in \symNatNum$, we define $\opPump{w}{i,j;l} = \sigma_1\dots\sigma_{i-1}(\sigma_i\dots\sigma_{j-1})^l\sigma_{j}\dots\sigma_n$.
We refer to $\opPump{w}{i,j;l}$ as a \highlightDef{pumping of} $w$.
The following lemma provides an alternative formulation of the \symPPShort{} by translating it into the $\opPump{w}{i,j;l}$-notation.

\begin{restatable}{lemma}{lempP}
	\label{lem:pP}\RestateRemark
	Consider a minimal linear safety $\symADFAp$ $\mathcal{A}$. Let $\symLinLenVar = \opLinLen{\mathcal{A}}$. Then $\mathcal{A}$ has the \symPPShort{} if and only if, for every max-visiting word $w \in \opDiffWords{\mathcal{A}}$, there are $i,j \in \{1,\dots,\symLinLenVar+1\}$ with $i<j$ so that $\opPump{w}{i,j;l} \notin \opLang{\mathcal{A}}$ holds for every $l \in \symNatNum$. \lipicsEnd
\end{restatable}

For
a minimal linear safety $\symADFAp$ $\mathcal{A}$,
a max-visiting word $w \in \opDiffWords{\mathcal{A}}$
and indices $i,j$,
we say that the \highlightDef{\symPPShort-condition holds for $w,i,j$} if
%$\delta(q_0, \opPump{w}{i,j;l}) = q_-$
$\opPump{w}{i,j;l} \notin \opLang{\mathcal{A}}$
holds for every $l \in \symNatNum$.
Note that $\opPump{w}{i,j;l} \notin \opLang{\mathcal{A}}$ implies $\opPump{w}{i,j;l}w' \notin \opLang{\mathcal{A}}$ for all words $w'$ because $\mathcal{A}$ is a safety DFA.

With the \symPPShort{} defined, we can characterize the compositionality of minimal linear safety $\symADFAp$s as follows:

\begin{restatable}{theorem}{thepmADFAMinLinSafeCompositionality}
	\label{the:pmADFAMinLinSafeCompositionality}\RestateRemark
	A minimal linear safety $\symADFAp$ is prime if and only if it does not have the \symPPShort. \lipicsEnd
\end{restatable}

We now set out to give an intuition for this result.

First, towards a proof by contraposition, assume that the minimal linear safety $\symADFAp$ $\mathcal{A}$ has the \symPPShort{} and consider a max-visiting word $w = \sigma_1 \dots \sigma_{\symLinLenVar+1} \in \opDiffWords{\mathcal{A}}$. Then there are some indices $i,j$ with $\opPump{w}{i,j;l} \notin \opLang{\mathcal{A}}$ for every $l$. Thus, when trying to build a DFA $\mathcal{B} \in \opDecompSet{\mathcal{A}}$ rejecting $w$, we know that $\mathcal{B}$ is allowed to reject every $\opPump{w}{i,j;l}$ as well. We can use this to construct $\mathcal{B}$ out of $\mathcal{A}$ by essentially merging the states that $\mathcal{A}$ reaches after reading $\sigma_1 \dots \sigma_{i-1}$ and $\sigma_1 \dots \sigma_{j-1}$. Since we can build a suitable DFA for every max-visiting word, $\mathcal{A}$ is composite.
% !!!!!!!!!! VERSIONSTART !!!!!!!!!!
% !!!!!!!!!! FULL VERSION !!!!!!!!!!
Details for this construction are given in \cref{def:awij}, \cref{subfig:awij} and \cref{lem:awij} in \cref{subsubsec:compositionalityADFApMLSProofs}.
% !!!!!!!!!! CONF VERSION !!!!!!!!!!
%Details for this construction are given in the full version \cite[Definition A.6, Figure 4b and Lemma A.7]{spenner2026deciding}.
% !!!!!!!!!! VERSIONEND  !!!!!!!!!!

Now assume that $\mathcal{A}$ does not have the \symPPShort. Then there is a word $w \in \opDiffWords{\mathcal{A}}$ that breaks the \symPPShort, meaning that, for every indices $i,j$, there is an $l$ with $\opPump{w}{i,j;l} \in \opLang{\mathcal{A}}$. We argue that this $w$ is a primality witness of $\mathcal{A}$.

Towards a contradiction, assume that there is a $\mathcal{B} \in \opDecompSet{\mathcal{A}}$ with $w \notin \opLang{\mathcal{B}}$. We pointed out above that there have to be indices $i,j$ so that, for every $l$, the DFA $\mathcal{B}$ cannot differentiate between $w$ and $\opPump{w}{i,j;l}$, meaning that $\mathcal{B}$ rejects every such pumping $\opPump{w}{i,j;l}$. Yet, because $w$ breaks the \symPPShort, there is some $l$ so that $\opPump{w}{i,j;l} \in \opLang{\mathcal{A}}$,
yielding a contradiction.
%With $\opPump{w}{i,j;l} \notin \opLang{\mathcal{B}}$ and $\opPump{w}{i,j;l} \in \opLang{\mathcal{A}}$, we have $\opLang{\mathcal{A}} \not\subseteq \opLang{\mathcal{B}}$, which contradicts $\mathcal{B} \in \opDecompSet{\mathcal{A}}$.

Since there is no $\mathcal{B} \in \opDecompSet{\mathcal{A}}$ with $w \notin \opLang{\mathcal{B}}$, the selected max-visiting word $w$ is a primality witness of $\mathcal{A}$, meaning that $\mathcal{A}$ is prime.

In conclusion, if $\mathcal{A}$ has the \symPPShort, we can build, for every $w \in \opDiffWords{\mathcal{A}}$, a suitable DFA rejecting $w$ by merging appropriate states of $\mathcal{A}$. Otherwise, that is, if $\mathcal{A}$ does not have the \symPPShort, then every word $w \in \opDiffWords{\mathcal{A}}$ breaking the \symPPShort{} is a primality witness of $\mathcal{A}$ because every DFA in $\opDecompSet{\mathcal{A}}$ is necessarily confused about some pumpings of $w$ and can therefore not reject $w$ itself. In total, we argued that a minimal linear safety $\symADFAp$ is prime if and only if it does not have the \symPPShort, which is formalized above in \cref{the:pmADFAMinLinSafeCompositionality}.

\begin{figure}
	\newcommand{\nHS}{1.5} 	% Norma horizontal step
	\newcommand{\sHS}{1.25}	% Small horizontal step
	\newcommand{\nVS}{0.75}	% Normal vertical step
	\newcommand{\lVS}{3}	% Large vertical step
	\scriptsize
	\begin{subfigure}{0.485\textwidth}
		\centering
		\begin{tikzpicture}
			\node[state, initial, accepting]		at ( 0 * \nHS,  0 * \nVS)			(q0)				{$q_0$};
			\node[state, accepting]					at ( 1 * \nHS,  0 * \nVS)			(q1)				{$q_1$};
			\node[state, accepting]					at ( 2 * \nHS,  0 * \nVS)			(q2)				{$q_2$};
			\node[state, accepting]					at ( 3 * \nHS,  1 * \nVS)			(q+)				{$q_+$};
			\node[state]							at ( 3 * \nHS, -1 * \nVS)			(q-)				{$q_-$};
			
			\draw (q0)	edge[above]						node{$a$}			(q1);
			\draw (q0)	edge[above, bend left]			node{$b$}			(q2);
			%\draw (q0)	edge[above, bend left, red]		node{$b$}			(q2);
			
			\draw (q1)	edge[below, bend right=12.5]	node{$a$}			(q-);
			\draw (q1)	edge[above]						node{$b$}			(q2);
			
			\draw (q2)	edge[above]						node{$a$}			(q+);
			\draw (q2)	edge[above]						node{$b$}			(q-);
			
			\draw (q+)	edge[loop right]				node{$\Sigma$}		(q+);
			
			\draw (q-)	edge[loop right]				node{$\Sigma$}		(q-);
		\end{tikzpicture}
		\caption{
			Minimal linear safety $\symADFAp$ $\mathcal{A}$ from \cref{exm:mpProp}. It has the \symPPShort{}. Thus, it is \linebreak composite.
		}
		\label{subfig:aDFAWithMpProp}
	\end{subfigure}
	\begin{subfigure}{0.485\textwidth}
		\centering
		\begin{tikzpicture}
			\node[state, initial, accepting]		at ( 0 * \nHS,  0 * \nVS)			(q0)				{$q_0$};
			\node[state, accepting]					at ( 1 * \nHS,  0 * \nVS)			(q1)				{$q_1$};
			\node[state, accepting]					at ( 2 * \nHS,  0 * \nVS)			(q2)				{$q_2$};
			\node[state, accepting]					at ( 3 * \nHS,  1 * \nVS)			(q+)				{$q_+$};
			\node[state]							at ( 3 * \nHS, -1 * \nVS)			(q-)				{$q_-$};
			
			\draw (q0)	edge[above]						node{$a$}			(q1);
			\draw (q0)	edge[above, bend left=12.5]		node{$b$}			(q+);
			%\draw (q0)	edge[above, bend left=12.5, red]node{$b$}			(q+);
			
			\draw (q1)	edge[below, bend right=12.5]	node{$a$}			(q-);
			\draw (q1)	edge[above]						node{$b$}			(q2);
			
			\draw (q2)	edge[above]						node{$a$}			(q+);
			\draw (q2)	edge[above]						node{$b$}			(q-);
			
			\draw (q+)	edge[loop right]				node{$\Sigma$}		(q+);
			
			\draw (q-)	edge[loop right]				node{$\Sigma$}		(q-);
		\end{tikzpicture}
		\caption{
			Minimal linear safety $\symADFAp$ $\mathcal{A}'$ from \cref{exm:mpProp}. It does not have the \symPPShort{}. Thus, it is prime.
		}
		\label{subfig:aDFAWithoutMpProp}
	\end{subfigure}
	\begin{subfigure}{1\textwidth}
		%
		%\centering
		\begin{tikzpicture}
			%
			% DFA A_1^+
			
			\node[state, draw=none]					at (-1 * \sHS,  0 * \nVS)						(xn)				{$\mathcal{A}_1^+$:};
			\node[state, initial, accepting]		at ( 0 * \sHS,  0 * \nVS)						(xq0)				{$q_0$};
			\node[state, draw=none]					at ( 1 * \sHS,  0 * \nVS)						(xq1)				{};
			\node[state, accepting]					at ( 2 * \sHS,  0 * \nVS)						(xq2)				{$q_2$};
			\node[state, accepting]					at ( 3 * \sHS,  1 * \nVS)						(xq+)				{$q_+$};
			\node[state]							at ( 3 * \sHS, -1 * \nVS)						(xq-)				{$q_-$};
			
			\draw (xq0)	edge[above, bend left=12.5]		node{$a$}			(xq+);
			\draw (xq0)	edge[above]						node{$b$}			(xq2);
			
			%\draw (q1)	edge[below, bend right=12.5]	node{$a$}			(q-);
			%\draw (q1)	edge[above]						node{$b$}			(q2);
			
			\draw (xq2)	edge[above]						node{$a$}			(xq+);
			\draw (xq2)	edge[above]						node{$b$}			(xq-);
			
			\draw (xq+)	edge[loop right]				node{$\Sigma$}		(xq+);
			
			\draw (xq-)	edge[loop right]				node{$\Sigma$}		(xq-);
		\end{tikzpicture}
		\begin{tikzpicture}
			%
			% DFA A_2^+
			
			\node[state, draw=none]					at (-1 * \sHS,  0 * \nVS)			(yn)				{$\mathcal{A}_2^+$:};
			\node[state, initial, accepting]		at ( 0 * \sHS,  0 * \nVS)			(yq0)				{$q_0$};
			\node[state, accepting]					at ( 1 * \sHS,  0 * \nVS)			(yq1)				{$q_1$};
			\node[state, draw=none]					at ( 2 * \sHS,  0 * \nVS)			(yq2)				{};
			\node[state, accepting]					at ( 3 * \sHS,  1 * \nVS)			(yq+)				{$q_+$};
			\node[state]							at ( 3 * \sHS, -1 * \nVS)			(yq-)				{$q_-$};
			
			\draw (yq0)	edge[above]						node{$a$}			(yq1);
			\draw (yq0)	edge[above, bend left=12.5]		node{$b$}			(yq+);
			
			\draw (yq1)	edge[below]						node{$a$}			(yq-);
			\draw (yq1)	edge[above]						node{$b$}			(yq+);
			
			%\draw (q2)	edge[above]						node{$a$}			(q+);
			%\draw (q2)	edge[above]						node{$b$}			(q-);
			
			\draw (yq+)	edge[loop right]				node{$\Sigma$}		(yq+);
			
			\draw (yq-)	edge[loop right]				node{$\Sigma$}		(yq-);
		\end{tikzpicture}
		\centering
		\begin{tikzpicture}
			%
			% DFA A_{w,1,2}
			
			\node[state, draw=none]					at (-1 * \sHS,  0 * \nVS)			(zn)				{$\mathcal{A}_{abb,1,2}$:};
			\node[state, initial, accepting]		at ( 0 * \sHS,  0 * \nVS)			(zq0)				{$q_0$};
			\node[state, draw=none]					at ( 1 * \sHS,  0 * \nVS)			(zq1)				{};
			\node[state, accepting]					at ( 2 * \sHS,  0 * \nVS)			(zq2)				{$q_2$};
			\node[state, accepting]					at ( 3 * \sHS,  1 * \nVS)			(zq+)				{$q_+$};
			\node[state]							at ( 3 * \sHS, -1 * \nVS)			(zq-)				{$q_-$};
			
			\draw (zq0)	edge[loop above]				node{$a$}			(zq+);
			\draw (zq0)	edge[above]						node{$b$}			(zq2);
			
			%\draw (q1)	edge[below, bend right=12.5]	node{$a$}			(q-);
			%\draw (q1)	edge[above]						node{$b$}			(q2);
			
			\draw (zq2)	edge[above]						node{$a$}			(zq+);
			\draw (zq2)	edge[above]						node{$b$}			(zq-);
			
			\draw (zq+)	edge[loop right]				node{$\Sigma$}		(zq+);
			
			\draw (zq-)	edge[loop right]				node{$\Sigma$}		(zq-);
		\end{tikzpicture}
		\caption{
			Decomposition of the composite minimal linear safety $\symADFAp$ $\mathcal{A}$ from \cref{subfig:aDFAWithMpProp}.
			The three DFAs are constructed according to
			% !!!!!!!!!! VERSIONSTART !!!!!!!!!!
			% !!!!!!!!!! FULL VERSION !!!!!!!!!!
			Definitions \ref{def:aip} and \ref{def:awij}.
			% !!!!!!!!!! CONF VERSION !!!!!!!!!!
			%\cite[Definitions A.3 and A.6]{spenner2026deciding}.
			% !!!!!!!!!! VERSIONEND  !!!!!!!!!!
		}
		\label{subfig:aDFAWithMpPropDecomp}
	\end{subfigure}
	\caption{Minimal linear safety $\symADFAp$s $\mathcal{A}$ and $\mathcal{A}'$ from \cref{exm:mpProp}. Decomposition of $\mathcal{A}$.}
	\label{fig:aDFAsWithAndWithoutMpProp}
\end{figure}

We conclude this discussion of the compositionality of minimal linear safety $\symADFAp$s by using an example to further illustrate the main ideas.
\begin{example}
	\label{exm:mpProp}
	In this example, we consider two simple minimal linear safety $\symADFAp$s and argue that one of them has the \symPPShort, while the other has not. Further, we outline how this leads to the one being composite and the other being prime.
	
	Consider the two DFAs $\mathcal{A} = (Q,\Sigma,q_0,\delta,F)$ and $\mathcal{A}' = (Q',\Sigma,q_0,\delta',F')$ depicted in \cref{subfig:aDFAWithMpProp,subfig:aDFAWithoutMpProp}.
	It is easy to verify that both DFAs are minimal linear safety $\symADFAp$s and that $\opLinLen{\mathcal{A}} = \opLinLen{\mathcal{A}'} = 2$ and $\opDiffWords{\mathcal{A}} = \opDiffWords{\mathcal{A}'} = \{abb\}$ hold. Thus, for both $\symADFAp$s, their possession of the \symPPShort{}, and with it their compositionality, depends solely on the max-visiting word $w = abb$.
	Note that $\mathcal{A}$ and $\mathcal{A}'$ differ in only one transition: We have $\delta(q_0,b) = q_2$ but $\delta'(q_0,b) = q_+$. Still, we will see that, with this difference, $\mathcal{A}$ has the \symPPShort{} and is composite, while $\mathcal{A}'$ does not have the \symPPShort{} and is prime.
	
	First, we consider $\mathcal{A}$ and argue that the \symPPShort-condition holds for $w=abb,i=1,j=2$.
	For $l=0$, we have $\opPump{w}{1,2;l} = a^lbb = a^0bb = bb \notin \opLang{\mathcal{A}}$. Here, the transition $\delta(q_0,b) = q_2$ (in contrast to $\delta'(q_0,b) = q_+$) is crucial.
	For $l=1$, we trivially have $\opPump{w}{1,2;l} = a^lbb = a^1bb = w \notin \opLang{\mathcal{A}}$.
	And lastly, for $l\geq 2$, we have $\opPump{w}{1,2;l} = a^lbb \notin \opLang{\mathcal{A}}$ due to $\delta(q_0,aa) = q_-$.
	Thus, the \symPPShort-condition holds for $w,i=1,j=2$. Since $\opDiffWords{\mathcal{A}} = \{w\}$, this immediately implies that $\mathcal{A}$ has the \symPPShort.
	
	With \cref{the:pmADFAMinLinSafeCompositionality}, $\mathcal{A}$ having the \symPPShort{} implies that it is composite.
	% !!!!!!!!!! VERSIONSTART !!!!!!!!!!
	% !!!!!!!!!! FULL VERSION !!!!!!!!!!
	In \cref{subsubsec:compositionalityADFApMLSProofs}, we explicitly state the decomposition of $\symADFAp$s with the \symPPShort{} (\cref{lem:pmADFAMinLinSafeWithPPDecomposition}) and outline all necessary DFAs (Definitions \ref{def:aip} and \ref{def:awij} and \cref{fig:aipAwij}).
	% !!!!!!!!!! CONF VERSION !!!!!!!!!!
	%In the full version, we explicitly state the decomposition of $\symADFAp$s with the \symPPShort{} \cite[Lemma A.8]{spenner2026deciding} and outline all necessary DFAs \cite[Definitions A.3 and A.6 and Figure 4]{spenner2026deciding}.
	% !!!!!!!!!! VERSIONEND  !!!!!!!!!!
	Here, we discuss the decomposition of the specific $\symADFAp$ $\mathcal{A}$ to motivate how its compositionality is connected to the \symPPShort.
	
	The decomposition of $\mathcal{A}$ consists of the three simple DFAs $\mathcal{A}_1^+, \mathcal{A}_2^{+}$ and $\mathcal{A}_{w,1,2}$ depicted in \cref{subfig:aDFAWithMpPropDecomp}.
	In this decomposition, $\mathcal{A}_{w,1,2}$ is the only DFA exploiting the \symPPShort.
	It is used to reject the max-visiting word $w$ (and its extensions).
	Basically, it omits the state $q_1$ and instead adds a self-loop $q_0 \xrightarrow{a} q_0$. For this to be allowed, it is crucial that $\opPump{w}{1,2;l} = a^lbb \notin \opLang{\mathcal{A}}$ holds for every $l \in \symNatNum$, that is, that the \symPPShort-condition holds for $w,i=1,j=2$.
	Further, $\mathcal{A}_1^+$ and $\mathcal{A}_2^+$ reject the words rejected by $\mathcal{A}$ on which $\mathcal{A}$ does not enter $q_1$ or $q_2$, respectively. They are constructed out of $\mathcal{A}$ by omitting the respective state $q_i$ and redirecting all in-$q_i$-transition into $q_+$.
	
	Indeed, in this instance, we have $\opLang{\mathcal{A}_{w,1,2}} \subseteq \opLang{\mathcal{A}_1^+}$, meaning that $\mathcal{A}_1^+$ can be omitted in the decomposition of $\mathcal{A}$.
	This is due to the simplicity of $\mathcal{A}$:
	For example, simply by adding the letter $c$ and the transitions $\delta(q_0,c) = q_2$ and $\delta(q_1,c) = \delta(q_2,c) = q_+$ to $\mathcal{A}$, we get a still composite DFA in whose decomposition all three DFAs of
	% !!!!!!!!!! VERSIONSTART !!!!!!!!!!
	% !!!!!!!!!! FULL VERSION !!!!!!!!!!
	\cref{lem:pmADFAMinLinSafeWithPPDecomposition}
	% !!!!!!!!!! CONF VERSION !!!!!!!!!!
	%\cite[Lemma A.8]{spenner2026deciding}
	% !!!!!!!!!! VERSIONEND  !!!!!!!!!!	
	are necessary.
	
	Next, we consider $\mathcal{A}'$ and argue that $w=abb$ breaks the \symPPShort.
	This follows from the simple observation that $\opPump{w}{i,j;0} \in \opLang{\mathcal{A}'}$ holds for all $i,j \in \{1,2,3\},i<j$.
	For $i=1,j=2$, we have $\opPump{w}{1,2;0} = a^0bb = bb \in \opLang{\mathcal{A}'}$. Here, the transition $\delta'(q_0,b) = q_+$ (in contrast to $\delta(q_0,b) = q_2$) is crucial.
	For $i=1,j=3$, we have $\opPump{w}{1,3;0} = (ab)^0b = b \in \opLang{\mathcal{A}'}$.
	And lastly, for $i=2,j=3$, we have $\opPump{w}{2,3;0} = ab^0b = ab \in \opLang{\mathcal{A}'}$.
	In total, there is no pair $i,j \in \{1,2,3\},i<j$ so that the \symPPShort-condition holds for $w,i,j$, meaning that $w$ breaks the \symPPShort.
	
	With \cref{the:pmADFAMinLinSafeCompositionality}, $\mathcal{A}'$ not having the \symPPShort{} implies that it is prime.
	% !!!!!!!!!! VERSIONSTART !!!!!!!!!!
	% !!!!!!!!!! FULL VERSION !!!!!!!!!!
	In \cref{subsubsec:compositionalityADFApMLSProofs},
	% !!!!!!!!!! CONF VERSION !!!!!!!!!!
	%In the full version \cite[Appendix A.2.1]{spenner2026deciding},
	% !!!!!!!!!! VERSIONEND  !!!!!!!!!!
	we present the complete proof that,
	for an $\symADFAp$ without the \symPPShort, a max-visiting word breaking the \symPPShort{} is a primality witness.
	Here, we discuss why $w$ is a primality witness of the specific $\symADFAp$ $\mathcal{A}'$.
	
	For $\mathcal{A}'$ to be composite, there would have to be a $\mathcal{B} \in \opDecompSet{\mathcal{A}'}$ with $w \notin \opLang{\mathcal{B}}$. This $\mathcal{B}$ would have to have at least one non-sink less than $\mathcal{A}'$
	% !!!!!!!!!! VERSIONSTART !!!!!!!!!!
	% !!!!!!!!!! FULL VERSION !!!!!!!!!!
	(see \cref{lem:safetyDecomp,lem:acceptingSink})
	% !!!!!!!!!! CONF VERSION !!!!!!!!!!
	%(see full version \cite[Lemmas A.10 and A.11]{spenner2026deciding})
	% !!!!!!!!!! VERSIONEND  !!!!!!!!!!
	and could therefore visit only two distinct states while reading the prefix $ab$ of $w$.
	Thus, $\mathcal{B}$ would "confuse"
	%$ab$ with $\varepsilon$, $a$ or $b$ and consequently
	$w$ with $bb$, $b$ or $ab$.
	But, as we argued above, these three words are all accepted by $\mathcal{A}'$, meaning that every DFA in $\opDecompSet{\mathcal{A}'}$ has to accept them as well.
	Thus, no suitable $\mathcal{B} \in \opDecompSet{\mathcal{A}'}$ with $w \notin \opLang{\mathcal{B}}$ exists, meaning that $w$ is a primality witness of $\mathcal{A}'$.\lipicsEnd
\end{example}

\subsubsection{From the Characterization to \fontComplexityClass{NP}-Hardness}
\label{subsubsec:characterizationToNPHardness}
For \cref{subsubsec:characterizationToNPHardness}, let $\symLinLenVar = \opLinLen{\mathcal{A}_\Phi}$ and $\symCVar = (\symLinLenVar+1) - (r+1)$.

Together, \cref{lem:aPhiIsMLSADFAp,the:pmADFAMinLinSafeCompositionality} imply that $\mathcal{A}_\Phi$ is prime if and only if it does not have the \symPPShort. Following \cref{def:pP}, this means that exactly the max-visiting words of $\mathcal{A}_\Phi$ are relevant for the compositionality of $\mathcal{A}_\Phi$. How these look like and where they can potentially be pumped is then answered by:

\begin{restatable}{lemma}{lemaPhi}
	\label{lem:aPhi}\RestateRemark
	The following assertions hold:
	\begin{romanenumerate}
		\item $\opDiffWords{\mathcal{A}_\Phi} = \{u d c^\symCVar \mid u \in \{0,1\}^r\}$.
		\item For every $w \in \opDiffWords{\mathcal{A}_\Phi}$ and every $i, j \in \{1, \dots, \symLinLenVar+1\}, i < j$, we have: If the \symPPShort-condition holds for $w,i,j$, then $i = 1$ and $j = r+1$. \lipicsEnd
	\end{romanenumerate}
\end{restatable}

\cref{lem:aPhi}~(i) formalizes the observation about the form of the max-visiting words of $\mathcal{A}_\Phi$ which we made in \cref{subsec:constructionAPhi} and which we already revisited when discussing \cref{lem:aPhiIsMLSADFAp}. In particular, note that, in every state of the formula device (except $\hat{p}_r^s$), reading $0$ or $1$ results in "skipping over" at least one state, while reading $d$ results in entering the rejecting sink. Thus, the max-vising words of $\mathcal{A}_\Phi$ indeed have to be suffixed by $c^\symCVar$.

The argumentation behind (ii) is somewhat technical.
For all $w = udc^\symCVar \in \opDiffWords{\mathcal{A}_\Phi}$ and all suitable $i,j$ with $i \neq 1$ or $j \neq r+1$, we can show $\delta_\Phi(p_0,\opPump{w}{i,j;0}) \neq p_-$, meaning that it suffices to consider the pumpings with $l=0$ to prove (ii).
For the simple case $i > r+1$, we have $\opPump{w}{i,j;0} = udc^x$ for some $x < \symCVar$, which is akin to removing some $c$-suffix of $w$. Clearly, this means that the initial run of $\mathcal{A}_\Phi$ on $\opPump{w}{i,j;0}$ is identical to the one on $w$ until it ends somewhere in the formula device, without reaching the rejecting sink.
For the other cases, that is, for $1 < i \leq r+1$ as well as for $i = 1$ and $j \neq r+1$, the key observation is that, in $\opPump{w}{i,j;0}$, some part of the prefix $ud$ is removed so that a $c$ or a $d$ is read "out of place". This means that a $c$ is read in one of the states $p_0,\dots,p_r$, or a $d$ is read in one of the states $p_1,\dots,p_{r-1}$, which results in entering $p_+$.

With \cref{lem:aPhi}~(i), the form of the max-vising words of $\mathcal{A}_\Phi$ is highly restricted:
As mentioned, they consist of an assignment string followed by $dc^\symCVar$. Thus, every assignment string $u \in \{0,1\}^r$ corresponds to exactly one max-visiting word in $\opDiffWords{\mathcal{A}_\Phi}$, namely, $udc^\symCVar$, and vice versa.
Then (ii) greatly restricts where pumping positions for the max-visiting words can lie that could witness that $\mathcal{A}_\Phi$ has the \symPPShort. For every $w \in \opDiffWords{\mathcal{A}_\Phi}$, only $i=1,j=r+1$ are possible candidates. This means that $\mathcal{A}_\Phi$ does not have the \symPPShort, and is therefore prime, if and only if $\delta_\Phi(p_0,\opPump{w}{1,r+1;l}) \neq p_-$ for some $w \in \opDiffWords{\mathcal{A}_\Phi}$ and $l \in \symNatNum$. And because the words in $\opDiffWords{\mathcal{A}_\Phi}$ have the form $udc^\symCVar$ for $u \in \{0,1\}^r$, the relevant pumpings are of the form $u^ldc^\symCVar$, meaning that we pump, and thereby repeat, exactly the assignment strings.

We have established that only the pumpings of the form $u^ldc^\symCVar$ are relevant for the compositionality of $\mathcal{A}_\Phi$.
With the following lemma, we finally cement the connection between the compositionality of $\mathcal{A}_\Phi$ and the satisfiability of $\Phi$, which was already perceivable in \cref{lem:aPhiClauseRows}.

\begin{restatable}{lemma}{lemaPhiOnerpOnePumpings}
	\label{lem:aPhiOnerpOnePumpings}\RestateRemark
	Let $u \in \{0,1\}^r$ be an assignment string and let $w = udc^\symCVar \in \opDiffWords{\mathcal{A}_\Phi}$ be the corresponding max-visiting word.
	Then the \symPPShort-condition does not hold for $w,i=1,j=r+1$ if and only if the assignment $\gamma_u$ induced by $u$ satisfies $\Phi$. \lipicsEnd
\end{restatable}

As \cref{lem:aPhiOnerpOnePumpings} connects the compositionality of $\mathcal{A}_\Phi$ and the satisfiability of $\Phi$, it is central for the correctness of our reduction. Also, the main ideas behind the design of $\mathcal{A}_\Phi$ become apparent in its proof.
Therefore, we briefly argue that it holds.
Let $u \in \{0,1\}^r$ be an assignment string and let $w = udc^\symCVar$ be the corresponding max-visiting word.

First, towards a proof by contraposition, we assume that the assignment $\gamma_u$ induced by $u$ does not satisfy $\Phi$.
We argue that the \symPPShort-condition holds for $w,i=1,j=r+1$.
Because $\gamma_u$ does not satisfy $\Phi$, there is a $k \in \{1,\dots,s\}$ so that $\gamma_u$ does not satisfy the $k$-th clause of $\Phi$. Let $k$ be the minimum value for which this holds.

The key insight here is that, firstly,
for all $1<l<k+1$, the pumping $\opPump{w}{1,r+1;l}$ is rejected because a $d$ is read "out of place", namely, in the state $\hat{p}_r^{l-1}$;
and that, secondly,
for all $l \geq k+1$, the pumping $\opPump{w}{1,r+1;l}$ is rejected because, after traversing the $k$-th clause row on $u$, the state $p_r^k$ is reached, from which every continuation with a letter $\sigma \in \{0,1,d\}$ leads to $p_-$.

To be more precise,
for all $1<l<k+1$,
we have $\delta_\Phi(p_0,u^l) = \hat{p}_r^{l-1}$ and therefore $\delta_\Phi(p_0,\opPump{w}{1,r+1;l}) = p_-$,
with \cref{lem:aPhiClauseRows}.
This holds because $\gamma_u$ satisfies the first $k-1$ clauses, meaning that, if traversing any of these clause rows on $u$, then $\mathcal{A}_\Phi$ ends up in a $\hat{p}$-state, from which $p_-$ is entered on reading the subsequent $d$.
Further,
for all $l \geq k+1$,
we have $\delta_\Phi(p_0,u^{k+1}) = p_r^k$ and therefore $\delta_\Phi(p_0,\opPump{w}{1,r+1;l}) = p_-$, again with \cref{lem:aPhiClauseRows}.
This holds because $\gamma_u$ does not satisfy the $k$-th clause, meaning that, if traversing the $k$-th clause row on $u$, then $\mathcal{A}_\Phi$ ends up in the state $p_r^k$, from which $p_-$ is entered on reading any of the possible subsequent letters $0$, $1$ or $d$.
Lastly, $\delta_\Phi(p_0,\opPump{w}{1,r+1;0}) = p_-$ and $\delta_\Phi(p_0,\opPump{w}{1,r+1;1}) = p_-$ hold obviously.
In total, we have $\delta_\Phi(p_0,\opPump{w}{1,r+1;l}) = p_-$ for all $l \in \symNatNum$.
So the \symPPShort-condition holds for $w,i=1,j=r+1$.

Second, we assume that the assignment $\gamma_u$ induced by $u$ satisfies $\Phi$.
We argue that the \symPPShort-condition does not hold for $w,i=1,j=r+1$.

The key insight here is that, for a satisfying assignment $\gamma$ and a sufficiently large $l$, $\mathcal{A}_\Phi$ reaches $\hat{p}_r^s$ on reading $u_\gamma^l$.
The reason is that, because $\gamma$ satisfies every clause, $\mathcal{A}_\Phi$ enters the $\hat{p}$-states in every clause row and therefore always advances to the next clause row, finally reaching $\hat{p}_r^s$. From there, the transitions are such that $p_+$ is entered.

More precisely, with \cref{lem:aPhiClauseRows}, we have $\delta_\Phi(p_0,u^{s+1}) = \hat{p}_r^s$. The argument is similar to above, only that now the assignment satisfies all clauses, not just the clauses up to some $k<s$. With $\delta_\Phi(\hat{p}_r^s,\sigma) = p_+$ for $\sigma \in \{0,1\}$, we then immediately have $\delta_\Phi(p_0,\opPump{w}{1,r+1;s+2}) = p_+ \neq p_-$. Thus, $l = s+2$, and indeed every $l \geq s+2$, witnesses that there is an $l \in \symNatNum$ with $\delta_\Phi(p_0,\opPump{w}{1,r+1;l}) = \delta_\Phi(p_0,u^ldc^\symCVar) \neq p_-$.
So the \symPPShort-condition does not hold for $w,i=1,j=r+1$.

In total, we have argued that, for every $w = u dc^\symCVar \in \opDiffWords{\mathcal{A}_\Phi}$, the \symPPShort-condition does not hold for $w,i=1,j=r+1$ if and only if the assignment $\gamma_u$ induced by $u$ satisfies $\Phi$.
Thus, \cref{lem:aPhiOnerpOnePumpings} holds.

With \cref{the:pmADFAMinLinSafeCompositionality,lem:aPhiIsMLSADFAp,lem:aPhi,lem:aPhiOnerpOnePumpings}, we have everything we need to prove the correctness of our reduction. That is, we have everything we need to prove \cref{lem:aPhiCharacterization}, that $\Phi$ is satisfiable if and only if $\mathcal{A}_\Phi$ is prime.

\begin{proof}[Proof of \cref{lem:aPhiCharacterization}]
	With \cref{lem:aPhiIsMLSADFAp}, the \symCNFDFA{} $\mathcal{A}_\Phi$ is a minimal linear safety $\symADFAp$.
	Together with \cref{the:pmADFAMinLinSafeCompositionality}, this implies that $\mathcal{A}_\Phi$ is prime if and only if it does not have the \symPPShort.
	Thus, we only have to prove that $\Phi$ is satisfiable if and only if $\mathcal{A}_\Phi$ does not have the \symPPShort{}.
	
	First, towards a proof by contraposition, let $\Phi$ be unsatisfiable. Let $w \in \opDiffWords{\mathcal{A}_\Phi}$. With \cref{lem:aPhi}~(i), we have $w = u d c^\symCVar$, where $u \in \{0,1\}^r$ is an assignment string.
	Because $\Phi$ is unsatisfiable, the assignment $\gamma_u$ induced by $u$ does not satisfy $\Phi$.
	With with \cref{lem:aPhiOnerpOnePumpings}, we therefore have that the \symPPShort{}-condition holds for $w,i=1,j=r+1$.
	Thus, the \symPPShort{}-condition holds for every max-visiting word, meaning that $\mathcal{A}_\Phi$ has the \symPPShort.
	
	Now let $\Phi$ be satisfiable. Let $\gamma$ be a satisfying assignment. We consider the assignment string $u_\gamma$ induced by $\gamma$ and the corresponding word $w = u_\gamma d c^\symCVar$.
	With \cref{lem:aPhi}~(i), we have $w \in \opDiffWords{\mathcal{A}_\Phi}$.
	And with \cref{lem:aPhi}~(ii), the \symPPShort{}-condition does not hold for $w,i,j$ where $i \neq 1$ or $j \neq r+1$. 
	Further, because $\gamma$ satisfies $\Phi$, we have, with \cref{lem:aPhiOnerpOnePumpings}, that the \symPPShort-condition does not hold for $w,i=1,j=r+1$ either. Thus, $w$ witnesses that $\mathcal{A}_\Phi$ does not have the \symPPShort.
	
	We have proven that $\Phi$ is satisfiable if and only if $\mathcal{A}_\Phi$ does not have the \symPPShort{}. As outlined, this implies that $\Phi$ is satisfiable if and only if $\mathcal{A}_\Phi$ is prime. The proof is complete.
\end{proof}

With \cref{lem:aPhiCharacterization}, we have proven that our reduction is correct.
Because the construction of $\mathcal{A}_\Phi$ is clearly possible in polytime, we then immediately get \cref{the:primeDFANPHard}, the \fontComplexityClass{NP}-hardness of \probPrimeDFA{}.

\subsection{\fontComplexityClass{NP}-Completeness of $\probPrimeADFApMLS{}$}
\label{subsec:aDFApMLSNPComplete}
In \cref{subsec:constructionAPhi,subsec:correctnessAPhi}, we improved the lower complexity bound of the general problem \probPrimeDFA{}.
Now we conclude \cref{sec:npHardnessShortV2} by arguing that this improved lower bound is tight for the restriction of \probPrimeDFA{} to minimal linear safety $\symADFAp$s.

We denote the restriction of \probPrimeDFA{} to DFAs whose respective minimal DFA is a minimal linear safety $\symADFAp$ by $\probPrimeADFApMLS{}$.

The restriction of \probPrimeDFA{} to \symADFA{}s, so to truly acyclic DFAs, was studied in \cite{spenner2023decomposing}. Motivated by this, we study $\probPrimeADFApMLS{}$. For a comparison of the compositionality of \symADFA{}s and $\symADFAp$s, we refer to \cref{subsec:compADFAADFApMLS}.
Here, we prove:

\begin{restatable}{theorem}{theprimeADFApmLinSafeNPComplete}
	\label{the:primeADFApmLinSafeNPComplete}\RestateRemark
	The problem $\probPrimeADFApMLS{}$ is \fontComplexityClass{NP}-complete. \lipicsEnd
\end{restatable}

We proved \cref{the:primeDFANPHard}, the \fontComplexityClass{NP}-hardness of \probPrimeDFA{}, by a reduction from \probCNFSAT. At the heart of this reduction was the construction of the \symCNFDFA{} $\mathcal{A}_\Phi$ out of the given CNF-formula $\Phi$. With \cref{lem:aPhiIsMLSADFAp}, every \symCNFDFA{} is also a minimal linear safety $\symADFAp$. Thus, we immediately get:

\begin{restatable}{lemma}{lemprimeADFApmLinSafeNPHard}
	\label{lem:primeADFApmLinSafeNPHard}\RestateRemark
	The problem $\probPrimeADFApMLS{}$ is \fontComplexityClass{NP}-hard. \lipicsEnd
\end{restatable}

To arrive at \cref{the:primeADFApmLinSafeNPComplete}, we additionally prove:

\begin{restatable}{lemma}{lemprimeADFApmLinSafeInNP}
	\label{lem:primeADFApmLinSafeInNP}\RestateRemark
	The problem $\probPrimeADFApMLS{}$ is in \fontComplexityClass{NP}. \lipicsEnd
\end{restatable}

To prove that $\probPrimeADFApMLS{}$ is in \fontComplexityClass{NP}, we describe a guess-and-verify \fontComplexityClass{NP}-algorithm for $\probPrimeADFApMLS{}$ which exploits \cref{the:pmADFAMinLinSafeCompositionality}.
Given a minimal linear safety $\symADFAp$ $\mathcal{A}$, the algorithm guesses a max-visiting word $w \in \opDiffWords{\mathcal{A}}$ and then verifies that $w$ breaks the \symPPShort, meaning that $w$ witnesses that $\mathcal{A}$ does not have the \symPPShort.
To do so, it simulates the initial run of $\mathcal{A}$ on $\opPump{w}{i,j;l}$ for every pair $i,j \in \{1,\dots,\opLinLen{\mathcal{A}}+1\},i<j$ and every relevant $l$. The key insight here is that, for every such pair $i,j$, there is a polynomially bounded value $l_{i,j}'$ such that only the values $l \leq l_{i,j}'$ are relevant. Thus, the algorithm has to simulate only polynomially many runs.
If the algorithm finds that $w$ indeed breaks the \symPPShort, it accepts.
Else, it rejects.

The polynomial runtime of the algorithm follows immediately from the key insight just mentioned. The runtime is dominated by the simulation of the runs. There are only polynomially many pairs $i,j$. And for each pair, there are only polynomially many relevant pumpings $\opPump{w}{i,j;l}$, each having polynomial length. Thus, the algorithm has polynomial runtime.

Regarding the correctness, if $\mathcal{A}$ is prime, then there is a word breaking the \symPPShort. The algorithm guesses this word and subsequently accepts.
Else, that is, if $\mathcal{A}$ is composite, then no word breaking the \symPPShort{} exists. Thus, the algorithm is forced to reject no matter which word it guesses.
In total, the algorithm works correctly for $\probPrimeADFApMLS{}$.

To summarize, we have described an \fontComplexityClass{NP}-algorithm for $\probPrimeADFApMLS{}$. Thus, we have proven \cref{lem:primeADFApmLinSafeInNP}. Together with the already established \cref{lem:primeADFApmLinSafeNPHard}, we immediately get \cref{the:primeADFApmLinSafeNPComplete}, the \fontComplexityClass{NP}-completeness of $\probPrimeADFApMLS{}$.

%% file: discussion.tex
\section{Discussion}
\label{sec:discussion}
We studied the primality of DFAs and thereby of regular languages. We proved the \fontComplexityClass{NP}-hardness of \probPrimeDFA{}.
This is the first improvement of a complexity bound for \probPrimeDFA{} since \cite{kupferman2015prime} introduced \probPrimeDFA{} and proved that it is \fontComplexityClass{NL}-hard and in \fontComplexityClass{ExpSpace}.

We proved the \fontComplexityClass{NP}-hardness of \probPrimeDFA{} by a reduction from \probCNFSAT. For a given CNF-formula $\Phi$, the reduction yields a so-called \symCNFDFA{} $\mathcal{A}_\Phi$ so that $\Phi$ is satisfiable if and only if $\mathcal{A}_\Phi$ is prime. 
To link the satisfiability of $\Phi$ and the compositionality of $\mathcal{A}_\Phi$, 
the construction of $\mathcal{A}_\Phi$ ensures that
(1) $\Phi$ is encoded in $\mathcal{A}_\Phi$, and
(2) the variable assignments for $\Phi$ are encoded in the words relevant for the compositionality of $\mathcal{A}_\Phi$.

Central to the \fontComplexityClass{NP}-hardness proof was the characterization of the compositionality of \symCNFDFA{}s. However, only certain structural aspects of \symCNFDFA{}s were relevant for this characterization. Thus, we actually characterized the compositionality of the more general minimal linear safety $\symADFAp$s. Exploiting this characterization, we then proved the \fontComplexityClass{NP}-completeness of $\probPrimeADFApMLS{}$, the restriction of \probPrimeDFA{} to DFAs whose respective minimal DFA is a minimal linear safety $\symADFAp$.

To wrap up this paper, we briefly discuss two points.

Firstly, we compare the complexity of $\probPrimeADFApMLS{}$ and $\probPrimeDFAfin{}$, the restriction of \probPrimeDFA{} to DFAs deciding finite languages.
And secondly, we highlight that the \fontComplexityClass{NP}-hardness of \probPrimeDFA{} carries over to \probSPrimeDFA{}, which is a variant of the problem based on a slightly different notion of compositionality.

\subsection{$\probPrimeADFApMLS{}$ and $\probPrimeDFAfin{}$}
\label{subsec:compADFAADFApMLS}
We denote the restriction of \probPrimeDFA{} to DFAs deciding a finite language by $\probPrimeDFAfin{}$.
Note that a DFA decides a finite language if and only if its minimal DFA is an \symADFA.
Thus, the problems $\probPrimeDFAfin{}$ and $\probPrimeADFApMLS{}$ are very similar: $\probPrimeDFAfin{}$ considers \symADFA{}s, so truly acyclic DFAs, while $\probPrimeADFApMLS{}$ considers $\symADFAp$s, so "almost" acyclic DFAs possessing an accepting sink.

In this paper, we have proven that $\probPrimeADFApMLS{}$ is \fontComplexityClass{NP}-complete.
With \cite[Theorem 4.1]{spenner2023decomposing}, $\probPrimeDFAfin{}$ is \fontComplexityClass{NL}-complete.
Therefore, the simple addition of an accepting sink leads to a jump in the complexity of the primality problem (assuming $\fontComplexityClass{NL} \subset \fontComplexityClass{NP}$).
We give an intuition for this somewhat surprising fact.

The reduction from \probCNFSAT{} crucially hinges on the fact that the \symCNFDFA{} $\mathcal{A}_\Phi$ can potentially enter the accepting sink after reading the prefix $u^{s+2}$ of $\opPump{w}{1,r+1;s+2}$, where $w = u d c^\symCVar \in \opDiffWords{\mathcal{A}_\Phi}$ is a max-visiting word. In other words, it can enter the accepting sink after "pumping up" $w$ so that the assignment string $u$ is checked against every clause row. Once in the accepting sink, the suffix $d c^\symCVar$ of dummy letters can be read and the pumping $\opPump{w}{1,r+1;s+2}$ is accepted.

No construction of this kind is possible for \symADFA{}s. Every \symADFA{} has a threshold value so that it rejects every word longer than that threshold value. This is obvious, because for words longer than the threshold value, it runs out of non-sinks and necessarily enters the rejecting sink.
Therefore, "pumping up" and thereby lengthening the \symADFA{}-equivalent of a max-visiting word necessarily leads to an overlong word that is rejected.

Thus, with the existence of an accepting sink, no threshold value exists so that words longer than that threshold value are rejected. This simple additional complication is enough to lead to the jump in complexity, from \fontComplexityClass{NL}-completeness for $\probPrimeDFAfin{}$ to \fontComplexityClass{NP}-completeness for $\probPrimeADFApMLS{}$.

\subsection{\fontComplexityClass{NP}-Hardness of \probSPrimeDFA{}}
With \cref{def:compositionality}, we followed \cite{kupferman2015prime} and defined compositionality using the index of the given DFA.
However, without stating so explicitly, \cite{jecker2020unary,jecker2021decomposing} employed a slightly different definition, using the size instead of the index.
In \cite{spenner2023decomposing}, this difference was made explicit by introducing the notion of S-compositionality and the respective decision problem.
We follow \cite{spenner2023decomposing} and define:

\begin{definitionRoman}
	A DFA $\mathcal{A}$ is \highlightDef{S-composite} if it is $(\opSize{\mathcal{A}}-1)$-decomposable.
	Otherwise, it is \highlightDef{S-prime}.\lipicsEnd
\end{definitionRoman}
We denote the problem of deciding S-primality for a given DFA by \probSPrimeDFA{}.

Many results known for compositionality carry over trivially to S-compositionality.
In particular, \probPrimeDFA{} is in \fontComplexityClass{ExpSpace} \cite[Theorem 2.4]{kupferman2015prime}. The proof of this result has to be adapted only slightly to prove that \probSPrimeDFA{} is in \fontComplexityClass{ExpSpace} as well \cite[Theorem 6.3]{spenner2023decomposing}.

However, the \fontComplexityClass{NL}-hardness of \probPrimeDFA{}, established in \cite[Theorem 2.5]{kupferman2015prime}, does not carry over trivially to \probSPrimeDFA{}. The \fontComplexityClass{NL}-hardness of \probSPrimeDFA{} is established in \cite[Theorem 6.3]{spenner2023decomposing} by using a reduction of a different problem than the one used for \probPrimeDFA{}.

Here, we highlight that the improved lower complexity bound of \probPrimeDFA{} does carry over trivially to \probSPrimeDFA{}.
We established the \fontComplexityClass{NP}-hardness of \probPrimeDFA{} by a reduction from \probCNFSAT{}. The \symCNFDFA{}s yielded by this reduction are minimal. Thus, with the same reduction, we get:

\begin{restatable}{theorem}{thesPrimeDFANPHard}
	\label{the:sPrimeDFANPHard}\RestateRemark
	The problem \probSPrimeDFA{} is \fontComplexityClass{NP}-hard. \lipicsEnd
\end{restatable}

Indeed, it is well known that DFA minimization can be done in polytime. Thus, by making the transition from \fontComplexityClass{NL}-hardness to \fontComplexityClass{NP}-hardness, the distinction between compositionality and S-compositionality loses its edge because we can always minimize the given DFA. This simple argument is an alternative way to arrive at \cref{the:sPrimeDFANPHard}.

%% file: npHardness.tex
\section{\fontComplexityClass{NP}-Hardness of \probPrimeDFA{} – Proofs}
\label{sec:npHardness}
In \cref{sec:npHardnessShortV2}, we outlined the main ideas behind the \fontComplexityClass{NP}-hardness proof for \probPrimeDFA{}. Here, we present the full argument with complete formal proofs. To do so, we introduce some additional lemmas.

As in \cref{sec:npHardnessShortV2}, the goal is to prove:

\theprimeDFANPHard*

The structure of \cref{sec:npHardness} mirrors that of \cref{sec:npHardnessShortV2}.
In \cref{subsec:constructionAPhiProofs} (corresponding to \cref{subsec:constructionAPhi}), we specify the reduction by defining the \symCNFDFA{}s yielded by the reduction.
In \cref{subsec:correctnessAPhiProofs} (corresponding to \cref{subsec:correctnessAPhi}), we prove the correctness of the reduction.
To do so, we characterize the compositionality of $\symADFAp$s and thereby of \symCNFDFA{}s in \cref{subsubsec:compositionalityADFApMLSProofs} (corresponding to \cref{subsubsec:compositionalityADFApMLS}).
Then we exploit this characterization to prove the correctness in \cref{subsubsec:characterizationToNPHardnessProofs} (corresponding to \cref{subsubsec:characterizationToNPHardness}).
To conclude, we prove, in \cref{subsec:aDFApMLSNPCompleteProofs} (corresponding to \cref{subsec:aDFApMLSNPComplete}), that the restriction of \probPrimeDFA{} to the DFAs relevant for the reduction is \fontComplexityClass{NP}-complete.

\subsection{Construction of $\mathcal{A}_\Phi$ – Proofs}
\label{subsec:constructionAPhiProofs}
We use the same notation and the same assumptions as in \cref{sec:npHardnessShortV2}. For the sake of readability, we restate them here.

For the remainder of \cref{sec:npHardness}, let $X = \{x_1, \dots, x_r\}$ be a set of $r \in \symNatNumGeq{1}$ numbered variables, and let $\Phi$ be a CNF-formula over $X$. 
W.l.o.g.\ we assume that $\Phi$ has at least one clause and that every variable appears at most once in every clause. Clearly, these assumptions are permissible.

To simplify the construction of $\mathcal{A}_\Phi$, we fix a special notation for $\Phi$.
Namely, we write:
\begin{align*}
	\Phi = (\symElem_1^1 \vee \dots \vee \symElem_r^1) \wedge \dots \wedge (\symElem_1^s \vee \dots \vee \symElem_r^s),
\end{align*}
where $s \in \symNatNumGeq{1}$ is the number of clauses of $\Phi$, and $\symElem_i^k \in \{\neg x_i, x_i, \bot\}$ is the $i$-th \highlightDef{element} in the $k$-th clause where $k \in \{1, \dots, s\}, i \in \{1, \dots, r\}$.
In other words, in our notation, every clause has exactly $r$ elements and the $i$-th element of every clause is either a literal of the $i$-th variable $x_i$ or $\bot$.
Here, $\bot$ is a special symbol that always evaluates to \emph{False}.

Given a CNF-formula with at least one clause in which every variable appears at most once in every clause, it is straightforward to construct an equivalent CNF-formula that adheres to this special notation. We can simply order the variables in each clause and include $\bot$'s for "missing" variables. Thus, it is possible to write $\Phi$ in this way.

Towards the definition of $\mathcal{A}_\Phi$, we introduce some additional notation.
Let $\Sigma = \{0,1,c,d\}$.
We call a string $u \in \{0,1\}^r$ an \highlightDef{assignment string}.
As usual, an assignment over $X$ is a function $\gamma: X \rightarrow \{0,1\}$.
An assignment $\gamma$ over $X$ induces an assignment string $u_\gamma \in \{0,1\}^r$ in the following way: $u_\gamma = \gamma(x_1) \dots \gamma(x_r)$. Conversely, an assignment string $u = \sigma_1 \dots \sigma_r \in \{0,1\}^r$ induces an assignment $\gamma_u$ over $X$ in the following way: $\gamma_u(x_i) = \sigma_i$ for all $i \in \{1, \dots, r\}$.

We consider $\mathcal{A}_\Phi = (Q_\Phi, \Sigma, p_0, \delta_\Phi, F_\Phi)$, which is outlined in \cref{fig:aPhiaPhiXTransitions} and formally defined as follows:

\begin{definitionRoman} % def:aPhi
	\label{def:aPhi}
	To simplify the definition, we use the following notation: For every $x \in X$, let $\opLitToBin{x} = 1$ and $\opLitToBin{\neg x} = 0$.
	We define the \highlightDef{\symCNFDFA{}} $\mathcal{A}_\Phi = (Q_\Phi, \Sigma, p_0, \delta_\Phi, F_\Phi)$ induced by $\Phi$ as:
	\begin{align*}
		Q_\Phi 	& = \{p_0, \dots, p_r\} \cup \{p_c^0\} \cup \{p_i^k, \hat{p}_i^k, p_c^k \mid k \in \{1, \dots, s\} \wedge i \in \{1, \dots, r\}\} \cup \{p_+, p_-\}, \\
		F_\Phi	& = Q \setminus \{p_-\}, \\
		\delta_\Phi(p_i, \sigma) & =
		\begin{cases}
			p_{i+1}		& \text{ if $i < r$ and $\sigma \in \{0,1\}$} \\
			p_1^1		& \text{ if $i = r$ and $\symElem_1^1 = \bot$ and $\sigma \in \{0,1\}$} \\
			p_1^1		& \text{ if $i = r$ and $\symElem_1^1 \neq \bot$ and $\sigma = 1 - \opLitToBin{\symElem_1^1}$} \\
			\hat{p}_1^1	& \text{ if $i = r$ and $\symElem_1^1 \neq \bot$ and $\sigma = \opLitToBin{\symElem_1^1}$} \\
			p_+			& \text{ if $\sigma = c$} \\
			p_-			& \text{ if $i = 0$ and $\sigma = d$} \\
			p_+			& \text{ if $0 < i < r$ and $\sigma = d$} \\
			p_c^0		& \text{ else, thus if $i = r$ and $\sigma = d$}
		\end{cases} \\
		& \text{for all } i \in \{1, \dots, r\}, \\
		\delta_\Phi(p_i^k, \sigma) & =
		\begin{cases}
			p_{i+1}^k	& \text{ if $i < r$ and $\symElem_{i+1}^k = \bot$ and $\sigma \in \{0,1\}$} \\
			p_{i+1}^k	& \text{ if $i < r$ and $\symElem_{i+1}^k \neq \bot$ and $\sigma = 1 - \opLitToBin{\symElem_{i+1}^k}$} \\
			\hat{p}_{i+1}^k & \text{ if $i < r$ and $\symElem_{i+1}^k \neq \bot$ and $\sigma = \opLitToBin{\symElem_{i+1}^k}$} \\
			p_-			& \text{ if $i = r$ and $\sigma \in \{0,1\}$} \\
			\hat{p}_i^k	& \text{ if $\sigma = c$} \\
			p_-			& \text{ else, thus if $\sigma = d$}
		\end{cases} \\
		& \text{for all } k \in \{1, \dots, s\}, i \in \{1, \dots, r\}, \\
		\delta_\Phi(\hat{p}_i^k, \sigma) & =
		\begin{cases}
			\hat{p}_{i+1}^k	& \text{ if $i < r$ and $\sigma \in \{0,1\}$} \\
			p_1^{k+1}	& \text{ if $i = r$ and $k < s$ and $\symElem_1^{k+1} = \bot$ and $\sigma \in \{0,1\}$} \\
			p_1^{k+1}	& \text{ if $i = r$ and $k < s$ and $\symElem_1^{k+1} \neq \bot$ and $\sigma = 1 - \opLitToBin{\symElem_1^{k+1}}$} \\
			\hat{p}_1^{k+1}	& \text{ if $i = r$ and $k < s$ and $\symElem_1^{k+1} \neq \bot$ and $\sigma = \opLitToBin{\symElem_1^{k+1}}$} \\
			p_+			& \text{ if $i = r$ and $k = s$ and $\sigma \in \{0,1\}$} \\
			p_{i+1}^k	& \text{ if $i < r$ and $\sigma = c$} \\
			p_c^k		& \text{ if $i = r$ and $\sigma = c$} \\
			p_-			& \text{ else, thus if $\sigma = d$}
		\end{cases} \\
		& \text{for all } k \in \{1, \dots, s\}, i \in \{1, \dots, r\}, \\
		\delta_\Phi(p_c^k, \sigma) & =
		\begin{cases}
			p_-			& \text{ if $k < s$ and $\sigma \in \{0,1,d\}$} \\
			p_+			& \text{ if $k = s$ and $\sigma \in \{0,1,d\}$} \\
			p_1^{k+1}	& \text{ if $k < s$ and $\sigma = c$} \\
			p_-			& \text{ else, if $k = s$ and $\sigma = c$}
		\end{cases} \\
		& \text{for all } k \in \{0, \dots, s\}, \\
		\delta_\Phi(p, \sigma) & = p \\
		& \text{for all } p \in \{p_+,p_-\}. \tag*{\lipicsEnd}
	\end{align*}
\end{definitionRoman}

We repeat that we denote the state $p_r$ also by $\hat{p}_r^0$.

Having defined $\mathcal{A}_\Phi$, we now work towards proving that $\Phi$ is satisfiable if and only if $\mathcal{A}_\Phi$ does not have the \symPPShort{} and is therefore prime.

Our goal is to prove:

\lemaPhiCharacterization*

We begin working towards this by proving \cref{lem:aPhiClauseRows}, which states that the behavior of the clause rows is as intended. That is, the clause rows actually encode the clauses of $\Phi$,
so that which of the two end states $\hat{p}_r^k, p_r^k$ is reached after traversing the $k$-th clause row on an assignment string depends on whether the induced assignment satisfies the $k$-th clause or not.
For the sake of readability, we restate the lemma. Afterwards, we give the formal proof.

\lemaPhiClauseRows*

\begin{proof}[Proof of \cref{lem:aPhiClauseRows}]
	Let $u = \sigma_1 \dots \sigma_r, k$ be as required. We inspect the behavior of the $k$-th clause row of $\mathcal{A}_\Phi$ on reading $u$ when starting in $\hat{p}_r^{k-1}$. 
	
	As in \cref{def:aPhi}, we use the following notation: For every $x \in X$, let $\opLitToBin{x} = 1$ and $\opLitToBin{\neg x} = 0$.
	
	Note that, for every $i \in \{1, \dots, r\}$, we have $\delta_\Phi(\hat{p}_r^{k-1}, \sigma_1 \dots \sigma_i) \in \{p_i^k, \hat{p}_i^k\}$. Further, note that, once $\mathcal{A}_\Phi$ has entered the $\hat{p}$-states of the $k$-th clause row when reading $u$, it will not leave them again. These observations can be easily verified by inspecting \cref{def:aPhi,fig:aPhiaPhiXTransitions}.
	
	We prove the two implications separately.
	
	We begin by assuming that $\gamma_u$ satisfies the $k$-th clause of $\Phi$. We argue that this implies $\delta_\Phi(\hat{p}_r^{k-1},u) = \hat{p}_r^k$.
	
	Because $\gamma_u$ satisfies the $k$-th clause of $\Phi$, there is an $i \in \{1, \dots, r\}$ with $\symElem_i^k \neq \bot$ and $\opLitToBin{\symElem_i^k} = \gamma(x_i) = \sigma_i$. If $\mathcal{A}_\Phi$ has already entered the $\hat{p}$-states on $\sigma_1 \dots \sigma_{i-1}$, then it will stay in them. That is, if $\delta_\Phi(\hat{p}_r^{k-1}, \sigma_1 \dots \sigma_{i-1}) = \hat{p}_{i-1}^k$, then $\delta_\Phi(\hat{p}_r^{k-1}, \sigma_1 \dots \sigma_i) = \hat{p}_i^k$. Otherwise, that is, if $\delta_\Phi(\hat{p}_r^{k-1}, \sigma_1 \dots \sigma_{i-1}) = p_{i-1}^k$, we also have $\delta_\Phi(\hat{p}_r^{k-1}, \sigma_1 \dots \sigma_i) = \hat{p}_i^k$ because $\opLitToBin{\symElem_i^k} = \sigma_i$. Thus, we definitely have $\delta_\Phi(\hat{p}_r^{k-1}, \sigma_1 \dots \sigma_i) = \hat{p}_i^k$. This implies $\delta_\Phi(\hat{p}_r^{k-1}, \sigma_1 \dots \sigma_r) = \hat{p}_r^k$. We have proven the first implication.
	
	Next, we assume that $\delta_\Phi(\hat{p}_r^{k-1},u) = \hat{p}_r^k$. We argue that this implies that $\gamma_u$ satisfies the $k$-th clause of $\Phi$. Because $\delta_\Phi(\hat{p}_r^{k-1}, \sigma_1 \dots \sigma_r) \in \{p_r^k, \hat{p}_r^k\}$, arguing for this is sufficient to prove the second implication.
	
	Let $i$ be the minimum value in $\{1, \dots, r\}$ with $\delta_\Phi(\hat{p}_r^{k-1}, \sigma_1 \dots \sigma_i) = \hat{p}_i^k$. Clearly, this implies $\symElem_i^k \neq \bot$ and $\opLitToBin{\symElem_i^k} = \sigma_i$. Thus, we have $\gamma_u(x_i) = \sigma_i = \opLitToBin{\symElem_i^r}$. This directly implies that $\gamma_u$ satisfies the $k$-th clause of $\Phi$. We have proven the second implication.
	
	We have proven both implications. The proof is complete.
\end{proof}

With \cref{def:aPhi}, we have specified the \symCNFDFA{} $\mathcal{A}_\Phi$ induced by the given CNF-formula $\Phi$. And with \cref{lem:aPhiClauseRows}, we have stated a first important characteristic of this $\mathcal{A}_\Phi$. It remains to prove the correctness of the reduction by proving \cref{lem:aPhiCharacterization}.

\subsection{Correctness of the Reduction – Proofs}
\label{subsec:correctnessAPhiProofs}
We prove the correctness of the reduction. To do so, we first characterize the compositionality of minimal linear safety $\symADFAp$s in \cref{subsubsec:compositionalityADFApMLSProofs}. Then we exploit this characterization to prove \cref{lem:aPhiCharacterization} in \cref{subsubsec:characterizationToNPHardnessProofs}.

\subsubsection{Compositionality of Minimal Linear Safety $\symADFAp$s – Proofs}
\label{subsubsec:compositionalityADFApMLSProofs}
For the definitions of $\symADFAp$s, linear DFAs, and safety DFAs, we refer to \cref{subsubsec:compositionalityADFApMLS}. We begin by proving \cref{lem:aDFASizeAndLinearity}, which states fundamental properties regarding the size of $\symADFAp$s.

\lemaDFASizeAndLinearity*

\begin{proof}[Proof of \cref{lem:aDFASizeAndLinearity}]
	Let $\mathcal{A} = (Q, \Sigma, q_I, \delta, F)$ be a minimal $\symADFAp$.
	
	Note that $\mathcal{A}$ is minimal and possesses both an accepting and a rejecting sink. This immediately implies $\opSize{\mathcal{A}} \geq 3$ and that the initial state of $\mathcal{A}$ is not a sink. Consequently, there is at least one word on which $\mathcal{A}$ does not enter a sink, namely, the empty word.
	
	Let $U \subset Q$ be the set of non-sinks of $\mathcal{A}$. As just mentioned, we have $q_I \in U$. Let $\symLinLenVar = \opLinLen{\mathcal{A}}$.
	
	We consider Assertion (i). We argue that $\opSize{\mathcal{A}} \geq \symLinLenVar + 3$ holds. Note that this is equivalent to $\symLinLenVar \leq \opSize{\mathcal{A}} - 3$. To argue this, let $w \in \Sigma^*$ with $w > \opSize{\mathcal{A}} - 3$. We argue that $\delta(q_I, w) \notin U$, meaning that $\delta(q_I, w)$ is a sink. This directly implies $\symLinLenVar \leq \opSize{\mathcal{A}} - 3$.
	
	Consider the initial run of $\mathcal{A}$ on $w$. Because $\mathcal{A}$ is acyclic, every non-sink can appear at most once in this run. With $|U| = \opSize{\mathcal{A}} - 2$ and $w > \opSize{\mathcal{A}} - 3$, this immediately implies that a sink has to appear in this run. Thus, this run ends in a sink, meaning $\delta(q_I, w) \notin U$.
	
	We have argued that $\delta(q_I, w) \notin U$ for every $w \in \Sigma^*$ with $w > \opSize{\mathcal{A}} - 3$. As mentioned, this directly implies $\symLinLenVar \leq \opSize{\mathcal{A}} - 3$. The proof of Assertions (i) is complete.
	
	Next, we consider Assertion (ii). We prove the two directions of the equivalence separately.
	
	We start by assuming that $\opSize{\mathcal{A}} = \symLinLenVar + 3$ holds. We prove that $\mathcal{A}$ is linear.
	
	We begin by introducing the following claim:
	
	\begin{claimWithinProof}
		\label{claWP:qqP}
		Let $q, q' \in U, q \neq q'$. Then either $q'$ is reachable from $q$, or $q$ is reachable from $q'$, but not both.
	\end{claimWithinProof}
	\begin{claimproof}
		Let $w \in \Sigma^\symLinLenVar$ be a word on which $\mathcal{A}$ does not enter a sink. By assumption, we have $|w| = \symLinLenVar = \opSize{\mathcal{A}} - 3$. Additionally, because $\mathcal{A}$ is acyclic, no state appears more than once in the initial run of $\mathcal{A}$ on $w$. Thus, the initial run of $\mathcal{A}$ on $w$ contains $|w| + 1 = \symLinLenVar + 1 = \opSize{\mathcal{A}} - 2$ distinct states, no sinks, and no state appears more than once. Clearly, this means that the initial run of $\mathcal{A}$ on $w$ contains exactly the $\opSize{\mathcal{A}} - 2$ non-sinks of $\mathcal{A}$. In particular, this implies that, for every two non-sinks $q, q' \in U, q \neq q'$, we have that $q'$ is reachable from $q$ or that $q$ is reachable from $q'$. Because $\mathcal{A}$ is acyclic, it can not be that $q'$ is reachable from $q$ and that $q$ is reachable from $q'$. Therefore, we have that, for every two non-sinks $q, q' \in U, q \neq q'$, either $q'$ is reachable from $q$, or $q$ is reachable from $q'$, but not both. The proof of \cref{claWP:qqP} is complete.
	\end{claimproof}
	
	With \cref{claWP:qqP} in hand, we come to the actual proof of the linearity of $\mathcal{A}$. Let $q, q' \in Q, q \neq q'$. We prove that exactly one of the three linearity conditions holds. We make a case distinction. Case 1 is very simple. Case 2 is made simple by \cref{claWP:qqP}.
	\begin{description}
		\item[Case 1: \normalfont{At least one of the states $q, q'$ is a sink.}] W.l.o.g.\ let $q$ be a sink. This clearly implies that $q'$ is unreachable from $q$, meaning that the first linearity condition does not hold.
		
		If $q$ is reachable from $q'$, then the second linearity condition holds and the third linearity condition does not hold. Thus, exactly one of the three linearity conditions holds, namely, the second one.
		
		If $q$ is not reachable from $q'$, then the second linearity condition does not hold and the third linearity condition holds. Thus, exactly one of the three linearity conditions holds, namely, the third one.
		
		We are done with Case 1.
		
		\item[Case 2: \normalfont{Both $q, q'$ are non-sinks.}] With \cref{claWP:qqP}, we have that either $q'$ is reachable from $q$, or $q$ is reachable from $q'$, but not both. This immediately implies that exactly one of the three linearity conditions holds, either the first or the second one. We are done with Case 2.
	\end{description}
	With Cases 1 and 2, exactly one of the three linearity conditions holds. Thus, the $\symADFAp$ $\mathcal{A}$ is linear. We have proven the first direction of the equivalence.
	
	Next, we assume that $\mathcal{A}$ is linear. We prove that $\opSize{\mathcal{A}} = \symLinLenVar + 3$ holds. Note that $\opSize{\mathcal{A}} \geq \symLinLenVar + 3$ holds with Assertion (i). Therefore, we only have to prove $\opSize{\mathcal{A}} \leq \symLinLenVar + 3$. Note that this is equivalent to $\symLinLenVar \geq \opSize{\mathcal{A}} - 3$. To prove $\symLinLenVar \geq \opSize{\mathcal{A}} - 3$, we construct a word of length at least $\opSize{\mathcal{A}} - 3$ on which $\mathcal{A}$ does not enter a sink.
	
	To start, we introduce a further claim:
	\begin{claimWithinProof}
		\label{claWP:u}
		Let $U' \subseteq U$. Then there is a state $q \in U'$ so that every state in $U'$ is reachable from $q$.
	\end{claimWithinProof}
	\begin{claimproof}
		\cref{claWP:u} holds trivially for $|U| = 1$. Thus, we assume $|U| > 1$.
		
		We employ a prove by induction over the size of the subsets. \cref{claWP:u} holds trivially for all subsets $U' \subseteq U$ with $\opSize{U'} \leq 1$. Now let
		%$i \in \{0, \dots, |U|\}, 1 \leq i < |U|$
		$i \in \{1, \dots, |U|-1\}$
		arbitrary but fixed. We assume that \cref{claWP:u} holds for all subsets $U' \subseteq U$ with $\opSize{U'} = i$. We prove that \cref{claWP:u} then holds for all subsets  $U' \subseteq U$ with $\opSize{U'} = i+1$.
		
		Let $U' \subseteq U$ with $|U'| = i+1$. Consider some $U'' \subset U'$ with $|U''| = i$. By assumption, there is a $q \in U''$ so that every state in $U''$ is reachable from $q$. Let $\hat{q} \in U' \setminus U''$. Thus, state $\hat{q}$ is the sole state in $U'$ that is not in $U''$. Because $\mathcal{A}$ is linear and $q,\hat{q} \in U$, either $q$ is reachable from $\hat{q}$, or $\hat{q}$ is reachable from $q$.
		\begin{description}
			\item[Case 1: \normalfont{$q$ is reachable from $\hat{q}$.}]
			By selection of $q$, every state in $U''$ is reachable from $q$. And by assumption, this state $q$ is reachable from $\hat{q}$, which is the sole state in $U' \setminus U''$. Thus, every state in $U'$ is reachable from $\hat{q}$. We are done with Case 1.
			\item[Case 2: \normalfont{$\hat{q}$ is reachable from $q$.}]
			By selection of $q$, every state in $U''$ is reachable from $q$. And by assumption, state $\hat{q}$, which is the sole state in $U' \setminus U''$, is reachable from $q$. Thus, every state in $U'$ is reachable from $q$. We are done with Case 2.
		\end{description}
		With Cases 1 and 2, \cref{claWP:u} holds for $U'$.
		
		We have proven that \cref{claWP:u} holds for all subsets $U' \subseteq U$ with $\opSize{U'} = i+1$ under the assumption that it holds for all subsets $U' \subseteq U$ with $\opSize{U'} = i$. By induction, \cref{claWP:u} holds for all subsets. The proof of \cref{claWP:u} is complete.
	\end{claimproof}
	
	With \cref{claWP:u} in hand, we construct a word of length at least $\symLinLenVar$ on which $\mathcal{A}$ does not enter a sink.
	We do so by using a simple procedure which exploits \cref{claWP:u} to iteratively construct a word $w$ on which $\mathcal{A}$ visits every of its non-sinks.
	Intuitively, throughout the procedure, the set $X$ contains all non-sinks not visited by $\mathcal{A}$ on $w$ as well as the state $x = \delta(q_I,w)$ currently reached by $\mathcal{A}$ on $w$.
	Further, it holds throughout the procedure that every state in $X$ is reachable from $x$. That there is a state in $X$ from which every state in $X$ is reachable is guaranteed by \cref{claWP:u}.
	The procedure is as follows:
	\begin{enumerate}
		\item Initialize $X = U$. As argued above, we have $q_I \in U$. Note that every state in $U$ is reachable from $q_I$ because $\mathcal{A}$ is minimal. Initialize $x = q_I$. Initialize $w = \varepsilon$.
		\item While $\opSize{X} > 1$, consider $X' = X \setminus \{x\}$. With \cref{claWP:u}, there is a state $x' \in X'$ so that every state in $X'$ is reachable from $x'$. Because $x'$ is reachable from $x$, there is a word $v \in \Sigma^+$ with $\delta(x,v) = x'$. Make the following updates: $X \leftarrow X', x \leftarrow x', w \leftarrow w v$.
		\item Return $w$.
	\end{enumerate}
	
	Clearly, the procedure constructs a word $w$ with $|w| \geq \opSize{U} - 1 = \opSize{\mathcal{A}} - 3$ on which $\mathcal{A}$ does not enter a sink. Indeed, the initial run of $\mathcal{A}$ on $w$ ends in the sole state that is not removed from $X$ during the procedure. Thus, word $w$ witnesses $\symLinLenVar \geq \opSize{\mathcal{A}} - 3$.
	
	We have just proven that $\symLinLenVar \geq \opSize{\mathcal{A}} - 3$ holds. We mentioned above that $\symLinLenVar \leq \opSize{\mathcal{A}} - 3$ holds with Assertion (i). In total, we have that $\symLinLenVar = \opSize{\mathcal{A}} - 3$ holds. We have proven the second direction of the equivalence.
	
	In total, we have proven both directions of the equivalence. The proof of Assertion (ii) is complete.
	
	We have proven Assertions (i) and (ii). The proof is complete.
\end{proof}

Our interest in minimal linear safety $\symADFAp$s is motivated by the following result:

\lemaPhiIsMLSADFAp*

For the proof of \cref{lem:aPhiIsMLSADFAp}, we make observations about the structure of $\mathcal{A}_\Phi$. Because the same is true for the proof of \cref{lem:aPhi}, we prove both results together further below in \cref{subsubsec:characterizationToNPHardnessProofs}. For now, we move on to characterize the compositionality of minimal linear safety $\symADFAp$s by proving:

\thepmADFAMinLinSafeCompositionality*

To increase readability, we repeat the definitions of max-visiting words and the \symPPShort{} as well as the alternative formulation of the \symPPShort{}:

\defdifficultWords*

\defpP*

We repeat that, for a word $w = \sigma_1 \dots \sigma_n \in \Sigma^n$ and values $i, j \in \{1,\dots,n+1\}, i < j, l \in \symNatNum$, we define $\opPump{w}{i,j;l} = \sigma_1\dots\sigma_{i-1}(\sigma_i\dots\sigma_{j-1})^l\sigma_{j}\dots\sigma_n$.
We refer to $\opPump{w}{i,j;l}$ as a \highlightDef{pumping of} $w$.

\lempP*

Lastly, we repeat that, for
a minimal linear safety $\symADFAp$ $\mathcal{A}$,
a max-visiting word $w \in \opDiffWords{\mathcal{A}}$
and indices $i,j$,
we say that the \highlightDef{\symPPShort-condition holds for $w,i,j$} if
%$\delta(q_0, \opPump{w}{i,j;l}) = q_-$
$\opPump{w}{i,j;l} \notin \opLang{\mathcal{A}}$
holds for every $l \in \symNatNum$.
Remember that $\opPump{w}{i,j;l} \notin \opLang{\mathcal{A}}$ implies $\opPump{w}{i,j;l}w' \notin \opLang{\mathcal{A}}$ for all words $w'$ because $\mathcal{A}$ is a safety DFA.

We introduce one further term which we will need throughout the remainder of \cref{subsubsec:compositionalityADFApMLSProofs}.
Consider a word $w \in \Sigma^*$.
A word $wv$ with $v \in \Sigma^*$ is an \highlightDef{extension} of $w$.

As we argued in \cref{subsubsec:compositionalityADFApMLS}, one can approach the task of decomposing a minimal linear safety $\symADFAp$ $\mathcal{A}$ by trying to find, for every $w \notin \opLang{\mathcal{A}}$, a DFA $\mathcal{B} \in \opDecompSet{\mathcal{A}}$ with $w \in \opLang{\mathcal{B}}$. As we will see,
this is relatively easy if $w$ is not a max-visiting word, because $\mathcal{A}$ does not visit every non-sink in its initial run on $w$. Thus, only the max-visiting words pose problems.

If a minimal linear safety $\symADFAp$ $\mathcal{A}$ has the \symPPShort, then we can exploit the \symPPShort{} to construct, for each max-visiting word $w \in \opDiffWords{\mathcal{A}}$, a DFA which rejects the extensions of $w$ and which is in $\opDecompSet{\mathcal{A}}$. We discuss this, in addition to the treatment of non-max-visiting words, in Appendix \ref{par:pmADFAMinLinSafeWithPPComposite}.
Otherwise, that is, if $\mathcal{A}$ does not have the \symPPShort, no such DFA exists for the word that breaks the \symPPShort{}, implying that this word is a primality witness of $\mathcal{A}$. We discuss this in Appendix \ref{par:pmADFAMinLinSafeWithoutPPPrime}. There, we also discuss the importance of $\mathcal{A}$ being a safety DFA with an accepting sink. In \cref{subsubsec:compositionalityADFApMLS}, we mentioned this aspect only briefly.

\begin{figure}[t] % fig:pmADFAMinLinSafe
	\centering
	\begin{tikzpicture}[node distance=2.2125cm]
		\footnotesize
		\node[state, initial, accepting] 					(q0) 	{$q_0$};
		\node[state, right of=q0, accepting] 				(q1) 	{$q_1$};
		\node[state, right of=q1, accepting] 				(q2) 	{$q_2$};
		\node[state, right of=q2, draw=none]				(emptyNode) 	{};
		\node[state, right of=emptyNode, accepting] 		(qn-1) 	{$q_{\symLinLenVar-1}$};
		\node[state, right of=qn-1, accepting] 				(qn) 	{$q_\symLinLenVar$};
		\node[state, above right of=qn, accepting] 			(q+) 	{$q_+$};
		\node[state, below right of=qn] 					(q-) 	{$q_-$};
		
		\draw	(q0)	edge[below]							node{$\Sigma_{0,1}$}	(q1);
		\draw	(q0)	edge[above, bend left=30]			node{$\Sigma_{0,2}$}	(q2);
		\draw	(q0)	edge[above, bend left=45]			node{$\Sigma_{0,\symLinLenVar-1}$}	(qn-1);
		\draw	(q0)	edge[above, bend left=60]			node{$\Sigma_{0,\symLinLenVar}$}	(qn);
		\draw	(q0)	edge[above, bend left=75]			node{$\Sigma_{0,+}$}	(q+);
		\draw	(q0)	edge[above, bend right=75]			node{$\Sigma_{0,-}$}	(q-);
		
		\draw	(q1)	edge[below]							node{$\Sigma_{1,2}$}	(q2);
		\draw	(q1)	edge[above, bend right=30]			node{$\Sigma_{1,\symLinLenVar-1}$}	(qn-1);
		\draw	(q1)	edge[above, bend right=45]			node{$\Sigma_{1,\symLinLenVar}$}	(qn);
		\draw	(q1)	edge[above, bend left=67.5]			node{$\Sigma_{1,+}$}	(q+);
		\draw	(q1)	edge[above, bend right=60]			node{$\Sigma_{1,-}$}	(q-);
		
		\draw	(q2)	edge[dashed]						node{}	(qn-1);
		\draw	(q2)	edge[above, bend left=30]			node{$\Sigma_{2,\symLinLenVar-1}$}	(qn-1);
		\draw	(q2)	edge[above, bend left=45]			node{$\Sigma_{2,\symLinLenVar}$}	(qn);
		\draw	(q2)	edge[above, bend left=60]			node{$\Sigma_{2,+}$}	(q+);
		\draw	(q2)	edge[above, bend right=60]			node{$\Sigma_{2,-}$}	(q-);
		
		\draw	(qn-1)	edge[below]							node{$\Sigma_{\symLinLenVar-1,\symLinLenVar}$}	(qn);
		\draw	(qn-1)	edge[above, bend left=30]			node[pos=0.6]{$\Sigma_{\symLinLenVar-1,+}$}(q+);
		\draw	(qn-1)	edge[below, bend right=30]			node[pos=0.6]{$\Sigma_{\symLinLenVar-1,-}$}(q-);
		
		\draw	(qn)	edge[left]							node{$\Sigma_{\symLinLenVar,+}$}	(q+);
		\draw	(qn)	edge[left]							node{$\Sigma_{\symLinLenVar,-}$}	(q-);
		
		\draw	(q+)	edge[loop below]					node{$\Sigma$}			(q+);
		
		\draw	(q-)	edge[loop above]					node{$\Sigma$}			(q-);
	\end{tikzpicture}
	\caption{Minimal linear safety $\symADFAp$ $\mathcal{A}$ with $\symLinLenVar = \opLinLen{\mathcal{A}}$. We have $\Sigma_{i-1, -} \neq \emptyset$ for all $i \in \{1, \dots, \symLinLenVar\}$ and $\Sigma_{\symLinLenVar,-} \neq \emptyset$.}
	\label{fig:pmADFAMinLinSafe}
\end{figure}

For the remainder of \cref{subsubsec:compositionalityADFApMLSProofs}, let $\mathcal{A} = (Q, \Sigma, q_I, \delta, F)$ be a minimal linear safety $\symADFAp$ deciding a language $L$. Let $\symLinLenVar = \opLinLen{\mathcal{A}}$. W.l.o.g.\ we assume $Q = \{q_0, \dots, q_\symLinLenVar\} \cup \{q_+, q_-\}$ where $q_+$ is the accepting sink, $q_-$ is the rejecting sink, and where $q_j$ is reachable from $q_i$ for all $i, j \in \{0, \dots, \symLinLenVar\}, i < j$. This implies $q_I = q_0$. We define $\Sigma_{i, \symHash} = \{\sigma \in \Sigma \mid \delta(q_i, \sigma) = q_\symHash\}$ for all $i \in \{0, \dots, \symLinLenVar\}, \symHash \in \{0, \dots, \symLinLenVar\} \cup \{+, -\}$.

Because $\mathcal{A}$ is a minimal $\symADFAp$, both $q_+$ and $q_-$ are reachable, implying that $q_I$ is a non-sink.
Note that we have $F = Q \setminus \{q_-\}$ because $\mathcal{A}$ is a safety DFA. Further, note that $\Sigma_{i-1,i} \neq \emptyset$ for all $i \in \{1, \dots, \symLinLenVar\}$, which is implied by our assumption that $q_j$ is reachable from $q_i$ for all $i, j \in \{0, \dots, \symLinLenVar\}, i < j$ together with the linearity of $\mathcal{A}$. Finally, note that the sole rejecting state $q_-$ is reachable from every state except $q_+$. Otherwise, $\mathcal{A}$ would not be minimal. In particular, this implies $\Sigma_{\symLinLenVar,-} \neq \emptyset$. The form of $\mathcal{A}$ is outlined in \cref{fig:pmADFAMinLinSafe}.

\paragraph{Minimal Linear Safety $\symADFAp$s with the \symPPShort}
\label{par:pmADFAMinLinSafeWithPPComposite}
As our first step towards \cref{the:pmADFAMinLinSafeCompositionality}, we prove:

\begin{restatable}{lemma}{lempmADFAMinLinSafeWithPPComposite}
	\label{lem:pmADFAMinLinSafeWithPPComposite}\RestateRemark
	The minimal safety $\symADFAp$ $\mathcal{A}$ is composite if it has the \symPPShort.\lipicsEnd
\end{restatable}

We assume throughout Appendix \ref{par:pmADFAMinLinSafeWithPPComposite} that $\mathcal{A}$ has the \symPPShort.
Note that this implies that $\mathcal{A}$ is non-trivial, meaning $\symLinLenVar>0$.
We prove the compositionality of $\mathcal{A}$ by specifying a $(\symLinLenVar+2)$-decomposition of $\mathcal{A}$. We define two DFA constructions, which are outlined in \cref{fig:aipAwij}. The first handles the extensions of non-max-visiting words, the second the extensions of max-visiting words.

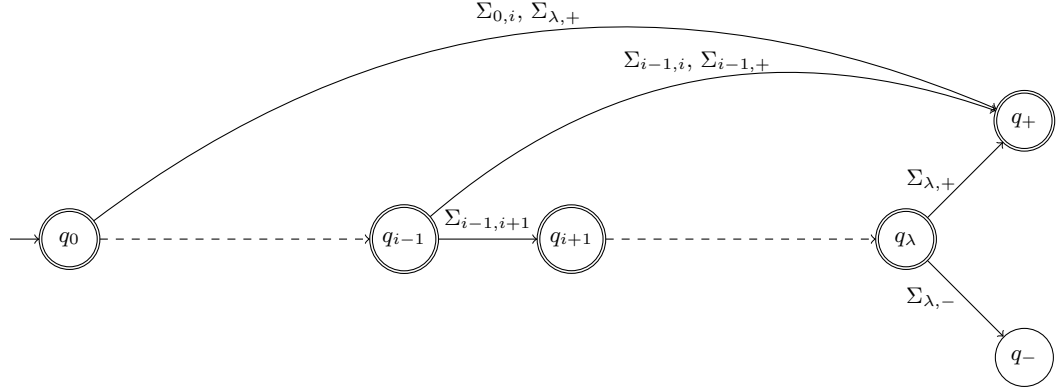
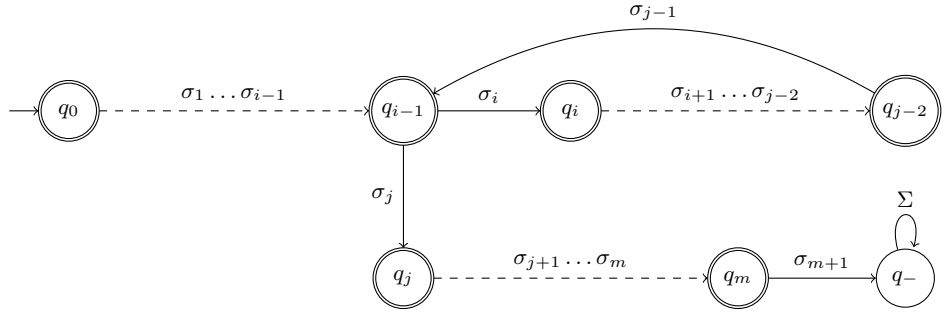
\begin{figure}[t] % fig:aipAwij
	\begin{subfigure}{1\textwidth}
		\centering
		\begin{tikzpicture}[node distance=2.2125cm]
			\footnotesize
			\node[state, initial, accepting] 					(q0) 	{$q_0$};
			\node[state, right of=q0, draw=none]				(e0) 	{};
			\node[state, right of=e0, accepting] 				(qi-1) 	{$q_{i-1}$};
			\node[state, right of=qi-1, accepting] 				(qi+1) 	{$q_{i+1}$};
			\node[state, right of=qi+1, draw=none]				(e1) 	{};
			\node[state, right of=e1, accepting] 				(qn) 	{$q_\symLinLenVar$};
			\node[state, above right of=qn, accepting] 			(q+) 	{$q_+$};
			\node[state, below right of=qn] 					(q-) 	{$q_-$};
			
			\draw	(q0)	edge[dashed]						node{}					(qi-1);
			\draw	(q0)	edge[above, bend left]				node{$\Sigma_{0,i}$, $\Sigma_{\symLinLenVar,+}$}(q+);
			
			\draw	(qi-1)	edge[above]							node{$\Sigma_{i-1,i+1}$}(qi+1);
			\draw	(qi-1)	edge[above, bend left]				node{$\Sigma_{i-1,i}$, $\Sigma_{i-1,+}$}(q+);
			
			\draw	(qi+1)	edge[dashed]						node{}					(qn);
			
			\draw	(qn)	edge[left]							node{$\Sigma_{\symLinLenVar,+}$}	(q+);
			\draw	(qn)	edge[left]							node{$\Sigma_{\symLinLenVar,-}$}	(q-);
		\end{tikzpicture}
		\caption{DFA $\mathcal{A}_i^+$ for an $i \in \{1, \dots, \symLinLenVar\}$. Most unaltered transitions are omitted.}
		\label{subfig:aip}
	\end{subfigure}
	\begin{subfigure}{1\textwidth}
		\centering
		\begin{tikzpicture}[node distance=2.2125cm]
			\footnotesize
			\node[state, initial, accepting] 					(q0) 			{$q_0$};
			\node[state, right of=q0, draw=none]				(e0) 			{};
			\node[state, right of=e0, accepting] 				(qi-1) 			{$q_{i-1}$};
			\node[state, right of=qi-1, accepting] 				(qi) 			{$q_i$};
			\node[state, right of=qi, draw=none]				(e1) 			{};
			\node[state, right of=e1, accepting] 				(qj-2) 			{$q_{j-2}$};
			\node[state, below of=qi-1, accepting] 				(qj) 			{$q_j$};
			\node[state, right of=qj, draw=none]				(e2) 			{};
			\node[state, right of=e2, accepting] 				(qm) 			{$q_m$};
			\node[state, right of=qm] 							(q-)			{$q_-$};
			
			\draw	(q0)	edge[above, dashed]					node{$\sigma_1 \dots \sigma_{i-1}$}(qi-1);
			
			\draw	(qi-1)	edge[above]							node{$\sigma_i$}(qi);
			\draw	(qi-1)	edge[left]							node{$\sigma_j$}(qj);
			
			\draw	(qi)	edge[above, dashed]					node{$\sigma_{i+1} \dots \sigma_{j-2}$}(qj-2);
			
			\draw	(qj-2)	edge[above, bend right]				node{$\sigma_{j-1}$}(qi-1);
			
			\draw	(qj)	edge[above, dashed]					node{$\sigma_{j+1} \dots \sigma_m$}(qm);
			
			\draw	(qm)	edge[above]							node{$\sigma_{m+1}$}(q-);
			
			\draw	(q-)	edge[loop above]					node{$\Sigma$}	(q-);
		\end{tikzpicture}
		\caption{DFA $\mathcal{A}_{w,i,j}$ for a $w = \sigma_1 \dots \sigma_{m+1} \in \Sigma^m$ with $m \in \symNatNumGeq{1}$ and $i,j \in \{1, \dots, m+1\}, i < j, \sigma_i \neq \sigma_j$. Omitted transitions lead to a not depicted accepting sink.}
		\label{subfig:awij}
	\end{subfigure}
	\caption{DFAs $\mathcal{A}_i^+$ and $\mathcal{A}_{w,i,j}$.}
	\label{fig:aipAwij}
\end{figure}

We start with the first DFA construction $\mathcal{A}_i^+$. We provide the definition and prove the lemma stating the relevant properties.

\begin{restatable}{definitionRoman}{defaip}
	\label{def:aip}\RestateRemark
	Let $i \in \{1, \dots, \symLinLenVar\}$. We define $\mathcal{A}_i^+ = (Q_i^+, \Sigma, q_0, \delta_i^+, F_i^+)$ where:
	\begin{align*}
		Q_i^+	& = Q \setminus \{q_i\}, \\
		F_i^+	& = F \setminus \{q_i\}, \\
		\delta_i^+(q, \sigma) & = 
		\begin{cases}
			\delta(q, \sigma)	& \text{ if $\delta(q,\sigma) \neq q_i$} \\
			q_+					& \text{ else, thus if $\delta(q,\sigma) = q_i$}
		\end{cases}.
	\end{align*}
	\lipicsEnd
\end{restatable}

For $i \in \{1, \dots, \symLinLenVar\}$, the DFA $\mathcal{A}_i^+$ is constructed out of $\mathcal{A}$ by removing $q_i$ and redirecting every in-$q_i$-transition into $q_+$.
The DFA $\mathcal{A}_i^+$ replicates every run of $\mathcal{A}$ that does not contain $q_i$ and defaults into $q_+$ for every other run. Thus, these DFAs can be used to reject the extensions of words on which $\mathcal{A}$ enters $q_-$ without visiting every non-sink, so the extensions of non-max-visiting words.
The DFA construction is outlined in \cref{subfig:aip}. \cref{lem:aip} states the relevant properties of $\mathcal{A}_i^+$.

\begin{restatable}{lemma}{lemaip}
	\label{lem:aip}\RestateRemark
	The following assertions hold:
	\begin{romanenumerate}
		\item $\mathcal{A}_i^+ \in \opDecompSet{\mathcal{A}}$ for all $i \in \{1, \dots, \symLinLenVar\}$.
		\item Let $w \notin L$ where $w$ is not an extension of a word in $\opDiffWords{\mathcal{A}}$. Then $w \notin \bigcap_{i=1}^\symLinLenVar \opLang{\mathcal{A}_i^+}$ holds. \lipicsEnd
	\end{romanenumerate}
\end{restatable}

\begin{proof}[Proof of \cref{lem:aip}]
	We begin by stating a key observation:
	\begin{claimWithinProof}
		\label{claWP:run}
		Let $i \in \{1, \dots, \symLinLenVar\}$. Let $m \in \symNatNum, w = \sigma_1 \dots \sigma_m \in \Sigma^m$. Let $\symRun = p_0, \sigma_1, p_1, \dots, \sigma_m, p_m$ be the initial run of $\mathcal{A}$ on $w$. Let $\symRun' = p_0', \sigma_1, p_1', \dots, \sigma_m, p_m'$ be the initial run of $\mathcal{A}_i^+$ on $w$. Then the following holds:
		\begin{itemize}
			\item If $\symRun$ does not contain $q_i$, then we have $\symRun = \symRun'$.
			\item Otherwise, let $j$ be the minimum value in $\{0, \dots, m\}$ with $p_j = q_i$. Then we have $p_k' = p_k$ for all $k \in \{0, \dots, j-1\}$ and $p_k' = q_+$ for all $k \in \{j, \dots, m\}$.
		\end{itemize}
	\end{claimWithinProof}
	
	In other words, if an initial run $\symRun$ of $\mathcal{A}$ does not contain $q_i$, then $\mathcal{A}_i^+$ replicates $\symRun$ completely.
	Otherwise, $\mathcal{A}_i^+$ replicates $\symRun$ up to the first appearance of $q_i$ and then enters the accepting sink.
	\cref{claWP:run} follows immediately from the definition of $\mathcal{A}_i^+$.
	
	With \cref{claWP:run} in hand, we can begin the proof of the two assertions.
	
	We begin with Assertion (i). Let $i \in \{1, \dots, \symLinLenVar\}$. We have $\opSize{\mathcal{A}_i^+} = \opSize{\mathcal{A}} - 1 < \opSize{\mathcal{A}}$. Thus, $\mathcal{A}_i^+$ is sufficiently small. To prove $L \subseteq \opLang{\mathcal{A}_i^+}$, let $w \in L$. If the initial run of $\mathcal{A}$ on $w$ does not contain $q_i$, then we have $\delta_i^+(q_0, w) = \delta(q_0, w)$ with \cref{claWP:run}. With $F_i^+ = F \setminus \{q_i\}$, this implies $w \in \opLang{\mathcal{A}_i^+}$. Otherwise, that is, if the initial run of $\mathcal{A}$ on $w$ contains $q_i$, then we have $\delta_i^+(q_0, w) = q_+$ with \cref{claWP:run} and therefore $w \in \opLang{\mathcal{A}_i^+}$. Thus, we have $L \subseteq \opLang{\mathcal{A}_i^+}$. With $\opSize{\mathcal{A}_i^+} < \opSize{\mathcal{A}}$ and $L \subseteq \opLang{\mathcal{A}_i^+}$, we have $\mathcal{A}_i^+ \in \opDecompSet{\mathcal{A}}$. The proof of Assertion (i) is complete.
	
	Next, we consider Assertion (ii). Let $w \notin L$ where $w$ is not an extension of a word in $\opDiffWords{\mathcal{A}}$. Then there exists an $i \in \{1, \dots, \symLinLenVar\}$ such that the initial run of $\mathcal{A}$ on $w$ does not contain $q_i$. With \cref{claWP:run}, $\mathcal{A}_i^+$ replicates the initial run of $\mathcal{A}$ on $w$, implying $w \notin \opLang{\mathcal{A}_i^+}$ and therefore $w \notin \bigcap_{i=1}^\symLinLenVar \opLang{\mathcal{A}_i^+}$. The proof of Assertion (ii) is complete.
	
	We have proven Assertions (i) and (ii). The proof is complete.
\end{proof}

Now we turn to the second DFA construction $\mathcal{A}_{w,i,j}$. With \cref{def:awij} and \cref{lem:awij}, we provide its definition and a lemma stating its relevant properties. But first, we give the intuition behind the construction and introduce a necessary technical result with \cref{lem:pPAndDifferingSigmas}.

Recall that the problem with rejecting the extensions of a max-visiting word $w = \sigma_1 \dots \sigma_{\symLinLenVar+1} \in \opDiffWords{\mathcal{A}}$ is that the DFAs in $\opDecompSet{\mathcal{A}}$ are "at least one state short" to keep track of $w$. However, the \symPPShort{} allows DFAs to "confuse" $w$ with $\opPump{w}{i_w,j_w;l}$ for all $l \in \symNatNum$, where $i_w, j_w$ are indices so that the \symPPShort-condition holds for $w, i_w, j_w$.
We exploit this to construct, for every $w \in \opDiffWords{\mathcal{A}}$, a DFA $\mathcal{A}_{w, i_w, j_w}$ which rejects exactly the extensions of the pumpings $\opPump{w}{i_w,j_w;l}$ for all $l \in \symNatNum$.
However, the definition of $\mathcal{A}_{w, i_w, j_w}$ requires $\sigma_{i_w} \neq \sigma_{j_w}$. That such indices can be selected is guaranteed by:

\begin{restatable}{lemma}{lempPAndDifferingSigmas}
	\label{lem:pPAndDifferingSigmas}\RestateRemark
	Let $w = \sigma_1 \dots \sigma_{\symLinLenVar+1} \in \opDiffWords{\mathcal{A}}$ and let $i$ be the maximum value in $\{1, \dots, \symLinLenVar + 1\}$ such that a $j \in \{1, \dots, \symLinLenVar + 1\}, i < j$ exists so that the \symPPShort-condition holds for $w, i, j$. Then we have $\sigma_i \neq \sigma_j$. \lipicsEnd
\end{restatable}

\begin{proof}[Proof of \cref{lem:pPAndDifferingSigmas}]
	Let $w = \sigma_1 \dots \sigma_{\symLinLenVar+1} \in \opDiffWords{\mathcal{A}}$ and let $i, j \in \{1, \dots, \symLinLenVar + 1\}, i < j$ be some indices so that the \symPPShort-condition holds for $w, i ,j$.
	
	\begin{description}
		\item[Case 1: \normalfont{$j = \symLinLenVar+1$.}] We prove that $\sigma_i \neq \sigma_j$ holds. If this holds for the arbitrarily selected index $i$, then it necessarily also holds if $i$ is the maximum possible value.
		
		We have $\delta(q_0, \sigma_1 \dots \sigma_{i-1} \sigma_{\symLinLenVar+1}) = \delta(q_0, \sigma_1 \dots \sigma_{i-1} (\sigma_i \dots \sigma_{j-1})^0 \sigma_j \dots q_{\symLinLenVar+1}) = \delta(q_0, \opPump{w}{i,j;0}) = q_-$. This implies $\sigma_j \in \Sigma_{i-1, -}$. However, we have $\sigma_i \in \Sigma_{i-1, i}$. Thus, we have $\sigma_i \neq \sigma_j$. We are done with Case 1.
		
		\item[Case 2: \normalfont{$j < \symLinLenVar+1$.}] Towards a proof by contraposition, let $\sigma_i = \sigma_j$. We prove that there are values $i', j' \in \{1, \dots, \symLinLenVar + 1\}, i' < j'$ with $i < i'$ so that the \symPPShort-condition holds for $w, i' ,j'$. In other words, we prove that $\sigma_i = \sigma_j$ implies that $i$ is not the maximum possible value.
		
		With $\sigma_i = \sigma_j$, it is easy to see that the following holds:
		\begin{align*}
			\forall l \in \symNatNum. \sigma_1 \dots \sigma_{i-1} (\sigma_i \dots \sigma_{j-1})^l \sigma_j = \sigma_1 \dots \sigma_i (\sigma_{i+1} \dots \sigma_{j})^l.
		\end{align*}
		\iffalse
		We claim that the following holds:
		\begin{align*}
			\forall l \in \symNatNum. \sigma_1 \dots \sigma_{i-1} (\sigma_i \dots \sigma_{j-1})^l \sigma_j = \sigma_1 \dots \sigma_i (\sigma_{i+1} \dots \sigma_{j})^l.
		\end{align*}
		We prove this claim by induction.
		
		For $l = 0$, we have:
		\begin{align*}
			& \sigma_1 \dots \sigma_{i-1} (\sigma_i \dots \sigma_{j-1})^0 \sigma_j \\
			= 	& \sigma_1 \dots \sigma_{i-1} \sigma_j \\
			=	& \sigma_1 \dots \sigma_{i-1} \sigma_i \\
			=	& \sigma_1 \dots \sigma_{i-1} \sigma_i (\sigma_{i+1} \dots \sigma_{j})^0.
		\end{align*}
		Thus, the claim holds for $l = 0$.
		
		Let $l \in \symNatNum$ arbitrary but fixed. We assume that the claim holds for this $l$. Then for $l + 1$, we have:
		\begin{align*}
			& \sigma_1 \dots \sigma_{i-1} (\sigma_i \dots \sigma_{j-1})^{l+1} \sigma_j \\
			=	& \sigma_1 \dots \sigma_{i-1} (\sigma_i \dots \sigma_{j-1})^l (\sigma_i \dots \sigma_{j-1}) \sigma_j \\
			=	& \sigma_1 \dots \sigma_{i-1} (\sigma_i \dots \sigma_{j-1})^l \sigma_i (\sigma_{i+1} \dots \sigma_j) \\
			=	& \sigma_1 \dots \sigma_{i-1} (\sigma_i \dots \sigma_{j-1})^l \sigma_j (\sigma_{i+1} \dots \sigma_j) \\
			=	& \sigma_1 \dots \sigma_i (\sigma_{i+1} \dots \sigma_{j})^l (\sigma_{i+1} \dots \sigma_j) \\
			=	& \sigma_1 \dots \sigma_i (\sigma_{i+1} \dots \sigma_{j})^{l+1}.
		\end{align*}
		Thus, the claim holds for $l+1$.
		
		By induction, the claim holds for all $l \in \symNatNum$.
		\fi
		
		Now we select $i' = i + 1, j' = j + 1$. Note that, with $i < j < \symLinLenVar+1$, we have $i' < j' \leq \symLinLenVar + 1$. Thus, we have $i', j' \in \{1, \dots, \symLinLenVar+1\}, i' < j'$ such that the following holds for every $l \in \symNatNum$:
		\begin{align*}
			& \opPump{w}{i',j';l} \\
			=	& \sigma_1 \dots \sigma_{i'-1} (\sigma_{i'} \dots \sigma_{j'-1})^l \sigma_{j'} \dots \sigma_{\symLinLenVar+1} \\
			=	& \sigma_1 \dots \sigma_i (\sigma_{i+1} \dots \sigma_{j})^l \sigma_{j+1} \dots \sigma_{\symLinLenVar+1} \\
			=	& \sigma_1 \dots \sigma_{i-1} (\sigma_i \dots \sigma_{j-1})^l \sigma_j \dots \sigma_{\symLinLenVar+1} \\
			=	& \opPump{w}{i,j;l}.
		\end{align*}
		Here, the third equality holds with the observation we just made above.
		
		Thus, the \symPPShort-condition holds for $w, i', j'$. Therefore, the index $i$ is not the maximum possible value. We are done with Case 2.
	\end{description}
	
	With Cases 1 and 2, we have $\sigma_i \neq \sigma_j$ if $i$ is the maximum possible value. The proof is complete.
\end{proof}

With \cref{lem:pPAndDifferingSigmas}, we can define: For every $w = \sigma_1 \dots \sigma_{\symLinLenVar+1} \in \opDiffWords{\mathcal{A}}$, let $i_w, j_w \in \{1, \dots, \symLinLenVar+1\}, i_w < j_w$ with $\sigma_{i_w} \neq \sigma_{j_w}$ so that the \symPPShort-condition holds for $w, i_w, j_w$. We use this notation for the remainder of Appendix \ref{par:pmADFAMinLinSafeWithPPComposite}.

With this notation in hand, we can consider DFA construction $\mathcal{A}_{w,i,j}$. 
Note that the form of $\mathcal{A}_{w,i,j}$ depends only on the word $w$ and the indices $i,j$ and is independent of the minimal safety $\symADFAp$ $\mathcal{A}$ which we have fixed for \cref{subsubsec:compositionalityADFApMLSProofs}.
Consequently, \cref{def:awij,subfig:awij} are independent of $\mathcal{A}$. Instead, they consider, for some $m \in \symNatNumGeq{1}$, a word $w = \sigma_1 \dots \sigma_{m+1} \in \Sigma^{m+1}$ and indices $i,j \in \{1,\dots,m+1\}, i<j$ with $\sigma_i \neq \sigma_j$.

\begin{restatable}{definitionRoman}{defawij}
	\label{def:awij}\RestateRemark
	Let $m \in \symNatNumGeq{1}$. Let $w = \sigma_1 \dots \sigma_{m+1} \in \Sigma^{m+1}$ with $w \neq \sigma^{m+1}$ for every $\sigma \in \Sigma$. Let $i, j \in \{1, \dots, m+1\}, i < j$ with $\sigma_i \neq \sigma_j$. We define $\mathcal{A}_{w,i,j} = (Q_{w,i,j}, \Sigma, q_0, F)$ where:
	\begin{align*}
		Q_{w,i,j}	& = (\{q_0, \dots, q_m\} \setminus \{q_{j-1}\}) \cup \{q_+, q_-\}, \\
		F_{w,i,j} 	& = Q_{w,i,j} \setminus q_-, \\
		\delta_{w,i,j}(q_k, \sigma) & =
		\begin{cases} 
			q_{k+1}		& \text{ if $k \notin \{i - 1, j - 2, m\}$ and $\sigma = \sigma_{k + 1}$} \\
			q_+			& \text{ if $k \notin \{i - 1, j - 2, m\}$ and $\sigma \neq \sigma_{k + 1}$} \\
			q_i			& \text{ if $k = i - 1$ and $\sigma = \sigma_i$} \\
			q_j			& \text{ if $k = i - 1$ and $\sigma = \sigma_j$} \\
			q_+			& \text{ if $k = i - 1$ and $\sigma \notin \{\sigma_i, \sigma_j\}$} \\
			q_{i-1}		& \text{ if $k = j - 2$ and $\sigma = \sigma_{j-1}$} \\
			q_+			& \text{ if $k = j - 2$ and $\sigma \neq \sigma_{j-1}$} \\
			q_-			& \text{ if $k = m$ and $\sigma = \sigma_{m+1}$} \\
			q_+			& \text{ else, thus if $k = m$ and $\sigma \neq \sigma_{m+1}$}
		\end{cases}\\
		& \text{for all } q_k \in Q_{w,i,j} \setminus \{q_+, q_-\}, \\
		\delta_{w,i,j}(q, \sigma) & = q\\
		& \text{for all } q \in \{q_+, q_-\}.
	\end{align*}
	In the above definition of $\delta_{w,i,j}$, we use $q_j = q_{m+1} = q_-$ for $j = m+1$.
	\lipicsEnd
\end{restatable}

Let $m \in \symNatNumGeq{1}$. Let $w = \sigma_1 \dots \sigma_{m+1}$ with $w \neq \sigma^{m+1}$ for every $\sigma \in \Sigma$. Let $i,j \in \{1, \dots, m+1\}, i < j$ with $\sigma_i \neq \sigma_j$. We describe the idea behind $\mathcal{A}_{w,i,j}$ by starting from a different, very simple DFA.

It is trivial to construct the minimal DFA that rejects exactly the extensions of $w$. This DFA has $m+3$ states and is a minimal linear safety $\symADFAp$. It advances exactly one state for each letter of $w$, entering the rejecting sink on the last letter $\sigma_{m+1}$. For every "unexpected" letter, this $\symADFAp$ immediately enters the accepting sink.

The DFA $\mathcal{A}_{w,i,j}$ is constructed out of this $\symADFAp$ by removing the state $q_{j-1}$ that is entered after reading $\sigma_1 \dots \sigma_{j-1}$ and redirecting the sole in-$q_{j-1}$-transition into the state $q_{i-1}$ that is entered after reading $\sigma_1 \dots \sigma_{i-1}$. Further, the sole in-$q_j$-transition, which is of the form $q_{j-1} \xrightarrow{\sigma_j} q_j$, is replaced by the transition $q_{i-1} \xrightarrow{\sigma_j} q_j$. Note that, with $\sigma_i \neq \sigma_j$, this change does not affect the transition $q_{i-1} \xrightarrow{\sigma_i} q_i$.

The effect of these modifications is that $\mathcal{A}_{w,i,j}$ now allows for repetitions of the subword $\sigma_i \dots \sigma_{j-1}$. Thus, $\mathcal{A}_{w,i,j}$ rejects exactly the extensions of the pumpings $\opPump{w}{i,j;l}$ for all $l \in \symNatNum$. The DFA construction is outlined in \cref{subfig:awij}. We will use this DFA construction for the words $w \in \opDiffWords{\mathcal{A}}$ and the respective indices $i_w, j_w$. \cref{lem:awij} states the relevant properties of $\mathcal{A}_{w,i_w,j_w}$.

\begin{restatable}{lemma}{lemawij}
	\label{lem:awij}\RestateRemark
	Let $w \in \opDiffWords{\mathcal{A}}$. Then the following assertions hold:
	\begin{romanenumerate}
		\item $\mathcal{A}_{w,i_w,j_w} \in \opDecompSet{\mathcal{A}}$.
		\item $\opLang{\mathcal{A}_{w,i_w,j_w}} = \Sigma^* \setminus \{\opPump{w}{i_w,j_w;l} v \mid l \in \symNatNum \wedge v \in \Sigma^*\}$. \lipicsEnd
	\end{romanenumerate}
\end{restatable}

\begin{proof}[Proof of \cref{lem:awij}]
	Let $w \in \opDiffWords{\mathcal{A}}$. Assertion (ii) follows immediately from the definition of $\mathcal{A}_{w,i_w,j_w}$. Therefore, we only consider Assertion (i). 
	
	We have $\opSize{\mathcal{A}_{w,i_w,j_w}} = \opSize{\mathcal{A}} - 1 < \opSize{\mathcal{A}}$. Thus, $\mathcal{A}_{w,i_w,j_w}$ is sufficiently small. To prove $L \subseteq \opLang{\mathcal{A}_{w,i_w,j_w}}$, let $w' \notin \opLang{\mathcal{A}_{w,i_w,j_w}}$. With Assertion (ii), this implies $w' = \opPump{w}{i_w,j_w;l} v$ for some $l \in \symNatNum, v \in \Sigma^*$.
	Because $i_w,j_w$ are selected so that the \symPPShort-condition holds for $w, i_w, j_w$, we have $\opPump{w}{i_w,j_w;l} \notin L$.
	And because $\mathcal{A}$ is a safety DFA, this implies $w' = \opPump{w}{i_w,j_w;l} v \notin L$.
	Thus, we have $L \subseteq \opLang{\mathcal{A}_{w,i_w,j_w}}$. With $\opSize{\mathcal{A}_{w,i_w,j_w}} < \opSize{\mathcal{A}}$ and $L \subseteq \opLang{\mathcal{A}_{w,i_w,j_w}}$, we have $\mathcal{A}_{w,i_w,j_w} \in \opDecompSet{\mathcal{A}}$. The proof of Assertion (i) is complete.
	
	We have proven Assertion (i). As mentioned, Assertion (ii) follows immediately from the definition of $\mathcal{A}_{w,i_w,j_w}$. The proof is complete.
\end{proof}

\cref{lem:aip,lem:awij} immediately imply:

\begin{restatable}{lemma}{lempmADFAMinLinSafeWithPPDecomposition}
	\label{lem:pmADFAMinLinSafeWithPPDecomposition}\RestateRemark
	The following condition holds:
	\begin{align*}
		L = \bigcap_{i=1}^\symLinLenVar \opLang{\mathcal{A}_{i}^+} \cap \bigcap_{w \in \opDiffWords{\mathcal{A}}} \opLang{\mathcal{A}_{w,i_w,j_w}}.
	\end{align*}
	Thus, the $\symADFAp$ $\mathcal{A}$ is composite.
	\lipicsEnd
\end{restatable}

\begin{proof}[Proof of \cref{lem:pmADFAMinLinSafeWithPPDecomposition}]
	We denote the language created by the intersections on the right hand side by $L_\cap$. With \cref{lem:aip,lem:awij}, all DFAs used to create $L_\cap$ are in $\opDecompSet{\mathcal{A}}$. Therefore, they are sufficiently small and $L \subseteq L_\cap$ holds. We only have to prove $L_\cap \subseteq L$.
	
	To prove $L_\cap \subseteq L$, let $w \notin L$. We show $w \notin L_\cap$. If $w \notin \bigcap_{i=1}^\symLinLenVar \opLang{\mathcal{A}_{i}^+}$, this holds trivially. Therefore, we assume $w \in \bigcap_{i=1}^\symLinLenVar \opLang{\mathcal{A}_{i}^+}$. With \cref{lem:aip}, this implies that $w$ is an extension of a word $u \in \opDiffWords{\mathcal{A}}$. And with \cref{lem:awij}, we then have $w \notin \opLang{\mathcal{A}_{u,i_u,j_u}}$, implying $w \notin L_\cap$. Thus, we have $L_\cap \subseteq L$.
	
	We have proven $L_\cap \subseteq L$. We already argued that $L \subseteq L_\cap$ holds. Additionally, we argued that the DFAs used to create $L_\cap$ are sufficiently small. Thus, we have proven the compositionality of $\mathcal{A}$. The proof is complete.
\end{proof}

\cref{lem:pmADFAMinLinSafeWithPPDecomposition} subsumes \cref{lem:pmADFAMinLinSafeWithPPComposite}, with the decomposition outlined in \cref{lem:pmADFAMinLinSafeWithPPDecomposition} witnessing the compositionality of $\mathcal{A}$ under the assumption that $\mathcal{A}$ has the \symPPShort.
Thus, we have proven that a minimal linear safety $\symADFAp$ is composite if it has the \symPPShort.

\paragraph{Minimal Linear Safety $\symADFAp$s without the \symPPShort}
\label{par:pmADFAMinLinSafeWithoutPPPrime}
As our second and final step towards \cref{the:pmADFAMinLinSafeCompositionality}, we prove:

\begin{restatable}{lemma}{lempmADFAMinLinSafeWithoutPPPrime}
	\label{lem:pmADFAMinLinSafeWithoutPPPrime}\RestateRemark
	The minimal safety $\symADFAp$ $\mathcal{A}$ is prime if it does not have the \symPPShort. \lipicsEnd
\end{restatable}

When we discussed this aspect of \cref{the:pmADFAMinLinSafeCompositionality} in \cref{subsubsec:compositionalityADFApMLS}, we only briefly mentioned an important detail regarding how the form of $\mathcal{A}$ restricts the DFAs we have to consider for decompositions of $\mathcal{A}$. With \cref{lem:safetyDecomp,lem:acceptingSink}, we fill this gap.

With \cref{lem:safetyDecomp}, we restrict the DFAs relevant for decompositions of $\mathcal{A}$ by exploiting that $\mathcal{A}$ is a safety DFA. \cref{lem:safetyDecomp} immediately implies that if $\mathcal{A}$ is composite, then it can be decomposed using only safety DFAs.
This result is originally from \cite{netser18decomposition}. But because \cite{netser18decomposition} is not widely available, we restate the result and offer a proof. In our restatement, we rephrase the result to fit our purposes, but it retains its original insight.

\begin{restatable}[{\cite[Theorem 13]{netser18decomposition}}]{lemma}{lemsafetyDecomp}
	\label{lem:safetyDecomp}\RestateRemark
	Let $\mathcal{B}$ be a minimal safety DFA. Let $\mathcal{B}'$ be a minimal DFA with $\opLang{\mathcal{B}} \subseteq \opLang{\mathcal{B}'}$. Then there is a minimal safety DFA $\mathcal{B}''$ with $\opSize{\mathcal{B}''} \leq \opSize{\mathcal{B}'}$ and $\opLang{\mathcal{B}} \subseteq \opLang{\mathcal{B}''} \subseteq \opLang{\mathcal{B}'}$. \lipicsEnd
\end{restatable}

\begin{proof}[Proof of \cref{lem:safetyDecomp}]
	Let $\mathcal{B} = (S, \Pi, s_I, \eta, G)$ be a minimal safety DFA, and let $\mathcal{B}' = (S', \Pi, s_I', \eta', G')$ be a minimal DFA with $\opLang{\mathcal{B}} \subseteq \opLang{\mathcal{B}'}$.
	
	If $\opSize{\mathcal{B}} = 1$ and $s_I \notin G$, then $\mathcal{B}$ consists only of the rejecting sink $s_I$. Thus, we can select $\mathcal{B}'' = \mathcal{B}$ and are done.
	
	If $\opSize{\mathcal{B}} = 1$ and $s_I \in G$, then $\mathcal{B}$ consists only of the accepting sink $s_I$, meaning $\opLang{\mathcal{B}} = \Pi^*$. Because $\mathcal{B}'$ is a minimal DFA for which $\Pi^* = \opLang{\mathcal{B}} \subseteq \opLang{\mathcal{B}'}$ holds, $\mathcal{B}'$ also consists only of an accepting sink. Thus, we can select $\mathcal{B}'' = \mathcal{B}$ (or alternatively $\mathcal{B}'' = \mathcal{B}'$) and are done.
	
	We have covered both cases with $\opSize{\mathcal{B}} = 1$. Therefore, from here on, we assume $\opSize{\mathcal{B}} > 1$. Together with the minimality of $\mathcal{B}$, this means that at least $s_I$ is an accepting non-sink. Together with $\opLang{\mathcal{B}} \subseteq \opLang{\mathcal{B}'}$, this implies that $s_I'$ is accepting as well.
	
	If $s_I'$ is a sink, then $\mathcal{B}'$ is a minimal safety DFA consisting only of the accepting sink $s_I'$. Thus, we can set $\mathcal{B}'' = \mathcal{B}'$ and are done. Therefore, from here on, we assume that $s_I'$ is not a sink. Together with the minimality of $\mathcal{B}'$, this implies that $\mathcal{B}'$ has, in addition to the accepting non-sink $s_I'$, at least one further state, which is rejecting and which may or may not be a sink.
	
	To briefly summarize, in the remainder, $\mathcal{B}'$ is a minimal safety DFA in which $s_I'$ is an accepting non-sink and which has at least one rejecting state which may or may not be a sink. Now we construct a DFA $\hat{\mathcal{B}}$ out of $\mathcal{B}'$ so that the minimal DFA accepting the same language as $\hat{\mathcal{B}}$ will be a suitable choice for $\mathcal{B}''$. We construct $\hat{\mathcal{B}}$ out of $\mathcal{B}'$ by removing all rejecting states of $\mathcal{B}'$, adding one rejecting sink, and redirecting every transition leading into a rejecting state into the rejecting sink. Formally, we define $\hat{\mathcal{B}} = (\hat{S}, \Pi, \hat{s}_I, \hat{\eta}, \hat{G})$ where:
	\begin{align*}
		\hat{S}		& = G' \uplus \{s_-\}, \\
		\hat{s}_I	& = s_I', \\
		\hat{G}		&= G', \\
		\hat{\eta}(s,\sigma) & =
		\begin{cases}
			\eta'(s, \sigma)	& \text{ if $s \in G'$ and $\eta'(s, \sigma) \in G'$} \\
			s_-					& \text{ else, thus if $s = s_-$ or if $s \in G'$ and $\eta'(s, \sigma) \notin G'$}
		\end{cases}.
	\end{align*}
	Note that, as argued above, $s_I' \in G'$ holds. Therefore, we do not remove $s_I'$ in our construction of $\hat{\mathcal{B}}$ and can retain it as the initial state.
	
	Clearly, $\hat{\mathcal{B}}$ is a safety DFA, with the newly introduced state $s_-$ as the rejecting sink. Further, we argued above that $\mathcal{B}'$ has at least one rejecting state. This implies $\opSize{\hat{\mathcal{B}}} \leq \opSize{\mathcal{B}'}$ because we remove at least one state and add exactly one state. We still have to prove $\opLang{\mathcal{B}} \subseteq \opLang{\hat{\mathcal{B}}} \subseteq \opLang{\mathcal{B}'}$.
	
	We begin by proving $\opLang{\hat{\mathcal{B}}} \subseteq \opLang{\mathcal{B}'}$. Let $w \in \Pi^*$ with $w \notin \opLang{\mathcal{B}'}$. We prove $w \notin \opLang{\hat{\mathcal{B}}}$, which immediately implies $\opLang{\hat{\mathcal{B}}} \subseteq \opLang{\mathcal{B}'}$. Let $w'$ be the shortest prefix of $w$ so that $\eta'(s_I',w') \notin G'$. With $w \notin \opLang{\mathcal{B}'}$, such a prefix exists. With the definition of $\hat{\mathcal{B}}$, this clearly means $\hat{\eta}(\hat{s}_I,w') = s_-$, with $\hat{\mathcal{B}}$ entering $s_-$ on the last letter of $w'$. Thus, because $s_-$ is a sink and $w'$ is a prefix of $w$, we have $\hat{\eta}(\hat{s}_I,w) = s_-$ and therefore $w \notin \opLang{\hat{\mathcal{B}}}$. As argued, this implies $\opLang{\hat{\mathcal{B}}} \subseteq \opLang{\mathcal{B}'}$.
	
	Next, we prove $\opLang{\mathcal{B}} \subseteq \opLang{\hat{\mathcal{B}}}$. Let $w \in \Pi^*$ with $w \notin \opLang{\hat{\mathcal{B}}}$. We prove $w \notin \opLang{\mathcal{B}}$, which immediately implies $\opLang{\mathcal{B}} \subseteq \opLang{\hat{\mathcal{B}}}$. Let $w'$ be the shortest prefix of $w$ so that $\hat{\eta}(\hat{s}_I,w') = s_-$. With $w \notin \opLang{\hat{\mathcal{B}}}$, such a prefix exists. With the definition of $\hat{\mathcal{B}}$, this clearly means $\eta'(s_I',w') \notin G'$ and therefore $w' \notin \opLang{\mathcal{B}'}$. Because $\opLang{\mathcal{B}} \subseteq \opLang{\mathcal{B}'}$ holds by selection of $\mathcal{B}'$, this implies $w' \notin \opLang{\mathcal{B}}$. And because $\mathcal{B}$ is a safety DFA and $w'$ is a prefix of $w$, we have that $w' \notin \opLang{\mathcal{B}}$ implies $w \notin \opLang{\mathcal{B}}$. Thus, we have proven $w \notin \opLang{\mathcal{B}}$. As argued, this implies $\opLang{\mathcal{B}} \subseteq \opLang{\hat{\mathcal{B}}}$.
	
	To summarize, $\hat{\mathcal{B}}$ is a safety DFA with $\opSize{\hat{\mathcal{B}}} \leq \opSize{\mathcal{B}'}$ and $\opLang{\mathcal{B}} \subseteq \opLang{\hat{\mathcal{B}}} \subseteq \opLang{\mathcal{B}'}$. Clearly, the minimal DFA $\hat{\mathcal{B}}_{\min}$ recognizing the same language as $\hat{\mathcal{B}}$ is also a safety DFA. Thus, we can select $\mathcal{B}'' = \hat{\mathcal{B}}_{\min}$. The proof is complete. 
\end{proof}

In the above proof of \cref{lem:safetyDecomp}, we followed the ideas from \cite{netser18decomposition}, but corrected what appears to be an inaccuracy in the central DFA construction.
Given DFAs $\mathcal{B}$ and $\mathcal{B}'$ as required, we constructed the DFA $\mathcal{B}''$ out of $\mathcal{B}'$ by essentially removing all rejecting states of $\mathcal{B}'$ and replacing them with one rejecting sink. 
In contrast, in \cite{netser18decomposition}, the DFA $\mathcal{B}''$ is constructed out of $\mathcal{B}'$ by removing all rejecting states and additionally all states reachable from them and replacing them with one rejecting sink.
This appears to be incorrect.
Indeed, we can construct some minimal safety DFA $\mathcal{B}$ with $\emptyset \subset \opLang{\mathcal{B}} \subset \Pi^*$ and a minimal DFA $\mathcal{B}'$ with $\opLang{\mathcal{B}} \subseteq \opLang{\mathcal{B}'} \subset \Pi^*$ so that the initial state of $\mathcal{B}'$ is reachable from one of its rejecting states.
To construct $\mathcal{B}''$ out of $\mathcal{B}'$ following \cite{netser18decomposition}, every state, including the initial state, would be removed.
What the new initial state would be in this case seems to be undefined in \cite{netser18decomposition}. Trying to repair this in the obvious way by setting the newly introduced rejecting sink as the new initial state would lead to $\opLang{\mathcal{B}''} = \emptyset$ and therefore $\opLang{\mathcal{B}} \not\subseteq \opLang{\mathcal{B}''}$.
Removing just the rejecting states of $\mathcal{B}'$, not the rejecting states and the states reachable from them, solves this problem.
Note that in this short discussion, we translated the DFA construction outlined in \cite{netser18decomposition} to our rephrasing of the result.

With \cref{lem:acceptingSink}, we further restrict the DFAs relevant for decompositions of $\mathcal{A}$ by exploiting that $\mathcal{A}$ is minimal and has an accepting sink. With \cref{lem:acceptingSink}, every DFA in $\opDecompSet{\mathcal{A}}$ has an accepting sink.

\begin{restatable}{lemma}{lemacceptingSink}
	\label{lem:acceptingSink}\RestateRemark
	Let $\mathcal{B}$ be a minimal DFA with an accepting sink. Let $\mathcal{B}'$ be a minimal DFA with $\opLang{\mathcal{B}} \subseteq \opLang{\mathcal{B}'}$. Then $\mathcal{B}'$ has an accepting sink as well. \lipicsEnd
\end{restatable}

\begin{proof}[Proof of \cref{lem:acceptingSink}]
	Let $\mathcal{B} = (S, \Pi, s_I, \eta, G)$ be a minimal DFA with an accepting sink $s_+ \in S$. We prove that every minimal DFA $\mathcal{B}'$ with $\opLang{\mathcal{B}} \subseteq \opLang{\mathcal{B}'}$ has an accepting sink as well.
	
	Towards a proof by contraposition, let $\mathcal{B}' = (S', \Pi, s_I', \eta', G')$ be a minimal DFA without an accepting sink. We prove $\opLang{\mathcal{B}} \not\subseteq \opLang{\mathcal{B}'}$ holds.
	
	Let $w \in \Pi^*$ with $\eta(s_I, w) = s_+$. Note that $s_+$ is reachable because $\mathcal{B}$ is minimal. Therefore, such a word exists. Additionally, note that every extension of $w$ is accepted by $\mathcal{B}$. Now consider the state $\eta'(s_I', w)$. Because $\mathcal{B}'$ is minimal and $\eta'(s_I', w)$ is not an accepting sink, there is a word $v \in \Pi^*$ with $\eta'(s_I', wv) \notin G'$. Thus, we have $wv \notin \opLang{\mathcal{B}'}$ but $wv \in \opLang{\mathcal{B}}$. This implies $\opLang{\mathcal{B}} \not\subseteq \opLang{\mathcal{B}'}$.
	
	We have proven that, for every minimal DFA $\mathcal{B}'$ without an accepting sink, $\opLang{\mathcal{B}} \not\subseteq \opLang{\mathcal{B}'}$ holds. By contraposition, every minimal DFA $\mathcal{B}'$ with $\opLang{\mathcal{B}} \subseteq \opLang{\mathcal{B}'}$ has an accepting sink. The proof is complete.
\end{proof}

\cref{lem:safetyDecomp,lem:acceptingSink} together imply that if $\mathcal{A}$ is composite, then it can be decomposed into safety DFAs with an accepting sink. This is an important point, because, as we outlined in \cref{subsubsec:compositionalityADFApMLS}, the proof of \cref{lem:pmADFAMinLinSafeWithoutPPPrime} hinges on the fact that the DFAs in $\opDecompSet{\mathcal{A}}$ are "at least one state short" to keep track of the words in $\opDiffWords{\mathcal{A}}$.
As becomes clear in the now following proof of \cref{lem:pmADFAMinLinSafeWithoutPPPrime}, this is the case because the DFAs relevant for a potential decomposition of $\mathcal{A}$ have to possess, like $\mathcal{A}$ itself, both an accepting and a rejecting sink, as is implied by \cref{lem:safetyDecomp,lem:acceptingSink}.
Thus, they are smaller not because they have no accepting or no rejecting sink, but because they have fewer non-sinks.
And because they have fewer non-sinks, they are "at least one state short" and thus unable to accurately keep track of the words in $\opDiffWords{\mathcal{A}}$.

\begin{proof}[Proof of \cref{lem:pmADFAMinLinSafeWithoutPPPrime}]
	We assume that $\mathcal{A}$ does not have the \symPPShort. We prove that this implies primality.
	
	Note that if $\mathcal{A}$ is trivial, that is, if $\symLinLenVar=0$, then it does not have the \symPPShort{} and is also trivially prime. Thus, for $\symLinLenVar=0$, we are already done. From here on, we therefore assume $\symLinLenVar>0$.
	
	Because $\mathcal{A}$ does not have the \symPPShort, there is a word $w = \sigma_1 \dots \sigma_{\symLinLenVar+1} \in \opDiffWords{\mathcal{A}}$ that breaks the \symPPShort, meaning:
	\begin{align*}
		\forall i,j \in \{1, \dots, \symLinLenVar+1\}.
			i < j \rightarrow \exists l \in \symNatNum.
				\delta(q_0, \opPump{w}{i,j;l}) \neq q_-.
	\end{align*}
	We show that this word $w$ is a primality witness of $\mathcal{A}$.
	
	Note that, with \cref{lem:acceptingSink}, because $\mathcal{A}$ has an accepting sink, every DFA in $\opDecompSet{\mathcal{A}}$ has an accepting sink as well.
	Further, with \cref{lem:safetyDecomp}, if $\mathcal{A}$ is composite, then it can be decomposed into sufficiently small safety DFAs.
	Together, we have that if $\mathcal{A}$ is composite, then there is a decomposition of $\mathcal{A}$ into sufficiently small minimal safety DFAs which possess an accepting sink.
	
	We prove that $\mathcal{A}$ is prime by proving that every sufficiently small minimal safety DFA which possesses an accepting sink and which rejects $w$ decides a language that does not contain $\opLang{\mathcal{A}}$. In other words, every such DFA rejects a word that $\mathcal{A}$ accepts. Namely, it rejects a pumping of $w$ that $\mathcal{A}$ accepts. Therefore, every such DFA is not in $\opDecompSet{\mathcal{A}}$.
	
	Formalizing this, let $\mathcal{B} = (S, \Sigma, s_I, \eta, G)$ be a minimal safety DFA with $\opSize{\mathcal{B}} < \opSize{\mathcal{A}}$ and $w \notin \opLang{\mathcal{B}}$ which possesses an accepting sink. Let $s_+$ denote the accepting sink. Let $s_-$ denote the rejecting sink. To summarize, the following holds:
	\begin{itemize}
		\item $\opSize{\mathcal{B}} < \opSize{\mathcal{A}} = \symLinLenVar+3$,
		\item $\mathcal{B}$ is minimal,
		\item $\mathcal{B}$ is a safety DFA with rejecting sink $s_-$,
		\item $\mathcal{B}$ has an accepting sink $s_+$,
		\item $w \notin \opLang{\mathcal{B}}$.
	\end{itemize}
	We prove $\opLang{\mathcal{A}} \not\subseteq \opLang{\mathcal{B}}$. Thereby, we prove that every DFA in $\opDecompSet{\mathcal{A}}$ accepts $w$.
	
	If $\sigma_1 \dots \sigma_\symLinLenVar \notin \opLang{\mathcal{B}}$, then we immediately have $\opLang{\mathcal{A}} \not\subseteq \opLang{\mathcal{B}}$ because $\sigma_1 \dots \sigma_\symLinLenVar \in \opLang{\mathcal{A}}$. Therefore, we assume $\sigma_1 \dots \sigma_\symLinLenVar \in \opLang{\mathcal{B}}$.
	
	We have $\sigma_1 \dots \sigma_\symLinLenVar \in \opLang{\mathcal{B}}$ and $\sigma_1 \dots \sigma_{\symLinLenVar+1} = w \notin \opLang{\mathcal{B}}$. In particular, this means $\eta(s_0, \sigma_1 \dots \sigma_\symLinLenVar) \notin \{s_+, s_-\}$. Now note that $\opSize{S \setminus \{s_+, s_-\}} = \opSize{\mathcal{B}} - 2 < \opSize{\mathcal{A}} - 2 = (\symLinLenVar + 3) - 2 = \symLinLenVar + 1$, meaning that $\mathcal{B}$ has at most $\symLinLenVar$ non-sinks. Clearly, this means that at least one non-sink appears more than once in the initial run of $\mathcal{B}$ on $\sigma_1 \dots \sigma_\symLinLenVar$. Thus, there exist $i, j \in \{1, \dots, \symLinLenVar+1\}, i < j$ with $\eta(s_I, \sigma_1 \dots \sigma_{i-1}) = \eta(s_I, \sigma_1 \dots \sigma_{j-1})$. This implies, for every $l \in \symNatNum$:
	\begin{align*}
		& \eta(s_I, \opPump{w}{i,j;l}) \\
		=	& \eta(s_I, \sigma_1 \dots \sigma_{i-1} (\sigma_i \dots \sigma_{j-1})^l \sigma_j \dots \sigma_{\symLinLenVar+1}) \\
		=	& \eta(s_I, \sigma_1 \dots \sigma_{i-1} (\sigma_i \dots \sigma_{j-1})^1 \sigma_j \dots \sigma_{\symLinLenVar+1}) \\
		=	& \eta(s_I, w) \\
		=	& s_-.
	\end{align*}
	But by selection of $w$, there is at least one $l \in \symNatNum$ with $\delta(q_0, \opPump{w}{i,j;l}) \neq q_-$. Therefore, for this $l$, we have $\opPump{w}{i,j;l} \in \opLang{\mathcal{A}} \setminus \opLang{\mathcal{B}}$. Thus, we have $\opLang{\mathcal{A}} \not\subseteq \opLang{\mathcal{B}}$.
	
	As argued above, proving $\opLang{\mathcal{A}} \not\subseteq \opLang{\mathcal{B}}$ is sufficient to prove that $w$ is a primality witness of $\mathcal{A}$. This implies the primality of $\mathcal{A}$. The proof is complete.
\end{proof}

\cref{lem:pmADFAMinLinSafeWithPPComposite}, proven in Appendix \ref{par:pmADFAMinLinSafeWithPPComposite}, and \cref{lem:pmADFAMinLinSafeWithoutPPPrime}, proven here in Appendix \ref{par:pmADFAMinLinSafeWithoutPPPrime}, together imply \cref{the:pmADFAMinLinSafeCompositionality}. Thus, we have characterized the compositionality of minimal linear safety $\symADFAp$s by proving that every minimal linear safety $\symADFAp$ is prime if and only if it does not have the \symPPShort.

\subsubsection{From the Characterization to \fontComplexityClass{NP}-Hardness – Proofs}
\label{subsubsec:characterizationToNPHardnessProofs}
We have characterized the compositionality of minimal linear safety $\symADFAp$s and thereby also of \symCNFDFA{}s. Now we can return to our reduction and prove its correctness.

For \cref{subsubsec:characterizationToNPHardnessProofs}, let $\symLinLenVar = \opLinLen{\mathcal{A}_\Phi}$ and $\symCVar = (\symLinLenVar+1) - (r+1)$.
Clearly, we have $\symLinLenVar > r \geq 1$ and $\symCVar > 0$.

Together, \cref{lem:aPhiIsMLSADFAp,the:pmADFAMinLinSafeCompositionality} imply that $\mathcal{A}_\Phi$ is prime if and only if it does not have the \symPPShort. Following \cref{def:pP}, this means that exactly the max-visiting words of $\mathcal{A}_\Phi$ are relevant for the compositionality of $\mathcal{A}_\Phi$. With \cref{lem:aPhi}, we make observations about the max-visiting words of $\mathcal{A}_\Phi$. We restate this result:

\lemaPhi*

We mentioned above that, in the proofs of \cref{lem:aPhiIsMLSADFAp,lem:aPhi}, we make observations about the structure of $\mathcal{A}_\Phi$ and that we therefore prove these two results together. Thus, the following proof covers both results.

\begin{proof}[Proof of \cref{lem:aPhiIsMLSADFAp,lem:aPhi}]
	We start by introducing some notation.
	Let $U = Q_\Phi \setminus \{p_+, p_-\}$ be the set of non-sinks of $\mathcal{A}_\Phi$.
	Let $\symLinLenVar' = \opSize{\mathcal{A}_\Phi} - 3$. Let $\symCVar' = (\symLinLenVar'+1) - (r+1)$.
	Let $J = \{u d c^{\symCVar' - 1} \mid u \in \{0,1\}^r\}$. For $q, q' \in Q_\Phi$ we write $q \preceq q'$ if $q'$ is reachable from $q$.
	
	Next, we make four simple observations about $\mathcal{A}_\Phi$. They can be easily verified by inspecting \cref{def:aPhi,fig:aPhiaPhiXTransitions}.
	\begin{enumerate}[(1)]
		\item Set $J$ contains exactly the words on which the initial run of $\mathcal{A}_\Phi$ contains exactly the non-sinks of $\mathcal{A}_\Phi$. For each $w \in J$, the non-sinks appear exactly once in the initial run of $\mathcal{A}_\Phi$ on $w$, and they appear in the following order, which we will denote by $\symRun_J$:
		\begin{align*}
			& p_0, \dots, p_r, p_c^0, \\
			& p_1^1, \hat{p}_1^1, \dots, p_r^1, \hat{p}_r^1, p_c^1, \\
			& \dots, \\
			& p_1^{s-1}, \hat{p}_1^{s-1}, \dots, p_r^{s-1}, \hat{p}_r^{s-1}, p_c^{s-1}, \\
			& p_1^s, \hat{p}_1^s, \dots, p_r^s, \hat{p}_r^s, p_c^s.
		\end{align*}
		Thus, for $q, q' \in U, q \neq q'$ we have $q \preceq q'$ if $q$ appears prior to $q'$ in $\symRun_J$.
		\item Consider $q, q' \in U, q \neq q'$ where $q$ appears prior to $q'$ in $\symRun_J$. Then there is no transition from $q'$ to $q$. Intuitively, the order $\symRun_J$ describes the "forward direction" in $\mathcal{A}_\Phi$, and there are no "backward transitions" in $\mathcal{A}_\Phi$. Clearly, this implies $q' \npreceq q$ for all $q, q' \in U, q \neq q'$ where $q$ appears prior to $q'$ in $\symRun_J$. Together with Observation (1), this implies, for all $q,q' \in U, q \neq q'$, that $q \preceq q'$ if and only if $q$ appears prior to $q'$ in $\symRun_J$.
		\item The DFA $\mathcal{A}_\Phi$ has two sinks, namely the accepting sink $p_+$ and the rejecting sink $p_-$. These two sinks are the only two states of $\mathcal{A}_\Phi$ which have self-loops.
		\item Both sinks $p_+, p_-$ are reachable from $p_c^s$, which is the last state in $\symRun_J$. This immediately implies $q \preceq p_+$ and $q \preceq p_-$ for every $q \in U$.
	\end{enumerate}
	
	With the notation and these observations in hand, we now consider the two lemmas.
	
	We begin with \cref{lem:aPhiIsMLSADFAp}.
	
	We start by arguing that $\mathcal{A}_\Phi$ is an $\symADFAp$. This follows from Observations (1)–(3). There is only one "direction" in which the non-sinks of $\mathcal{A}_\Phi$ can be traversed, which is specified by $\symRun_J$. No "backward transitions" exist. And only the two sinks of $\mathcal{A}_\Phi$ have self-loops. Therefore, only the two sinks can appear more than once in a run of $\mathcal{A}_\Phi$. Additionally, $\mathcal{A}_\Phi$ has both accepting and rejecting sinks. In total, we have that $\mathcal{A}_\Phi$ is an $\symADFAp$.
	
	Next, we argue that $\mathcal{A}_\Phi$ is minimal. With Observations (1) and (4), every state in $\mathcal{A}_\Phi$ is reachable. Let $q, q' \in Q_\Phi, q \neq q'$. We argue that there is a word $w \in \Sigma^*$ with $\delta_\Phi(q, w) \in F_\Phi$ if and only if $\delta_\Phi(q', w) \notin F_\Phi$. Together with the reachability of every state, this implies minimality. We make a case distinction, specifying an appropriate $w$ for each case.
	\begin{description}
		\item[Case 1: \normalfont{Both states $q, q'$ are sinks.}] Then one of the states is the accepting sink and the other the rejecting sink. Thus, we have $q \in F_\Phi$ if and only if $q' \notin F_\Phi$. We can select $w = \varepsilon$. We are done with Case 1.
		
		\item[Case 2: \normalfont{Exactly one of the states $q, q'$ is a sink.}] W.l.o.g.\ let $q'$ be the sink.
		\begin{description}
			\item[Case 2.1: \normalfont{$q' = p_-$.}] Then we have $q \in F_\Phi$ and $q' \notin F_\Phi$. We can select $w = \varepsilon$. We are done with Case 2.1.
			
			\item[Case 2.2: \normalfont{$q' = p_+$.}] Then, with Observation (4), there is a word $w' \in \Sigma^+$ with $\delta_\Phi(q, w') = p_- \notin F_\Phi$ and $\delta_\Phi(q', w') = p_+ \in F$. We can select $w = w'$. We are done with Case 2.2.
		\end{description}
		With Cases 2.1 and 2.2, we are done with Case 2.
		
		\item[Case 3: \normalfont{Both states $q, q'$ are non-sinks.}] With Observations (1) and (2), we have either $q \preceq q'$ or $q' \preceq q$ but not both. W.l.o.g.\ let $q \preceq q'$. 
		
		Note that the states reachable from $q$ are the two sinks and the states that appear after $q$ in $\symRun_J$.
		Similarly, the states reachable from $q'$ are the two sinks and the states that appear after $q'$ in $\symRun_J$.
		Thus, with $q \preceq q'$, there are strictly more states, and in particular strictly more non-sinks, reachable from $q$ than from $q'$.
		
		Let $w' = 0^r d c^{\symCVar - 1}$. Let $u',v' \in \Sigma^*$ with $w' = u' v'$ and $\delta_\Phi(p_0, u') = q$. Then we have $\delta_\Phi(q, v') = p_c^s$. Note that the run of $\mathcal{A}_\Phi$ on $v'$ starting in $q$ contains exactly the non-sinks reachable from $q$. Now consider the run of $\mathcal{A}_\Phi$ on $v'$ starting in $q'$. Because strictly more non-sinks are reachable from $q$ than from $q'$, because the run of $\mathcal{A}_\Phi$ on $v'$ starting in $q$ contains all non-sinks reachable from $q$, and because $\mathcal{A}_\Phi$ is acyclic, we necessarily have $\delta_\Phi(q', v') \in \{p_+, p_-\}$.
		\begin{description}
			\item[Case 3.1: \normalfont{$\delta_\Phi(q', v') = p_-$.}] Then we have $\delta_\Phi(q, v') \in F_\Phi$ and $\delta_\Phi(q', v') \notin F_\Phi$. We can select $w = v'$. We are done with Case 3.1.
			
			\item[Case 3.2: \normalfont{$\delta_\Phi(q', v') = p_+$.}] Then we have $\delta_\Phi(q, v'c) \notin F_\Phi$ and $\delta_\Phi(q', v'c) \in F_\Phi$. We can select $w = v' c$. We are done with Case 3.2.
		\end{description}
		With Cases 3.1 and 3.2, we are done with Case 3.
	\end{description}
	With Cases 1–3, we have argued that there is a word $w \in \Sigma^*$ so that $\delta_\Phi(q, w) \in F_\Phi$ if and only if $\delta_\Phi(q', w) \notin F_\Phi$. This implies the minimality of $\mathcal{A}_\Phi$.
	
	So far, we have argued that $\mathcal{A}_\Phi$ is a minimal $\symADFAp$. Now we use \cref{lem:aDFASizeAndLinearity} to argue for the linearity of $\mathcal{A}_\Phi$. With \cref{lem:aDFASizeAndLinearity}~(i), we have $\symLinLenVar \leq \symLinLenVar'$. The words in $J$ witness $\symLinLenVar \geq \symLinLenVar'$. Thus, we have  $\symLinLenVar = \symLinLenVar'$. With \cref{lem:aDFASizeAndLinearity}~(ii), this implies the linearity of $\mathcal{A}_\Phi$. 
	
	Finally, note that $\mathcal{A}_\Phi$ is minimal and has exactly one rejecting state, which is a rejecting sink. Thus, $\mathcal{A}_\Phi$ is a safety DFA.
	
	In total, we have argued that $\mathcal{A}_\Phi$ is an $\symADFAp$, is minimal, is linear and is a safety DFA. In other words, we have argued that $\mathcal{A}_\Phi$ is a minimal linear safety $\symADFAp$. The proof of \cref{lem:aPhiIsMLSADFAp} is complete.
	
	Next, we consider \cref{lem:aPhi}.
	
	We begin with Assertion (i). This assertions follows immediately from Observation (1), which states that $J$ contains exactly the words on which the initial run of $\mathcal{A}_\Phi$ contains exactly the non-sinks of $\mathcal{A}_\Phi$. The last state of any such run is $p_c^s$. Because $\delta_\Phi(p_c^s, c) = p_-$ and $\delta_\Phi(p_c^s, \sigma) = p_+$ for all $\sigma \in \Sigma \setminus \{c\}$, we have $\opDiffWords{\mathcal{A}_\Phi} = \{w c \mid w \in J\} = \{u d c^{\symCVar'} \mid u \in \{0,1\}^r\} = \{u d c^\symCVar \mid u \in \{0,1\}^r\}$.
	For the last equality, we used that, with the linearity of $\mathcal{A}_\Phi$, we have $\symLinLenVar' = \symLinLenVar$ and therefore also $\symCVar' = \symCVar$.
	The proof of Assertion (i) is complete.
	
	To finish, we consider Assertion (ii).
	Let $w \in \opDiffWords{\mathcal{A}_\Phi}$. Towards a proof by contraposition, let $i,j \in \{1, \dots, \symLinLenVar+1\}, i < j$ with $i \neq 1$ or $j \neq r+1$.
	We prove that the \symPPShort-condition does not hold for $w,i,j$.
	That is, we prove that there is an $l \in \symNatNum$ so that $\delta_\Phi(p_0, \opPump{w}{i,j;l}) \neq p_-$. Indeed, we prove that this holds for $l = 0$, so $\delta_\Phi(p_0, \opPump{w}{i,j,0}) \neq p_-$.
	We employ a case distinction over the values $i,j$. For every combination $i,j$, we argue that $\delta_\Phi(p_0, \opPump{w}{i,j;0}) \neq p_-$. 
	Using Assertion (i), let $u \in \{0,1\}^r$ so that $w = u d c^\symCVar$.
	\begin{description}
		\item[Case 1: \normalfont{$i > 1$.}]$ $
		\begin{description}
			\item[Case 1.1: \normalfont{$i > r + 1$.}] Then we remove only $c$'s, meaning that $\opPump{w}{i,j;0}= u d c^x$ with $x \in \symNatNumGeq{1}, x < \symCVar$. Thus, the situation is equivalent to removing a suffix of $w$. Clearly, this means that the initial run of $\mathcal{A}_\Phi$ on $\opPump{w}{i,j;0}$ does not reach $p_-$ but ends in one of the non-sinks. Thus, we have $\delta_\Phi(p_0, \opPump{w}{i,j;0}) \neq p_-$. We are done with Case 1.1.
			
			\item[Case 1.2: \normalfont{$i = r + 1$.}] Then we definitely remove the $d$ and potentially also $c$'s, meaning that $\opPump{w}{i,j;0}= u c^x$ with $x \in \symNatNumGeq{1}, x \leq \symCVar$. With $\delta_\Phi(p_0, u) = p_r$ and $\delta_\Phi(p_r, c) = p_+$, we have $\delta_\Phi(p_0, \opPump{w}{i,j;0}) = p_+ \neq p_-$. We are done with Case 1.2.
			
			\item[Case 1.3: \normalfont{$1 < i < r + 1$.}] Then we remove some but not all $0$'s and $1$'s and potentially also the $d$ and $c$'s, meaning that $\opPump{w}{i,j;0}= v e c^x$ with $v \in \{0,1\}^{r'}, r' \in \symNatNumGeq{1}, r' < r$ and $e \in \{d, \varepsilon\}$ and $x \in \symNatNumGeq{1}, x \leq \symCVar$. With $\delta_\Phi(p_0, v) = p_{r'}$ and $\delta_\Phi(p_{r'}, c) = p_+$ and $\delta_\Phi(p_{r'}, d) = p_+$, we have $\delta_\Phi(p_0, \opPump{w}{i,j;0}) = p_+ \neq p_-$. We are done with Case 1.3.
		\end{description}
		With Cases 1.1–1.3, we have $\delta_\Phi(p_0, \opPump{w}{i,j;0}) \neq p_-$. We are done with Case 1.
		
		\item[Case 2: \normalfont{$i = 1$.}] This implies $j \neq r+1$.
		\begin{description}
			\item[Case 2.1: \normalfont{$j > r + 1$.}] Then we remove all $0$'s and $1$'s as well as the $d$ and potentially also $c$'s, meaning that $w = c^x$ with $x \in \symNatNumGeq{1}, x \leq \symCVar$. With $\delta_\Phi(p_0, c) = p_+ \neq p_-$, we have $\delta_\Phi(p_0, \opPump{w}{i,j;0}) = p_+ \neq p_-$. We are done with Case 2.1.
			
			\item[Case 2.2: \normalfont{$j < r + 1$.}] Then we remove some but not all $0$'s and $1$'s, meaning that $w = v d c^\symCVar$ with $v \in \{0,1\}^{r'}, r' \in \symNatNumGeq{1}, r' < r$. With $\delta_\Phi(p_0, v) = p_{r'}$ and $\delta_\Phi(p_{r'}, d) = p_+$, we have $\delta_\Phi(p_0, \opPump{w}{i,j;0}) = p_+ \neq p_-$. We are done with Case 2.2.
		\end{description}
		With Cases 2.1 and 2.2, we have $\delta_\Phi(p_0, \opPump{w}{i,j;0}) \neq p_-$. We are done with Case 2.
	\end{description}
	With Cases 1 and 2, we have $\delta_\Phi(p_0, \opPump{w}{i,j;0}) \neq p_-$.
	
	We have proven that, for every $w \in \opDiffWords{\mathcal{A}_\Phi}$ and every $i,j \in \{1, \dots, n+1\}, i < j$ with $i \neq 1$ or $j \neq r+1$, we have $\delta_\Phi(p_0, \opPump{w}{i,j;0}) \neq p_-$. The proof of Assertion (ii) is complete.
	
	We have proven \cref{lem:aPhiIsMLSADFAp} as well as \cref{lem:aPhi} (i) and (ii). The proof is complete.
\end{proof}

Next, we restate and then prove \cref{lem:aPhiOnerpOnePumpings}, which establishes the connection between the satisfiability of $\Phi$ and the compositionality of $\mathcal{A}_\Phi$.

\lemaPhiOnerpOnePumpings*

\begin{proof}[Proof of \cref{lem:aPhiOnerpOnePumpings}]
	Let $u \in \{0,1\}^r$ be an assignment string and let $w = udc^\symCVar \in \opDiffWords{\mathcal{A}_\Phi}$ be the corresponding max-visiting word. 
	
	We prove that the \symPPShort-condition does not hold for $w,i=1,j=r+1$ if and only if the assignment $\gamma_u$ induced by $u$ satisfies $\Phi$.
	We prove the two directions of the equivalence separately.
	
	We consider the first direction of the equivalence. Towards a proof by contraposition, we assume that $\gamma_u$ does not satisfy $\Phi$. We prove that the \symPPShort-condition holds for $w,i=1,j=r+1$. Let $l \in \symNatNum$. We prove $\delta_\Phi(p_0, \opPump{w}{1,r+1;l}) = q_-$ by employing a case distinction over $l$.
	
	Before we begin the case distinction, we consider the assignment $\gamma_u$ induced by assignment string $u$. Because $\Phi$ is not satisfiable, there is at least one clause in $\Phi$ that is not satisfied by $\gamma_u$. Let $k$ be the minimum value in $\{1, \dots, s\}$ so that $\gamma_u$ does not satisfy the $k$-th clause $(\symElem_1^k \vee \dots \vee \symElem_r^k)$ of $\Phi$.
	
	Now we begin the case distinction over $l$.
	\begin{description}
		\item[Case 1: \normalfont{$l = 0$.}] Then we have $\opPump{w}{1,r+1;l} = d c^\symCVar$. With $\delta_\Phi(p_0, d) = p_-$, we have $\delta_\Phi(p_0, \opPump{w}{1,r+1;l}) = p_-$. We are done with Case 1.
		
		\item[Case 2: \normalfont{$l = 1$.}] Then we have $\opPump{w}{1,r+1;l} = w$ and therefore trivially $\delta_\Phi(p_0, \opPump{w}{1,r+1;l}) = p_-$. We are done with Case 2.
		
		\item[Case 3: \normalfont{$1 < l < k + 1$.}] Then we have $\opPump{w}{1,r+1;l} = u^l d c^\symCVar$ for $1 < l < k + 1$. It is easy to see that, with \cref{lem:aPhiClauseRows}, we have $\delta_\Phi(p_0,u^l) = \hat{p}_r^{l-1}$. After the first reading of assignment string $u$, the $\symADFAp$ $\mathcal{A}_\Phi$ is in $p_r$. Then $\mathcal{A}_\Phi$ reads $u$ a further $l-1$ times. For each reading, it traverses a clause row of a clause satisfied by $\gamma_u$. With \cref{lem:aPhiClauseRows}, it ends each clause row in the final $\hat{p}$-state of that row. Thus, $\mathcal{A}_\Phi$ is in $\hat{p}_r^{l-1}$ after reading $u^l$, meaning $\delta_\Phi(p_0,u^l) = \hat{p}_r^{l-1}$.
		
		With $\delta_\Phi(p_0,u^l) = \hat{p}_r^{l-1}$ and $\delta_\Phi(\hat{p}_r^{l-1},d) = p_-$, we have $\delta_\Phi(p_0, \opPump{w}{1,r+1;l}) = p_-$. We are done with Case 3.
		
		\item[Case 4: \normalfont{$l \geq k + 1$.}] Then we have $\opPump{w}{1,r+1;l} = u^l d c^\symCVar$ for $l \geq k + 1$. As we argued in Case 3, we have $\delta_\Phi(p_0,u^k) = \hat{p}_r^{k-1}$. Because $\gamma_u$ does not satisfy the $k$-th clause of $\Phi$, we then have $\delta_\Phi(p_0,u^{k+1}) = p_r^k$ with \cref{lem:aPhiClauseRows}. 
		
		Now note that $\delta_\Phi(p_r^k,\sigma) = p_-$ for $\sigma \in \{0,1,d\}$. So both for $l = k+1$ (where the next letter is $d$) and $l > k+1$ (where the next letter is $0$ or $1$), the $\symADFAp$ $\mathcal{A}_\Phi$ enters $p_-$ for the first letter after $u^{k+1}$. Formally, this means that, with $\delta_\Phi(p_0,u^{k+1}) = p_r^k$ and $\delta_\Phi(p_r^k,\sigma) = p_-$ for $\sigma \in \{0,1,d\}$, we have $\delta_\Phi(p_0, \opPump{w}{1,r+1;l}) = p_-$. We are done with Case 4.
	\end{description}
	With Cases 1–4, we have $\delta_\Phi(p_0, \opPump{w}{1,r+1;l}) = q_-$.
	
	We have assumed that $\gamma_u$ does not satisfy $\Phi$. Under this assumption, we have proven that the \symPPShort-condition holds for $w, i=1, j=r+1$. We have proven the first direction of the equivalence.
	
	We consider the second direction of the equivalence. We assume that $\gamma_u$ satisfies $\Phi$. We prove that the \symPPShort-condition does not hold for $w,i=1,j=r+1$. To be precise, we prove that $\delta_\Phi(p_0, \opPump{w}{1,r+1;s+2}) = q_+ \neq q_-$. Thus, we prove that $l = s+2$ witnesses that the \symPPShort-condition does not hold for $w,i=1,j=r+1$. Indeed, every $l \geq s+2$ is a suitable witness.
	
	Because $\gamma_u$ satisfies $\Phi$, every clause of $\Phi$ is satisfied by $\gamma_u$. Thus, with the same argumentation that we employed for Case 3 in the proof of the first direction of the equivalence, we have $\delta_\Phi(p_0, u) = p_r$ and $\delta_\Phi(p_0, u^{l'}) = \hat{p}_r^{l'-1}$ for all $l' \in \symNatNum, 1 < l' < s+2$. In particular, we have $\delta_\Phi(p_0, u^{s+1}) = \hat{p}_r^s$. Now note that $\delta_\Phi(\hat{p}_r^s, \sigma) = p_+$ for all $\sigma \in \{0,1\}$. Together with $\delta_\Phi(p_0, u^{s+1}) = \hat{p}_r^s$, this implies $\delta_\Phi(p_0, u^{s+2}) = p_+$ and therefore $\delta_\Phi(p_0, \opPump{w}{1,r+1;s+2}) = p_+ \neq p_-$. Thus, we have that $l = s+2$, and indeed every $l \geq s+2$, witnesses that the \symPPShort-condition does not hold for $w,i=1,j=r+1$.
	
	We have assumed that $\gamma_u$ satisfies $\Phi$. Under this assumption, we have proven that the \symPPShort-condition does not hold for $w, i=1, j=r+1$. We have proven the second direction of the equivalence.
	
	We have proven both directions of the equivalence. The proof is complete.
\end{proof}

With \cref{the:pmADFAMinLinSafeCompositionality,lem:aPhiIsMLSADFAp,lem:aPhi,lem:aPhiOnerpOnePumpings}, we have everything we need to prove \cref{lem:aPhiCharacterization}.
For the formal proof of \cref{lem:aPhiCharacterization}, we refer to \cref{subsubsec:characterizationToNPHardness}.

We have proven the correctness of our reduction.
For completeness sake, we now state the very short formal proof of \cref{the:primeDFANPHard}:

\begin{proof}[Proof of \cref{the:primeDFANPHard}]
	We polynomially reduce \probCNFSAT{} to \probPrimeDFA{}. We use our \symCNFDFA-construction.
	
	For any CNF-formula $\Phi$, our construction, outlined in \cref{fig:aPhiaPhiXTransitions} and formally defined in \cref{def:aPhi}, yields a \symCNFDFA{} $\mathcal{A}_\Phi$. Clearly, the construction is possible in polytime. And with \cref{lem:aPhiCharacterization}, $\Phi$ is satisfiable if and only if $\mathcal{A}_\Phi$ is prime. Thus, the construction can indeed serve as a polynomial reduction from \probCNFSAT{} to \probPrimeDFA{}. The proof is complete.
\end{proof}

Here in \cref{sec:npHardness}, we have offered proofs for every result presented but not already proven in \cref{sec:npHardnessShortV2}. With this, we have formally proven the \fontComplexityClass{NP}-hardness of \probPrimeDFA{}.

\subsection{\fontComplexityClass{NP}-Completeness of $\probPrimeADFApMLS{}$ – Proofs}
\label{subsec:aDFApMLSNPCompleteProofs}
In \cref{subsec:constructionAPhiProofs,subsec:correctnessAPhiProofs}, we gave the formal proofs for the improvement of the lower complexity bound of the general problem \probPrimeDFA{}. Here, we formally prove that this improved lower bound is tight for the restriction of \probPrimeDFA{} to minimal linear safety $\symADFAp$s.

As mentioned in \cref{subsec:aDFApMLSNPComplete}, we denote the restriction of \probPrimeDFA{} to DFAs whose respective minimal DFA is a minimal linear safety $\symADFAp$ by $\probPrimeADFApMLS{}$.
We prove:

\theprimeADFApmLinSafeNPComplete*

As argued in \cref{subsec:aDFApMLSNPComplete}, from the proof of \cref{the:primeDFANPHard}, we immediately get:

\lemprimeADFApmLinSafeNPHard*

To arrive at \cref{the:primeADFApmLinSafeNPComplete}, we now prove:

\lemprimeADFApmLinSafeInNP*

\begin{proof}[Proof of \cref{lem:primeADFApmLinSafeInNP}]
	To prove that $\probPrimeADFApMLS{}$ is in \fontComplexityClass{NP}, we describe a guess-and-verify \fontComplexityClass{NP}-algorithm which exploits \cref{the:pmADFAMinLinSafeCompositionality}.
	The basic idea of the algorithm is to guess a word and then to verify that this words breaks the \symPPShort, meaning that the word witnesses that the respective DFA does not have the \symPPShort.
	
	We begin by describing the guess-and-verify algorithm. Then we argue that it is indeed an \fontComplexityClass{NP}-algorithm for $\probPrimeADFApMLS{}$ by arguing that it works in nondeterministic polytime and that it works correctly for $\probPrimeADFApMLS{}$.
	
	It is well known that DFA minimization is possible in polytime. Therefore, for every given DFA whose respective minimal DFA is a minimal linear safety $\symADFAp$, it is possible to calculate its minimal DFA in polytime in a precomputation step before the actual guess-and-verify algorithm begins. Thus, it is sufficient to describe a guess-and-verify algorithm that decides primality of minimal linear safety $\symADFAp$s in nondeterministic polytime.
	
	Let $\mathcal{A} = (Q, \Sigma, q_I, \delta, F)$ be a minimal linear safety $\symADFAp$.
	Let $q_+, q_-$ denote the accepting and rejecting sink. The guess-and-verify algorithm works in the following manner:
	\begin{enumerate}
		\item
		The algorithm calculates $\opLinLen{\mathcal{A}}$.
		
		Because $\mathcal{A}$ is a minimal linear $\symADFAp$, we have $\opLinLen{\mathcal{A}} = \opSize{\mathcal{A}} - 3$ with \cref{lem:aDFASizeAndLinearity}. Thus, the algorithm can calculate $\opLinLen{\mathcal{A}}$ by inspecting the size of $\mathcal{A}$. From here on, let $\symLinLenVar = \opLinLen{\mathcal{A}}$.
		\item
		If $\mathcal{A}$ is trivial, that is, if $\symLinLenVar = 0$, then the algorithm accepts.
		\item
		The algorithm tries to guess a word from $\opDiffWords{\mathcal{A}}$.
		
		To do so, it guesses a word $w = \sigma_1 \dots \sigma_{\symLinLenVar+1} \in \Sigma^{\symLinLenVar+1}$. It then simulates the initial run of $\mathcal{A}$ on $w$. It checks whether $\delta(q_I, \sigma_1 \dots \sigma_\symLinLenVar) \notin \{q_+,q_-\}$ and $\delta(q_I, w) = q_-$ hold. If both hold, then we have $w \in \opDiffWords{\mathcal{A}}$ and the algorithm continues. Else, we have $w \notin \opDiffWords{\mathcal{A}}$ and the algorithm rejects.
		\item
		The algorithm verifies that $w$ breaks the \symPPShort.
		
		To do so, the algorithm considers every pair $i,j \in \{1, \dots, \symLinLenVar+1\}, i < j$. For every such pair, the algorithm first calculates the minimum value $l_{i,j}' \in \symNatNum$ so that $|\sigma_1 \dots \sigma_{i-1} (\sigma_i \dots \sigma_{j-1})^{l_{i,j}'}| = (i - 1) + l_{i,j}' (j - i) \geq \symLinLenVar + 1$ holds.
		Then it checks whether $\delta(q_I, \opPump{w}{i,j;l}) \neq q_-$ holds for some $l \in \{0, \dots, l_{i,j}'\}$ by simulating the runs of $\mathcal{A}$ on the respective words.
		
		If the algorithm finds a pair $i,j$ for which $\delta(q_I, \opPump{w}{i,j;l}) = q_-$ holds for all $l \in \{0, \dots, l_{i,j}'\}$, it rejects.
		Otherwise, that is, if, for every pair $i,j$, there is an $l \in \{0, \dots, l_{i,j}'\}$ so that $\delta(q_I, \opPump{w}{i,j;l}) \neq q_-$ holds, it accepts.
	\end{enumerate}
	
	We have described the guess-and-verify algorithm. Now we argue that this algorithm is indeed an \fontComplexityClass{NP}-algorithm for $\probPrimeADFApMLS{}$.
	
	We start by arguing that the algorithm works in nondeterministic polytime on a given minimal linear safety $\symADFAp$ $\mathcal{A}$.
	
	In Step 1, the algorithm counts the states of $\mathcal{A}$ and subtracts the value three. Clearly, this can be done in polytime.
	In Step 2, the algorithm checks whether the value calculated in Step 1 is equal to zero. This, too, can be done in polytime.
	In Step 3, the algorithm guesses a word of length $\symLinLenVar + 1 = \opBigO{\opSize{\mathcal{A}}}$ and simulates the initial run of $\mathcal{A}$ on this word. Obviously, this can be done in nondeterministic polytime.
	In Step 4, the algorithm inspects all pairs $i,j \in \{1, \dots, \symLinLenVar+1\}, i < j$. So it inspects $\opBigO{\opSize{\mathcal{A}}^2}$ pairs. For every pair, the algorithm calculates $l_{i,j}'$. Obviously, even calculating $l_{i,j}'$ naively by counting $l$ up from $0$ until an $l$ with $(i-1) + l (j-i) \geq \symLinLenVar+1$ is reached can be done in polytime. Thus, calculating $l_{i,j}'$ can be done in polytime. With $l_{i,j}'$ in hand, the algorithm simulates the run of $\mathcal{A}$ on $\opPump{w}{i,j;l}$ for every $l \in \{0, \dots, l_{i,j}'\}$. Note that $l_{i,j}' = \opBigO{\opSize{\mathcal{A}}}$ and $|\opPump{w}{i,j;l}| = \opBigO{\opSize{\mathcal{A}}}$ for every $l \in \{0, \dots, l_{i,j}'\}$, meaning that the algorithm simulates $\opBigO{\opSize{\mathcal{A}}}$ runs on words of length $\opBigO{\opSize{\mathcal{A}}}$ for every pair $i,j$. Obviously, this can be done in polytime. Thus, the inspection of each of the polynomially many pairs $i,j$ can be done in polytime.
	
	In total, Steps 1, 2 and 4 can be done in polytime, and Step 3 can be done in nondeterministic polytime. Thus, the algorithm works in nondeterministic polytime.
	
	To conclude the proof, we argue that the algorithm works correctly for $\probPrimeADFApMLS{}$. To do so, let $\mathcal{A}$ be a minimal linear safety $\symADFAp$ with sinks $q_+, q_-$ which serves as an input for the algorithm. Let $\symLinLenVar = \opLinLen{\mathcal{A}}$.
	
	First, we argue that the algorithm accepts if $\mathcal{A}$ is prime. To do so, we assume that $\mathcal{A}$ is prime.
	
	If $\mathcal{A}$ is trivial, that is, if $\symLinLenVar = 0$, then the algorithm calculates $\symLinLenVar$ correctly in Step 1 and accepts in Step 2. Thus, the algorithm accepts in Step 2 if $\mathcal{A}$ is prime and trivial.
	
	We assume that $\mathcal{A}$ is non-trivial. With \cref{the:pmADFAMinLinSafeCompositionality}, $\mathcal{A}$ does not have the \symPPShort. Let $w = \sigma_1 \dots \sigma_{\symLinLenVar+1} \in \opDiffWords{\mathcal{A}}$ be a word which breaks the \symPPShort.
	
	In Step 1, the algorithm calculates $\symLinLenVar$ correctly.
	In Step 2, the algorithm checks whether $\mathcal{A}$ is trivial by inspecting the calculated value. Because $\mathcal{A}$ is by assumption non-trivial, we have $\symLinLenVar > 0$, so the algorithm does not accept.
	In Step 3, the algorithm can guess $w$. Doing so, it does not reject in Step 3 and moves on to Step 4.
	In Step 4, the algorithm inspects every pair $i,j$ and checks whether $\delta(q_I, \opPump{w}{i,j;l}) \neq q_-$ holds for some $l \in \{0,\dots,l_{i,j}'\}$.
	
	Consider some pair $i,j$.
	Note that, with $|\sigma_1 \dots \sigma_{i-1} (\sigma_i \dots \sigma_{j-1})^{l_{i,j}'}| \geq \symLinLenVar + 1$, we have $\delta(q_I, \sigma_1 \dots \sigma_{i-1} (\sigma_i \dots \sigma_{j-1})^{l_{i,j}'}) \in \{q_+, q_-\}$ and therefore $\delta(q_I, \opPump{w}{i,j;l}) = \delta(q_I, \opPump{w}{i,j;l_{i,j}'})$ for every $l \in \symNatNum, l \geq l_{i,j}'$.
	This, together with $w$ breaking the \symPPShort{}, implies that there is an $l \in \{0,\dots,l_{i,j}'\}$ with $\delta(q_I, \opPump{w}{i,j;l}) \neq q_-$.
	Finding such an $l \in \{0,\dots,l_{i,j}'\}$ with $\delta(q_I, \opPump{w}{i,j;l}) \neq q_-$, the algorithm does not reject for this pair $i,j$. In total, the algorithm does not reject for any pair $i,j$ and accepts after inspecting every pair. Therefore, the algorithm accepts in Step 4 if $\mathcal{A}$ is prime and non-trivial.
	
	Second, we argue that the algorithm rejects if $\mathcal{A}$ is composite. To do so, we assume that $\mathcal{A}$ is composite. 
	This implies that $\mathcal{A}$ is non-trivial, meaning $\symLinLenVar > 0$. Further, it implies, with \cref{the:pmADFAMinLinSafeCompositionality}, that $\mathcal{A}$ has the \symPPShort.
	This means that, for every word $w \in \opDiffWords{\mathcal{A}}$, there exist $i,j \in \{1, \dots, \symLinLenVar+1\}, i < j$ so that the \symPPShort-condition holds for $w,i,j$.
	
	In Step 1, the algorithm calculates $\symLinLenVar$ correctly.
	In Step 2, the algorithm checks whether $\mathcal{A}$ is trivial by inspecting the calculated value. 
	As mentioned, $\mathcal{A}$ is non-trivial, so we have $\symLinLenVar > 0$ and the algorithm does not accept in Step 2 and moves on to Step 3.
	In Step 3, the algorithm guesses some word $w$ and checks if $w \in \opDiffWords{\mathcal{A}}$ holds. If $w \notin \opDiffWords{\mathcal{A}}$, the algorithm rejects. Otherwise, that is, if $w \in \opDiffWords{\mathcal{A}}$, the algorithm moves on to Step 4.
	In Step 4, the algorithm checks whether $w$ breaks the \symPPShort. However, because $\mathcal{A}$ has the \symPPShort{} and $w \in \opDiffWords{\mathcal{A}}$, there exist $i,j \in \{1, \dots, \symLinLenVar+1\}, i < j$ such that the \symPPShort-condition holds for $w,i,j$. The algorithm rejects for the first such pair it considers while inspecting all pairs $i,j$. Thus, the algorithm rejects in Step 3 or in Step 4 if $\mathcal{A}$ is composite.
	
	In total, the algorithm accepts if $\mathcal{A}$ is prime and rejects if $\mathcal{A}$ is composite. Thus, the algorithm works correctly for $\probPrimeADFApMLS{}$.
	
	We have described a guess-and-verify algorithm.
	We have proven that it works in nondeterministic polytime.
	Further, we have proven that it decides primality for minimal linear safety $\symADFAp$s.
	As argued, because DFA minimization can happen in a precomputation step before the actual guess-and-verify algorithm, we have thereby proven that $\probPrimeADFApMLS{}$ is in \fontComplexityClass{NP}. The proof is complete.
\end{proof}

Together, \cref{lem:primeADFApmLinSafeNPHard,lem:primeADFApmLinSafeInNP} immediately imply \cref{the:primeADFApmLinSafeNPComplete}, the \fontComplexityClass{NP}-completeness of $\probPrimeADFApMLS{}$.